\documentclass[12pt]{article}

\usepackage{amsmath}       
\usepackage{amsthm}        
\usepackage{amssymb}       
\usepackage{amsfonts}
\usepackage[all]{xy}
\usepackage{xspace}
\usepackage{calc}
\usepackage{comment}
\usepackage{pgfplots}
\usepgflibrary{fpu}
\usepackage{mathtools}
\usepackage{tabularx}
\usepackage{ifthen}             
\usepackage{graphicx}
\usepackage{docmute}
\usepackage{array}
\usepackage{subfloat} 
\usepackage{textcomp}
\usepackage{bbm}            
\usepackage{etoolbox}  
\usepackage{verbatim}
\usepackage{stmaryrd}
\usepackage{cancel}
\usepackage{xcolor}
\colorlet{Black}{black}
\usepackage{tensor}
\usepackage{enumerate}  
\usepackage{thm-restate}
\usepackage[numbers,sort]{natbib}
\usepackage{enumitem}
\usepackage{tocloft}
\usepackage{pdfpages}
\usepackage[export]{adjustbox}

\usepackage[margin=3cm]{geometry}

\newenvironment{tz}[1][]{%
                                \begin{tikzpicture}[baseline={([yshift=-.8ex]current bounding                        box.center)},#1] %
                                }{%
                        \end{tikzpicture} %
                        }

\DeclareRobustCommand{\SkipTocEntry}[5]{}

\usepackage{tikz}
\usetikzlibrary{matrix}
\usetikzlibrary{decorations.pathreplacing}
\usetikzlibrary{decorations.markings}
\usetikzlibrary{arrows}
\usetikzlibrary{calc}
\usetikzlibrary{shapes.misc}
\usetikzlibrary{fit}
\usepgflibrary{decorations.pathmorphing}
\usepgflibrary{shapes.geometric}

\pgfdeclarelayer{edgelayer}
\pgfdeclarelayer{nodelayer}
\pgfdeclarelayer{foreground}
\pgfsetlayers{edgelayer,nodelayer,main,foreground}

\tikzstyle{none}=[inner sep=0pt]

\tikzstyle{rn}=[circle,fill=Red,draw=Black,line width=0.8 pt]
\tikzstyle{gn}=[circle,fill=Lime,draw=Black,line width=0.8 pt]
\tikzstyle{bl}=[circle,fill=Blue,draw=Black,line width=0.8 pt]

\tikzstyle{simple}=[-,draw=Black,thick]
\tikzstyle{arrow}=[-,draw=Black,postaction={decorate},decoration={markings,mark=at position .5 with {\arrow{>}}},thick]
\tikzstyle{tick}=[-,draw=Black,postaction={decorate},decoration={markings,mark=at position .5 with {\draw (0,-0.1) -- (0,0.1);}},line width=2.000]
\makeatother\def\thickness{0.7pt}
\tikzstyle{dot}=[circle, draw=black, fill=black, inner sep=.5ex, line width=\thickness, node on layer=foreground]

\makeatletter
\pgfkeys{%
  /tikz/on layer/.code={
    \pgfonlayer{#1}\begingroup
    \aftergroup\endpgfonlayer
    \aftergroup\endgroup
  },
  /tikz/node on layer/.code={
     \gdef\node@@on@layer{%
      \setbox\tikz@tempbox=\hbox\bgroup\pgfonlayer{#1}\unhbox\tikz@tempbox\endpgfonlayer\egroup}
    \aftergroup\node@on@layer
  },
  /tikz/end node on layer/.code={
    \endpgfonlayer\endgroup\endgroup
  }
}

\def\node@on@layer{\aftergroup\node@@on@layer}

\makeatletter
\def\calign@preamble{%
   &\hfil\strut@
    \setboxz@h{\@lign$\m@th\displaystyle{##}$}%
    \ifmeasuring@\savefieldlength@\fi
    \set@field
    \hfil
    \tabskip\alignsep@
}
\let\cmeasure@\measure@
\patchcmd\cmeasure@{\divide\@tempcntb\tw@}{}{}{}
\patchcmd\cmeasure@{\divide\@tempcntb\tw@}{}{}{}
\patchcmd\cmeasure@{\ifodd\maxfields@
  \global\advance\maxfields@\@ne
  \fi}{}{}{}    
\newenvironment{calign}
{%
  \let\align@preamble\calign@preamble
  \let\measure@\cmeasure@
  \align
}
{%
  \endalign
}  
\makeatother
\tikzset{
    master/.style={
        execute at end picture={
            \coordinate (lower right) at (current bounding box.south east);
            \coordinate (upper left) at (current bounding box.north west);
        }
    },
    slave/.style={
        execute at end picture={
            \pgfresetboundingbox
            \path (upper left) rectangle (lower right);
        }
    }
}

\tikzset{blob/.style={draw, circle, fill=white, inner sep=1pt, minimum width=15pt, font=\scriptsize, line width=0.7pt}}
\tikzset{greenregion/.style={fill=green, fill opacity=0.3, draw=none}}
\tikzset{redregion/.style={fill=red, fill opacity=0.3, draw=none}}
\tikzset{blueregion/.style={fill=blue, fill opacity=0.3, draw=none}}
\tikzset{yellowregion/.style={fill=yellow, fill opacity=0.5, draw=none}}
\tikzset{cyanregion/.style={fill=cyan, fill opacity=0.3, draw=none}}
\tikzset{orangeregion/.style={fill=orange, fill opacity=0.6, draw=none}}
\tikzset{solidgreenregion/.style={fill=green!30, fill opacity=1, draw=none}}
\tikzset{solidredregion/.style={fill=red!30, fill opacity=1, draw=none}}
\tikzset{solidblueregion/.style={fill=blue!30, fill opacity=1, draw=none}}
\tikzset{solidyellowregion/.style={fill=yellow!30, fill opacity=1, draw=none}}
\tikzset{string/.style={line width=0.7pt}}
\tikzset{zig/.style={decoration={zigzag,segment length=3, amplitude=0.5}}}
\tikzset{bnd/.style={draw,string}}   
\tikzset{projector/.style={circle, draw, font=\scriptsize, inner sep=-5pt, minimum width=0.35cm, string, fill=white}}
\tikzset{dimension/.style={font=\scriptsize, inner sep=1pt}}

\tikzset{arrow data/.style 2 args={
      decoration={
         markings,
         mark=at position #1 with \arrow{#2}},
         postaction=decorate}
}

\tikzset{along path/.style={every path/.style={}, sloped, allow upside down}}

\def\zxnormal {
                \def \zxscale{0.55}
                \def\zxnodescale{0.8}
                \def\vertexscale{0.7}
                \def\zxshift{0.075cm}
                \def\hadscale{0.8}
                \def\trianglescale{1}
                \def\boxscale{1}
                }

\def\zxgreen{white}

\def\zxwhite{white}

\tikzset{front/.style ={node on layer=foreground}}
\tikzset{zx/.style = {string, scale=\zxscale}}
\tikzset{zxnode/.style n args={1}{blob,scale=\zxnodescale,fill=#1,node on layer=foreground}}
\tikzset{box/.style={draw, rectangle, fill=white, inner sep=1pt, minimum width=10pt,minimum height=10pt, font=\scriptsize, line width=0.7pt,scale=\zxnodescale,node on layer=foreground}}
\tikzset{boxvertex/.style={draw, rectangle, fill=white, line width=0.733pt,scale=0.75*\vertexscale}}
\tikzset{bigbox/.style={draw, rectangle, fill=white,  minimum width=\boxscale *18pt,minimum height=\boxscale*8pt, line width=0.7pt,scale=\zxnodescale}}
\newlength{\unitbox}
\newcommand\boxwidth[1] { \setlength{\unitbox}{0.2cm + #1cm}}
\tikzset{widebox/.style ={draw,rectangle, fill=white, line width=0.7pt,scale=0.75*\zxnodescale,minimum height=15pt,inner sep=1pt,  minimum width = \unitbox,   anchor=center }}
\tikzset{wideboxm/.style n args={1}{draw,rectangle, fill=white, line width=0.7pt,scale=0.75*\zxnodescale,minimum height=15pt,inner sep=1pt,  minimum width =2\unitbox+#1\unitbox,   anchor=center }}
\tikzset{triangleup/.style n args={1}{draw, shape=isosceles triangle, isosceles triangle stretches, fill=white, line width=0.7pt,scale=0.75*\zxnodescale,minimum height=15pt,inner sep=1pt,  minimum width = #1*\trianglescale cm +0.15*\trianglescale cm,  shape border rotate=90, anchor=south }}
\tikzset{triangledown/.style n args={1}{draw, shape=isosceles triangle, isosceles triangle stretches, fill=white, line width=0.7pt,scale=0.75*\zxnodescale,minimum height=15pt,inner sep=1pt,  minimum width = #1*\trianglescale cm +0.15*\trianglescale cm,  shape border rotate=-90, anchor=north }}
\tikzset{zxvertex/.style n args={1}{draw,fill=#1,circle,line width=0.7pt,scale=0.75*\vertexscale}}
\tikzset{zxdown/.style={yshift=-\zxshift}}
\tikzset{zxup/.style={yshift=\zxshift}}

\newcommand\mult[3]{ 
\draw[string] (#1.center) to [out=up, in=-135] +(0.5*#2,#3) to [out=-45, in=up] +(0.5*#2,-#3);
\node[zxvertex=\zxgreen,zxdown] at ($(#1)+(0.5*#2,#3)$){};
}
\newcommand\comult[3]{ 
\draw[string] (#1.center) to [out=down, in=135] +(0.5*#2,-#3) to [out=45, in=down] +(0.5*#2,#3);
\node[zxvertex=\zxgreen,zxup] at ($(#1) +(0.5*#2,-#3)$){};}

\newcommand\unit[2]{ 
\draw[string] (#1.center) to + (0, -#2);
\node[zxvertex=\zxgreen] at ($(#1) +(0,-#2)$){};
}

\newcommand\emptydiagram{
\begin{tz}[zx]
\draw[dotted] (-0.5,-0.5) rectangle (0.5,0.5);
\end{tz}
}

\newcommand\C{\ensuremath{\mathbb{C}}\xspace}

\newcommand\superequals[1]{\stackrel {\makebox[0pt]{\tiny #1}} =}
\newcommand\super[2]{\stackrel{\makebox[0pt]{\tiny #1}} #2}
\newcommand\superequalseq[1]{\stackrel {\makebox[0pt]{\tiny\eqref{#1}}} =}
\newcommand{\ket}[1]{\left|#1\right\rangle}
\newcommand{\bra}[1]{\left\langle#1\right|}

\newcommand{\ketbra}[1]{{\small\ket{#1}\!\bra{#1}}}

\renewcommand{\to}[1][]{\ensuremath{\xrightarrow{#1}}}

\allowdisplaybreaks[1]

\theoremstyle{plain} 

\newtheorem{theorem}{Theorem}[section]

\newtheorem{corollary}[theorem]{Corollary}          
\newtheorem{proposition}[theorem]{Proposition}              
          
\newtheorem{prop}[theorem]{Proposition}

\theoremstyle{definition} 
\newtheorem{definition}[theorem]{Definition}
\newtheorem{remark}[theorem]{Remark}

\newtheorem{philosophy}[theorem]{Philosophy}
\newtheorem{terminology}[theorem]{Terminology}

\theoremstyle{remark}  
\newtheorem{example}[theorem]{Example}

\newtheoremstyle{special_statement} 
        {\topskip}
        {\topskip}
        {\addtolength{\leftskip}{2.5em} \itshape }
        {}
        {\bfseries}
        {:}
        {.5em}
        {}
\theoremstyle{special_statement}

\DeclareMathOperator{\Hom}{Hom}
\DeclareMathOperator{\End}{End}

\newcommand{\id}{\mathrm{id}}

\newcommand{\Aut}{\ensuremath{\mathrm{Aut}}}

\newcommand{\Cat}{\mathrm{Cat}}

\newcommand{\Hilb}{\ensuremath{\mathrm{Hilb}}}
\newcommand{\Set}{\ensuremath{\mathrm{Set}}}
\newcommand{\deloop}{\ensuremath{\mathbf{B}}}

\newcommand\conj[1]{\overline{#1}}

\newcommand{\op}{\text{op}}
\newcommand{\BHilb}{\ensuremath{\deloop\Hilb}}
\newcommand{\CAlg}{\ensuremath{C^*\mathrm{Alg}}}
\newcommand{\QSet}{\ensuremath{\mathrm{QSet}}}
\newcommand{\QEl}{\ensuremath{\mathrm{QEl}}}
\newcommand{\QBij}{\ensuremath{\mathrm{QBij}}}

\newcommand{\QGraph}{\ensuremath{\mathrm{QGraph}}}

\newcommand{\QIso}{\ensuremath{\mathrm{QIso}}}
\newcommand{\CF}{\ensuremath{\mathrm{SCFA}}}
\newcommand{\F}{\ensuremath{\mathrm{SSFA}}}

\makeatletter
\DeclareFontFamily{OMX}{MnSymbolE}{}
\DeclareSymbolFont{MnLargeSymbols}{OMX}{MnSymbolE}{m}{n}
\SetSymbolFont{MnLargeSymbols}{bold}{OMX}{MnSymbolE}{b}{n}
\DeclareFontShape{OMX}{MnSymbolE}{m}{n}{
    <-6>  MnSymbolE5
   <6-7>  MnSymbolE6
   <7-8>  MnSymbolE7
   <8-9>  MnSymbolE8
   <9-10> MnSymbolE9
  <10-12> MnSymbolE10
  <12->   MnSymbolE12
}{}
\DeclareFontShape{OMX}{MnSymbolE}{b}{n}{
    <-6>  MnSymbolE-Bold5
   <6-7>  MnSymbolE-Bold6
   <7-8>  MnSymbolE-Bold7
   <8-9>  MnSymbolE-Bold8
   <9-10> MnSymbolE-Bold9
  <10-12> MnSymbolE-Bold10
  <12->   MnSymbolE-Bold12
}{}

\let\llangle\@undefined
\let\rrangle\@undefined
\DeclareMathDelimiter{\llangle}{\mathopen}%
                     {MnLargeSymbols}{'164}{MnLargeSymbols}{'164}
\DeclareMathDelimiter{\rrangle}{\mathclose}%
                     {MnLargeSymbols}{'171}{MnLargeSymbols}{'171}
\makeatother

\usepackage[colorlinks=true, linkcolor=black, citecolor=black, urlcolor=black, hypertexnames=false
        ]{hyperref}

\newcounter{DRcomment}
\newcommand\DR[1]{\ensuremath{{}^{\color{red}\theDRcomment}}\marginpar{\color{red}\tiny\raggedright \theDRcomment: #1}\stepcounter{DRcomment}}
\newcommand\DRcomm[1]{{\color{red}#1}}

\newcounter{DVcomment}
\newcommand\DV[1]{\ensuremath{{}^{\color{green}\theDVcomment}}\marginpar{\color{green}\tiny\raggedright \theDVcomment: #1}\stepcounter{DVcomment}}

\newcounter{BMcomment}
\newcounter{JVcomment}
\newcommand\JV[1]{\ensuremath{{}^{\color{black}\theJVcomment}}\marginpar{\color{black}\tiny\raggedright \theJVcomment: #1}\stepcounter{JVcomment}}

\newcommand\ignore[1]{}

\tikzstyle{blackdot}=[circle, draw=black, fill=black, inner sep=.5ex, line width=\thickness, node on layer=foreground]
\tikzstyle{whitedot}=[circle, draw=black, fill=white, inner sep=.5ex, line width=\thickness, node on layer=foreground]

\tikzset{proofdiagram/.style={scale=1}}
\newlength\morphismheight
\newlength\minimummorphismwidth
\newlength\stateheight
\title{A compositional approach to quantum functions}
\author{\small\hspace{-1cm}\begin{tabular}{c c c}
Benjamin Musto &\, David Reutter &\, Dominic Verdon\\
\texttt{benjamin.musto@cs.ox.ac.uk} &\; \texttt{david.reutter@cs.ox.ac.uk} &\;\texttt{dominic.verdon@cs.ox.ac.uk} \end{tabular}\\[20pt]
Department of Computer Science, University of Oxford}
\date{\today}

\begin{document}

\normalsize
\zxnormal
\maketitle      

\setcounter{tocdepth}{2}

\begin{abstract}

We introduce a notion of quantum function, and develop a compositional framework for finite quantum set theory based on a $2$-category of quantum sets and quantum functions.
We use this framework to formulate a $2$-categorical theory of quantum graphs, which captures the quantum graphs and quantum graph homomorphisms recently discovered in the study of nonlocal games and zero-error communication, and relates them to quantum automorphism groups of graphs considered in the setting of compact quantum groups. 
We show that the $2$-categories of quantum sets and quantum graphs are semisimple. We analyse dualisable and invertible $1$-morphisms in these $2$-categories and show that they correspond precisely to the existing notions of quantum isomorphism and classical isomorphism between sets and graphs. 
\end{abstract}
\section{Introduction}
 \paragraph{On quantum functions.}
In this work, we introduce a notion of \emph{quantum function} between finite \emph{quantum sets} and propose a framework for finite quantum set theory. \ignore{This can be used as a foundation for quantum combinatorial theories; here, we specifically investigate quantum graph theory.

 such as quantum graph theory which we also investigate here. 
\ignore{ in this work, we also investigate quantum graph theory. }}
%
These quantum functions underlie various notions of `quantum morphism' previously defined in quantum information theory and noncommutative topology.\footnote{An extensive discussion of related work can be found in Section~\ref{sec:relatedwork}.}

Several such notions\ignore{quantum variants of classical combinatorial structures} have recently emerged from the study of \emph{quantum pseudo-telepathy}~\cite{Brassard2005}, a phenomenon in quantum theory where pre-shared entanglement is used to perform a task classically impossible without communication. Such tasks are usually formulated as games,     where winning classical strategies correspond to certain homomorphisms between combinatorial structures; quantum strategies can then be understood as quantised versions of these homomorphisms. One such game is the \emph{graph homomorphism game}~\cite{Mancinska2016}, leading to a notion of \emph{quantum graph homomorphism}. 
We show that quantum graph homomorphisms are quantum functions between vertex sets preserving the graph structure; in fact, quantum functions can themselves be understood as perfect quantum strategies for a certain `function game'. 
The theory of quantum functions places this game-theoretic approach to quantisation into a broader mathematical context, in particular relating it to more conventional approaches to quantisation.\looseness=-2

Indeed, since the advent of quantum mechanics, quantisation has been associated with the passage from commuting to noncommuting variables. In the spirit of noncommutative topology\footnote{Here, noncommutative topology refers to the study of $C^*$-algebras in light of Gelfand duality. The term noncommutative geometry is usually reserved for the study of spectral triples (see e.g.~\cite{Connes2000}), a noncommutative version of the theory of Riemannian manifolds.}, we regard $C^*$-algebras as noncommutative analogues of topological spaces, or \emph{quantum spaces}, with finite-dimensional $C^*$-algebras playing the role of \emph{finite quantum sets}.
 Finding an appropriate definition of a quantum space of quantum functions between finite sets is a more subtle issue. Conventionally, these quantum spaces are defined as $C^*$-algebras satisfying a universal property. A prominent example is Wang's \emph{quantum permutation group}~\cite{Wang1998} of a set, defined in the framework of compact quantum groups~\cite{Woronowicz1998} and later extended to the \emph{quantum automorphism group}~\cite{Banica2005} of a graph.
Our framework captures the finite-dimensional representation theory of these $C^*$-algebras and extends the existing theory of quantum permutation and automorphism groups to a theory of general morphisms between different sets and graphs. In particular, it relates quantum automorphism groups to the \emph{quantum graph isomorphisms}~\cite{Atserias2016} considered in the study of pseudo-telepathy, providing a concrete link between results in noncommutative topology and recent problems in quantum information theory.  

Our theory of quantum functions --- unifying and extending the above approaches to the quantisation of combinatorial structures --- is formulated in categorical terms and allows us to apply techniques from categorical algebra, such as the notions of semisimplicity and dualisability. The following is a brief summary of our main results:
 
\ignore{We define quantum elements, quantum functions and quantum bijections, as well as quantum graph homomorphisms and quantum graph isomorphisms (Definitions \ref{def:quantumelement},~\ref{def:quantumfunction},~\ref{def:quantumbijection},~\ref{def:quantumgraphhom} and~\ref{def:quantumgraphiso}).} 
\begin{itemize}
\item We show that quantum sets and quantum functions naturally form a 2-category $\QSet$ (Definition~\ref{def:2catqset}) `quantising' the category of finite sets and functions. Similarly, quantum graphs (Definition~\ref{def:quantumgraphsbyadjmats}, see also~\cite{Duan2013,Weaver2015}) and quantum graph homomorphisms naturally form a 2-category $\QGraph$ (Definition~\ref{def:2catqgraph}) `quantising' the category of finite graphs and graph homomorphisms. This higher compositional structure has not been noticed before. 
\item We characterise quantum bijections and quantum isomorphisms as \emph{dagger-dualisable} 1-morphisms (Definition~\ref{def:daggerdualisable}) in the 2-categories $\QSet$ and $\QGraph$, respectively (Theorem \ref{thm:dual2}, Theorem~\ref{thm:dualisablequantumgraph}). We also show that classical bijections and graph isomorphisms are equivalences in these 2-categories (see Proposition~\ref{prop:equivalence}). In particular, we emphasise that quantum bijections and quantum graph isomorphisms should not be thought of as invertible but merely as dualisable, a characterisation which crucially depends on the $2$-categorical structure. 

\item We show that the categories $\QSet(A,B)$ of quantum functions between quantum sets $A$ and $B$, and the categories $\QGraph(G,H)$ of quantum homomorphisms between quantum graphs $G$ and $H$, are \emph{semisimple}, a categorical property which commonly arises in representation theory. This is crucial for the structural understanding of quantum functions, and allows us to characterise certain quantum functions as essentially classical (Definition~\ref{def:classical}).
\end{itemize}

\noindent
Our notions of quantum graph homomorphism and quantum graph isomorphism coincide with those considered in the theory of nonlocal games~\cite{Mancinska2016,Atserias2016} (Propositions \ref{prop:mancinskagraphhom} and \ref{prop:Atserias}). 
We also show that the monoidal category of quantum bijections on a quantum set is the category of finite-dimensional representations of the Hopf $C^*$-algebra associated to Wang's quantum permutation group (Proposition~\ref{prop:Wang},~\cite{Wang1998}). Similarly, we show that the monoidal category of quantum automorphisms of a graph is the category of finite-dimensional representations of the Hopf $C^*$-algebra associated to Banica's quantum automorphism group (Proposition~\ref{prop:Banica},~\cite{Banica2005}).
\ignore{In other words, our quantum permutations and quantum graph automorphisms can be understood as `quantum elements' of these compact quantum groups (cf. Theorem~\ref{thm:universalprop}, Section~\ref{sec:Wang} and Remark~\ref{rem:Banicagraph}).} The $2$-categories $\QSet$ and $\QGraph$ therefore extend these quantum groups analogously to the way in which categories extend groups or monoids.

The results in this work serve as a foundation for finite quantum set and quantum graph theory. In a subsequent paper~\cite{paper1b}, we use this framework to classify quantum isomorphic graphs in terms of algebraical and group theoretical data.

\paragraph{A new approach to quantum functions.}While noncommutative topology is formulated in the language of operator algebras, we present a compositional, linear algebraic framework based on the graphical calculus of string diagrams. In our diagrams, wires represent finite-dimensional Hilbert spaces, and vertices represent linear maps between them. Wiring diagrams, read from bottom to top, represent composite linear maps in the obvious way. We start from the following simple graphical observation:\footnote{This observation is a version of the famous \emph{Eckmann-Hilton argument}~\cite{Eckmann1962} playing a central role in algebraic topology and higher category theory.}
\begin{calign}\label{eq:basicidea}
\def\s{0.4cm}
\def \d{0.1cm}
\begin{tz}[zx,every to/.style={out=up, in=down},scale=1.5,yscale=-1]
\path (1.5,0) to node[zxnode=\zxwhite, pos=0.25,scale=1.2](a) {$a$} node[zxnode=\zxwhite, pos=0.75,scale=1.2] (b) {$b$}(0,2);
\draw[dashed,->] ([xshift=\s+\d, yshift=\s-\d]b.center)  to [out=20, in=70,looseness=1.1] ([xshift=\s-\d,yshift=\s+\d]a.center);
\draw[dashed, ->] ([xshift=-\s-\d, yshift=-\s+\d] a.center) to [out=-160, in=-110, looseness=1.1] ([xshift=-\s+\d, yshift=-\s-\d] b.center) ;
\end{tz} 
~~=~~
\begin{tz}[zx,every to/.style={out=up, in=down},scale=1.5,yscale=-1]
\path (1.5,0) to node[zxnode=\zxwhite, pos=0.25,scale=1.2]{$b$} node[zxnode=\zxwhite, pos=0.75,scale=1.2] {$a$}(0,2);
\end{tz} 
&
\begin{tz}[zx,every to/.style={out=up, in=down},scale=1.5,yscale=-1]
\draw (1.5,0) to node[zxnode=\zxwhite, pos=0.25,scale=1.2]{$a$} node[zxnode=\zxwhite, pos=0.75,scale=1.2] {$b$}(0,2);
\end{tz} 
~~\neq ~~
\begin{tz}[zx,every to/.style={out=up, in=down},scale=1.5,yscale=-1]
\draw (1.5,0) to node[zxnode=\zxwhite, pos=0.25,scale=1.2]{$b$} node[zxnode=\zxwhite, pos=0.75,scale=1.2] {$a$}(0,2);
\end{tz}
\end{calign} 
\noindent
The two free-floating vertices on the left can move around each other, whereas on the right they are confined to move on a line and do not commute.\footnote{The two equations in~\eqref{eq:basicidea} express nothing more than the fact that the algebra $\mathbb{C}$ is commutative, while the algebra $\End(H)$ of endomorphisms of a finite-dimensional Hilbert space $H$ is not.} In other words, adding a wire to a diagram can turn a commutative situation into a noncommutative one.

Taking this idea seriously leads to the following sketch of a programme. By Gelfand duality, a combinatorial theory, such as the theory of finite sets and functions or finite graphs and graph homomorphisms, can be expressed in terms of finite-dimensional commutative algebra and therefore represented in the graphical calculus. To quantise the objects of the theory --- in our case, sets or graphs --- one simply passes from commutative to noncommutative algebras. The novelty of our approach lies in our treatment of morphisms: having formulated the theory in term of string diagrams, we quantise morphisms by adding a wire through the corresponding vertices. We illustrate this idea with an example:%
\begin{calign} \label{eq:quantisationexample}
\begin{tz}[zx, master, every to/.style={out=up, in=down},xscale=-1]
\draw (0,0) to (0,2) to [out=135] (-0.75,3);
\draw (0,2) to [out=45] (0.75, 3);
\node[zxnode=\zxwhite] at (0,1) {$f$};
\node[zxvertex=\zxwhite, zxup] at (0,2) {};
\end{tz}
~~~=~~~
\begin{tz}[zx, every to/.style={out=up, in=down},xscale=-1]
\draw (0,0) to (0,0.75) to [out=135] (-0.75,1.75) to (-0.75,3);
\draw (0,0.75) to [out=45] (0.75, 1.75) to +(0,1.25);
\node[zxnode=\zxwhite] at (-0.75,2) {$f$};
\node[zxnode=\zxwhite] at (0.75,2) {$f$};
\node[zxvertex=\zxwhite, zxup] at (0,0.75) {};
\end{tz}
\hspace{1cm}\rightsquigarrow\hspace{1cm}
\begin{tz}[zx,every to/.style={out=up, in=down},xscale=-1]
\draw (0,0) to (0,2) to [out=45] (0.75,3);
\draw (0,2) to [out=135] (-0.75,3);
\draw[arrow data={0.2}{>}, arrow data={0.8}{>}] (1.75,0) to [looseness=0.9] node[zxnode=\zxwhite, pos=0.5] {$f$} (-1.75,2.5) to (-1.75,3);
\node[zxvertex=\zxwhite, zxup] at (0,2){};
\end{tz}
~~~=~~~
\begin{tz}[zx,every to/.style={out=up, in=down},xscale=-1]
\draw (0,0) to (0,0.75) to [out=45] (0.75,1.75) to (0.75,3);
\draw (0,0.75) to [out=135] (-0.75,1.75) to (-0.75,3);
\draw[arrow data={0.2}{>}, arrow data={0.9}{>}] (1.75,0) to (1.75,0.75) to  [looseness=1.1, in looseness=0.9] node[zxnode=\zxwhite, pos=0.36] {$f$} node[zxnode=\zxwhite, pos=0.64] {$f$}(-1.75,3);
\node[zxvertex=\zxwhite, zxup] at (0,0.75){};
\end{tz}
\end{calign}
\noindent 
The equation on the left is one of the properties of any function between two sets, while the equation on the right is its `quantisation'. 
Further examples of classical morphisms and their quantisations are displayed in Figure~\ref{fig:classicalquantum}.
%
\ignore{
Thus defined\JV{You haven't said how to define them.}, quantum morphisms between classical sets admit concrete operational interpretations\JV{It is unclear to me why this would be, or why this is significant.}. For example, a quantum element of a finite set $X$ is a projective measurement with outcomes in $X$, while a quantum function $X\to Y$ between finite sets is a family of projective measurements controlled by elements in $X$ with outcomes in $Y$ (see Corollary \ref{cor:quantumelementmeasurement}).\looseness=-2}

\paragraph{The compositional structure of quantum functions.}

Due to the presence of the additional Hilbert space wire introduced in our quantisation procedure, we are naturally led to consider maps on this Hilbert space, which interact in a particular way with quantum functions. \ignore{With these additional maps, quantum functions from $A$ to $B$ form, not a set, but a \emph{category} $\QSet(A,B)$ of which they are the objects.\JV{This sentence is quite convoluted, but the content is critical. I suggest you rephrase it to make it absolutely crystal clear.}
}
With these additional maps as morphisms, quantum functions between quantum sets $A$ and $B$ are the objects of a category $\QSet(A,B)$. This category of quantum functions between two quantum sets can be understood as a quantisation of the set of classical functions between two classical sets. In other words, our approach to quantisation leads to categorification.
In particular, the category of quantum bijections on a set $X$ quantises the symmetric group $S_X$; as stated above, this is the category of quantum elements of Wang's quantum permutation group.

A shortcoming of the compact quantum group approach is that it is restricted to automorphisms and cannot easily be applied to study quantum functions or graph homomorphisms between nonisomorphic sets or graphs, as is for example necessary for quantum pseudo-telepathy. 
In contrast, our approach allows us to consider all quantum functions and quantum graph homomorphisms, including those between nonisomorphic sets and graphs. We show that these form 2-categories $\QSet$ and $\QGraph$, which incorporate all the categories of quantum functions and homomorphisms between different sets and graphs into a single mathematical structure. These $2$-categories unify, generalise, and most importantly expose the connection between work on quantum permutations and quantum graph automorphisms in noncommutative topology, and on quantum graph homomorphisms and isomorphisms in quantum information theory.

\subsection{Outlook}
We hope that our results lead to the development of further connections between the work on compact quantum groups and quantum information theory. In particular, we expect that the structural understanding of the automorphism categories $\QBij(B,B)$ and $\QIso(G,G)$ can lead to new insights into pseudo-telepathy. First considerations along these lines will appear in a companion paper~\cite{paper1b}. \ignore{There, we use properties of the category $\QIso(G,G)$ to classify quantum graph isomorphisms $G\to H$ out of this graph\ignore{, where $H$ is a possibly non-classically isomorphic graph}. This allows us to classify quantum strategies for the graph isomorphism game of Atserias et al~\cite{Atserias2016} and to devise a scheme for producing pairs of graphs exhibiting non-classical behaviour. }

 We now suggest a number of other possible applications of our categorical \mbox{framework}:

\begin{itemize}
\item Quantum functions are closely related to quantum communication. Indeed, in Remark~\ref{rem:teleportation}, we note that quantum bijections between a matrix algebra $\mathrm{Mat}_n$ and an $n^2$-element classical set can be understood as generalised teleportation protocols. More generally, it has been observed that quantum homomorphisms between quantum graphs admit an interpretation in terms of zero-error source-channel coding~\cite{Stahlke2016}.

\item Our classification in~\cite{paper1b} is based on the study of Frobenius monads in $\QGraph$. As symmetric monoidal $2$-categories (see Remark~\ref{rem:symmmon2cat}), $\QSet$ and $\QGraph$  admit several other as yet unexplored algebraic structures.

\item Our quantisation framework may be extended to other combinatorial theories besides finite set theory and graph theory. In particular, in Section~\ref{app:quantumrelation}, we show how Kuperberg and Weaver's finite-dimensional quantum relations~\cite{Kuperberg2012} fit into our framework; this suggests, for example, a $2$-category of quantum posets and quantum monotone functions.\looseness=-2
\end{itemize}

\subsection{Outline of the paper}

In Section~\ref{sec:background}, we recall the diagrammatic calculus for Hilbert spaces and linear maps as well as the correspondence between finite-dimensional $C^*$-algebras and certain Frobenius algebras. We express Gelfand duality between finite sets and finite-dimensional commutative $C^*$-algebras in this setting.

In Section~\ref{sec:quantumset}, we introduce quantum sets, quantum elements and quantum functions, and define the 2-category $\QSet$ which encodes their compositional structure. We  justify our definitions by a universal property. 

In Section~\ref{sec:quantumbijection}, we quantise bijections and prove that the resulting quantum bijections are precisely dagger-dualisable 1-morphisms in $\QSet$. We show that quantum bijections between classical sets correspond to projective permutation matrices as defined in the theory of nonlocal games~\cite{Atserias2016}, and to finite-dimensional representations of Wang's quantum symmetry group algebras~\cite{Wang1998}.

In Section~\ref{sec:qgt}, we quantise finite graphs and their homomorphisms and define the \mbox{$2$-category} $\QGraph$ which encodes their compositional structure. We show that our definitions capture the quantum graph homomorphisms and isomorphisms of Man{\v{c}}inska and Roberson~\cite{Mancinska2016} and Atserias et al~\cite{Atserias2016}, as well as the finite-dimensional representation theory of Banica and Bichon's quantum automorphism group algebras~\cite{Banica2005,Bichon2003}. We also show that quantum graph isomorphisms are precisely dagger-dualisable 1-morphisms in $\QGraph$. 

In Section~\ref{sec:semisimple}, we show that the categories of quantum functions and quantum graph homomorphisms are semisimple. We discuss the operational interpretation of the direct sum of quantum functions and show how this can be used to distinguish between classical and quantum structures. 
\ignore{Finally, we show how semisimplicity brings various results from finite noncommutative geometry into our 2-categorical framework. }

In Section~\ref{app:quantumrelation}, we show that our reflexive quantum graphs are precisely symmetric and reflexive quantum relations in the sense of Kuperberg and Weaver~\cite{Kuperberg2012,Weaver2010,Weaver2015} and fully capture the noncommutative graphs appearing in zero-error communication~\cite{Duan2013}.

\subsection{Related work}\label{sec:relatedwork}
A compressed summary of this section can be found in Figure~\ref{fig:relatedwork}.
\begin{figure}[]
\vspace{-10pt}
\begin{tabular}{l l l}
\scalebox{.6}{\begin{tikzpicture}
	\begin{pgfonlayer}{nodelayer}
		\node [style=none] (0) at (0, 4) {};
		\node [style=none] (1) at (4, -0) {};
		\node [style=none] (2) at (8, 4) {};
		\node [style=none] (3) at (4, 8) {};
		\node [style=none] (4) at (4, 1) {{\Huge \bf QSet}};
		\node [style=none, draw, circle, thick, black, minimum size=.7cm] (5) at (4, 4.4) {C};
		\node [style=none] (6) at (5.25, 6.25) {1.};
		\node [style=none] (7) at (2.5, 3) {};
		\node [style=none] (8) at (5.5, 3) {};
		\node [style=none] (9) at (4.25, 3.25) {2.};
		\node [style=none, draw, circle, thick, black, minimum size=.7cm] (10) at (2, 3) {};
		\node [style=none, draw, circle, thick, black, minimum size=.7cm] (11) at (6, 3) {};
		\node [style=none] (12) at (4.5, 4.75) {};
		\node [style=none] (13) at (3.5, 4.75) {};
	\end{pgfonlayer}
	\begin{pgfonlayer}{edgelayer}
		\draw [style=simple, opacity=.5, bend left=45, looseness=1.00] (0.center) to (3.center);
		\draw [style=simple, opacity=.5, bend left=45, looseness=1.00] (3.center) to (2.center);
		\draw [style=simple, opacity=.5, bend right=45, looseness=1.00] (1.center) to (2.center);
		\draw [style=simple, opacity=.5, bend left=45, looseness=1.00] (1.center) to (0.center);
		\draw [style=arrow] (7.center) to (8.center);
		\draw [style=tick, in=45, out=135, looseness=8.50] (13.center) to (12.center);
	\end{pgfonlayer}
\end{tikzpicture}} & \scalebox{.6}{\begin{tikzpicture}
	\begin{pgfonlayer}{nodelayer}
		\node [style=none] (0) at (4, 1) {{\Huge \bf QGraph}};
		\node [style=none] (1) at (8, 4) {};
		\node [style=none] (2) at (4, 8) {};
		\node [style=none] (3) at (4, -0) {};
		\node [style=none] (4) at (0, 4) {};
		\node [style=none, draw, circle, thick, black, minimum size=.7cm] (5) at (2, 2.75) {};
		\node [style=none, draw, circle, thick, black, minimum size=.7cm] (6) at (6, 2.75) {};
		\node [style=none] (7) at (2.5, 2.75) {};
		\node [style=none] (8) at (5.5, 2.75) {};
		\node [style=none] (9) at (4, 6.5) {4.};
		\node [style=none] (10) at (2, 4.75) {};
		\node [style=none] (11) at (0.75, 4) {3.};
		\node [style=none] (12) at (5.25, 4.5) {};
		\node [style=none, draw, circle, thick, black, minimum size=.7cm] (13) at (5.5, 5) {C};
		\node [style=none] (14) at (3, 4.5) {};
		\node [style=none] (15) at (2, 5.25) {};
		\node [style=none, draw, circle, thick, black, minimum size=.7cm] (16) at (2.5, 5) {C};
		\node [style=none] (17) at (3, 5.5) {};
		\node [style=none] (18) at (5.25, 5.5) {};
		\node [style=none] (19) at (4, 4.5) {5.};
		\node [style=none] (20) at (4, 2.25) {6.};
		\node [style=none] (21) at (6.5, 3) {};
		\node [style=none] (22) at (6.5, 3.25) {7.};
	\end{pgfonlayer}
	\begin{pgfonlayer}{edgelayer}
		\draw [style=simple, opacity=.5, bend left=45, looseness=1.00] (4.center) to (2.center);
		\draw [style=simple, opacity=.5, bend left=45, looseness=1.00] (2.center) to (1.center);
		\draw [style=simple, opacity=.5, bend right=45, looseness=1.00] (3.center) to (1.center);
		\draw [style=simple, opacity=.5, bend left=45, looseness=1.00] (3.center) to (4.center);
		\draw [style=arrow] (7.center) to (8.center);
		\draw [style=tick, bend left=135, looseness=13.75] (10.center) to (15.center);
		\draw [style=tick, bend left=45, looseness=1.00] (17.center) to (18.center);
		\draw [style=arrow, bend right=45, looseness=0.75] (14.center) to (12.center);
	\end{pgfonlayer}
\end{tikzpicture}} &\raisebox{1.2cm}{\scalebox{.8}{\begin{tikzpicture}
	\begin{pgfonlayer}{nodelayer}
		\node [style=none] (0) at (1, 2.5) {};
		\node [style=none] (1) at (4.25, 2) {\begin{tabular}{l} 1-morphisms \\[.37cm] Dagger-dualisable 1-morphisms \\[.37cm] Objects \\[.37cm] Classical objects \end{tabular}};
		\node [minimum size=.5cm, black, thick, circle, draw, style=none] (2) at (0.5, .7) {C};
		\node [style=none] (3) at (0, 2.5) {};
		\node [style=none, draw, circle, thick, black, minimum size=.5cm] (4) at (0.5, 1.65) {};
		\node [style=none] (5) at (0, 3.25) {};
		\node [style=none] (6) at (1, 3.25) {};
	\end{pgfonlayer}
	\begin{pgfonlayer}{edgelayer}
		\draw [style=tick] (3.center) to (0.center);
		\draw [style=arrow] (5.center) to (6.center);
	\end{pgfonlayer}
\end{tikzpicture}}}
\end{tabular}\vspace{.2cm} {\footnotesize\ 1. Finite-dimensional (f.d.) representations of Wang's quantum permutation group algebra~\cite{Wang1998}. \\ \hphantom{1.}\, Projective permutation matrices of Atserias et al.~\cite{Atserias2016}.\\2. F.d. representations of So{\l}tan's quantum space of all maps between finite quantum spaces~\cite{Soltan2009}. \\ 3. F.d. representations of Banica's quantum automorphism group algebra~\cite{Banica2005}. \\ 4. Quantum graph isomorphisms of Atserias et al.~\cite{Atserias2016}. \\ 5. Quantum graph homomorphisms of Man\v{c}inska and Roberson~\cite{Mancinska2016}. \\ 6. `Pure' entanglement-assisted graph homomorphisms of Stahlke~\cite{Stahlke2016}. \\ 7. Quantum graphs of Weaver~\cite{Weaver2015}, generalising noncommutative graphs of Duan et al~\cite{Duan2013}.}
\caption{A summary of related work.}\label{fig:relatedwork}
\vspace{-5pt}
\end{figure}
\paragraph{Quantum symmetry groups and noncommutative topology.}

The study of quantum permutation groups --- quantum variants of the symmetric groups $S_n$ in noncommutative topology ---  was suggested by Connes, and carried out by Wang, Banica, Bichon and others~\cite{Wang1998, Banica2005,Banica2007, Banica2007_2, Bichon2003, Banica2009, Banica2007_3}.
These quantum permutation groups are compact quantum groups in the sense of Woronowicz~\cite{Woronowicz1998}, obtained from a universal construction~\cite{Wang1998}. Quantum automorphism groups of finite graphs are defined similarly~\cite{Bichon2003,Banica2005}. A more general universal construction of quantum spaces of maps is given by So{\l}tan~\cite{Soltan2009} (see Remark~\ref{rem:soltan}).
In Section~\ref{sec:universal}, we show that our categories of quantum functions $\QSet(A,B)$ can be obtained from an analogous universal construction, as the category of finite-dimensional representations of an internal hom $[A,B]$ in the opposite of the category of $C^*$-algebras.
Likewise, our categories $\QBij([n],[n])$ of quantum bijections on an $n$-element set are the categories of finite-dimensional representations of the Hopf $C^*$-algebra $A(n)$ corresponding to Wang's quantum permutation group~\cite{Wang1998} (Proposition \ref{prop:Wang}), and our categories $\QIso(G,G)$ of quantum automorphisms of a graph $G$ are the categories of finite-dimensional representations of the Hopf $C^*$-algebra $A(G)$ corresponding to Banica's quantum automorphism group of the graph~\cite{Banica2007} \mbox{(Proposition \ref{prop:Banica}).}

Our framework therefore captures the finite-dimensional representation theory of the Hopf $C^*$-algebras considered in noncommutative topology. However, we note that a Tannakian correspondence exists only between these algebras and their categories of corepresentations (or, in the language of quantum group theory, with the representations of their associated compact quantum groups~\cite{Neshveyev2013}. Also see Remark~\ref{rem:problemTannaka} and Remark~\ref{rem:Wangnotdiscrete}). 

Based on the theory of quantum relations developed by Kuperberg and Weaver~\cite{Kuperberg2012,Weaver2010,Weaver2015}, Kornell~\cite{Kornell2011,Kornell2018} defines (possibly infinite) quantum sets and quantum functions between them. His notion of finite quantum set coincides with ours; his quantum functions between finite quantum sets are *-homomorphisms and thus one-dimensional quantum functions in our sense.

\paragraph{Quantum information theory.} 
Quantum graph homomorphisms were defined by Man\v{c}inska and Roberson~\cite{Mancinska2016} as generalisations of quantum graph colourings~\cite{Cameron2007} in the context of nonlocal games (see Remark~\ref{rem:quantumhomMancinska}) and have been the subject of intensive study in their various forms~\cite{Mancinska2016,Scarpa2012,Avis2006,Cameron2007,Paulsen2015,Paulsen2016,Roberson2016}. Quantum graph isomorphisms are originally due to Atserias et al~\cite{Atserias2016} (see Remark~\ref{rem:Atseriasquantumiso}).

A related $C^*$-algebraic approach to quantum homomorphisms between classical graphs was recently discovered by Ortiz and Paulsen~\cite{Ortiz2016}. They define a $C^*$-algebra which is essentially an internal hom, analogous to the algebra of quantum functions in Remark~\ref{rem:quantumfunctionalgebra}.

Our notion of reflexive quantum graph coincides with that of Kuperberg and Weaver~\cite{Kuperberg2012, Weaver2015,Weaver2010} (see Remark~\ref{rem:literatureqgraph} and Theorem~\ref{thm:relationadjacency}) and in particular generalises Duan, Severini and Winter's noncommutative graphs~\cite{Duan2013} (see Proposition~\ref{prop:quantumgraphDuanthesame}). Quantum graph homomorphisms between noncommutative graphs can be understood as pure versions of Stahlke's entanglement-assisted morphisms~\cite{Stahlke2016} (see Remark~\ref{rem:stahlkegraphhoms}).

\paragraph{Categorical quantum mechanics (CQM).}
Our work emerges from the CQM research programme, initiated by Abramsky and Coecke~\cite{Abramsky2008} and developed by them and others~\cite{Coecke2008,Coecke2007,Coecke2009,Vicary2010,Coecke2014}, which uses the graphical calculus of monoidal categories to provide a high-level syntax for quantum information flow. In particular, Vicary's reformulation of the theory of finite-dimensional $C^*$-algebras in terms of Frobenius algebras~\cite{Vicary2010} and Coecke, Pavlovi{\'c} and Vicary's proof of Gelfand duality in this setting~\cite{Coecke2009} are important starting points and guide posts for our work (see Section~\ref{sec:background}). Many of our proofs and constructions can be understood as simple quantisations, in the sense of equation~\eqref{eq:quantisationexample}, of constructions from categorical quantum mechanics.

A related 2-categorical framework for quantum theory based on the 2-category $\mathrm{2Hilb}$ of categorified Hilbert spaces was studied by Vicary and the second author~\cite{Vicary2012hq, Vicary2012_2, Reutter2016}; we remark that there is a locally faithful 2-functor $\QSet \to \mathrm{2Hilb}$ allowing us to translate most of our results into this setting.

\paragraph{Monads.}
The 2-category of quantum sets $\QSet$ is an instance of Street's 2-category of monads~\cite{Street1972} in the $2$-category $\BHilb$, the delooping of the monoidal category\! $\Hilb$ (see Remark~\ref{rem:street}). Hinze and Marsden~\cite{Marsden2014,Hinze2016} give an analogous graphical treatment of the 2-category $\mathrm{Mnd}(\mathrm{Cat})$ of monads in the 2-category $\mathrm{Cat}$ of categories, functors and natural transformations.

Recent work by Abramsky and others~\cite{Abramsky2017} uses monads to study binary constraint systems and quantum graph homomorphisms. There are many similarities between their work and ours; for example, composition in their Kleisli category corresponds to the composition of $1$-morphisms in $\QSet$ (see Remark~\ref{rem:ppmcomposition}).

\subsection{Definitions and conventions}
We assume some basic familiarity with monoidal category theory~\cite{Selinger2010} and $2$-category theory~\cite[Chapter 7]{Borceux1994}. Dagger categories are defined in~\cite{Selinger2010}; strict dagger $2$-categories\footnote{Weak dagger 2-categories are the obvious generalisation, with unitary associators and unitors.} and dagger $2$-functors are defined in~\cite{Heunen2016}. 
When we refer to local properties of a 2-functor (such as local faithfulness), we mean that these properties are true of the induced functors on hom-categories.
We use the words \emph{module} and \emph{representation} interchangeably. A projector on a Hilbert space $H$ is an endomorphism $P: H \to H$ which is idempotent and self-adjoint $P = P^{\dagger} = P^2$. All sets appearing in this work are finite and, except where clearly specified, all vector spaces and all algebras are finite-dimensional; a notable exception is Section~\ref{sec:universal}. Consequently, we use the labels \Set\ and \Hilb\ for the categories of finite sets and finite-dimensional Hilbert spaces respectively. We also take all $C^*$-algebras to be unital. We denote the $n$-element set by $[n]$.

\subsection*{Acknowledgments}
We are especially grateful to Jamie Vicary for many useful conversations and comments on an earlier draft of this paper. We thank Andre Kornell for a useful conversation about our definitions of quantum functions. We also thank Samson Abramsky, and Rui Soares Barbosa for discussing the relationship between our categorical frameworks; Giulio Chiribella for introducing us to existing notions of noncommutative graphs; and Matty Hoban for informative conversations about nonlocal games. We appreciate the helpful comments of an anonymous referee regarding the presentation of these results.

\newpage
\def\wl{0.4}
\def\wr{0.5}
\def\wtext{1.5cm}
\def\bnd{-1cm}
\def\bndr{-1cm}
\def\bndm{0.15cm}
\begin{figure}[h!]
\vspace{-8pt}
\[
\hspace{-2.3cm}
\setlength\arraycolsep{1pt}
\begin{array}{|c|c||c|c|}
\hline
\rule{0pt}{3ex} 
\begin{minipage}{1.7cm}\center
\text{set}
\end{minipage}&
\text{commutative algebra} &\text{non-commutative algebra} & \text{quantum set}
\\[-9pt] \rule{0pt}{0ex}&&&\\\hline
\text{element}& 
\begin{minipage}{\wl\textwidth}
\begin{calign}\nonumber
\begin{tz}[zx, master, every to/.style={out=up, in=down},xscale=-1]
\draw (0,1) to (0,2) to [out=135] (-0.75,3);
\draw (0,2) to [out=45] (0.75, 3);
\node[zxnode=\zxwhite] at (0,1) {$\psi$};
\node[zxvertex=\zxwhite, zxup] at (0,2) {};
\end{tz}
=~~~
\begin{tz}[zx,slave, every to/.style={out=up, in=down},xscale=-1]
\draw (-0.75,2) to (-0.75,3);
\draw (0.75, 2) to +(0,1.);
\node[zxnode=\zxwhite] at (-0.75,2) {$\psi$};
\node[zxnode=\zxwhite] at (0.75,2) {$\psi$};
\end{tz}
\hspace{1cm}
\begin{tz}[zx, every to/.style={out=up, in=down},xscale=-1]
\draw (0,1) to (0,2) ;
\node[zxnode=\zxwhite] at (0,1) {$\psi$};
\node[zxvertex=\zxwhite, zxup] at (0,2) {};
\end{tz}
~~=~~
\emptydiagram
\\\nonumber
\begin{tz}[zx,slave, every to/.style={out=up, in=down},xscale=-1]
\draw (0,0) to (0,1.5);
\node[zxnode=\zxwhite] at (0,1.5) {$\psi^\dagger$};
\end{tz}
=~~~
\begin{tz}[zx,slave,every to/.style={out=up, in=down},xscale=-1]
\draw (0,1.5) to (0,2) to [in=left] node[pos=1] (r){} (0.5,2.5) to [out=right, in=up] (1,2)  to (1,0);
\node[zxnode=\zxwhite] at (0,1.5) {$\psi$};
\node[zxvertex=\zxwhite,zxdown] at (r.center){};
\end{tz}
\end{calign}
\end{minipage}
\hspace{\bnd}
&
\hspace{\bndm}
\begin{minipage}{\wr\textwidth}
\begin{calign}\nonumber
\begin{tz}[zx,every to/.style={out=up, in=down},xscale=-1]
\draw (0,0) to (0,2) to [out=45] (0.75,3);
\draw (0,2) to [out=135] (-0.75,3);
\draw[arrow data={0.2}{>}, arrow data={0.8}{>}] (1.75,0) to [looseness=0.9] node[zxnode=\zxwhite, pos=0.5](Q) {$Q$} (-1.75,2.5) to (-1.75,3);
\node[zxvertex=\zxwhite, zxup] at (0,2){};
\draw[white,double] (0,0) to (Q.center);
\end{tz}
=
\begin{tz}[zx,every to/.style={out=up, in=down},xscale=-1]
\draw (0.75,1.75) to (0.75,3);
\draw  (-0.75,1.75) to (-0.75,3);
\draw[arrow data={0.2}{>}, arrow data={0.9}{>}] (1.75,0) to (1.75,0.75) to  [looseness=1.1, in looseness=0.9] node[zxnode=\zxwhite, pos=0.36] {$Q$} node[zxnode=\zxwhite, pos=0.64] {$Q$}(-1.75,3);
\end{tz}
\hspace{1cm}
\begin{tz}[zx,every to/.style={out=up, in=down},xscale=-1]
\draw (0,0) to (0,2.25);
\draw[arrow data={0.2}{>}, arrow data={0.8}{>}] (1,0) to [looseness=0.9] node[zxnode=\zxwhite, pos=0.5](Q) {$Q$} (-1,2.5) to (-1,3);
\node[zxvertex=\zxwhite] at (0,2.25){};
\draw[white,double] (0,0) to (Q.center);
\end{tz}
=
\begin{tz}[zx, every to/.style={out=up, in=down},xscale=-1]
\draw[arrow data={0.2}{>}, arrow data={0.9}{>}] (1.,0) to [looseness=0.9]  (-1,2.5) to (-1,3);
\end{tz}
\\[-3pt]\nonumber
\begin{tz}[zx,every to/.style={out=up, in=down},xscale=-0.8]
\draw [arrow data={0.2}{>},arrow data={0.8}{>}]  (0,0) to (2.25,3);
\draw (2.25,0) to [in=-45] node[zxnode=\zxwhite, pos=1] (Q) {$Q^\dagger$} (1.125,1.5);
\node at (1,-0.5){};
\node at (1,3.5){};
\end{tz}
~=~~
\begin{tz}[zx,xscale=-0.6,yscale=-1]
\draw[arrow data={0.5}{<}] (0.25,-0.5) to (0.25,0) to [out=up, in=-135] (1,1);
\draw[arrow data={0.5}{>}] (1.75,2.5) to (1.75,2) to [out=down, in=45] (1,1);
\draw (1,1) to [out= -45, in= left] node[zxvertex=\zxwhite, pos=1] {} (2.3,0.3) to [out=right, in=down] (3.25,1) to (3.25,2.5);
\node [zxnode=\zxwhite] at (1,1) {$Q$};
\node at (1,-1){};
\node at (1,3){};
\end{tz}
\end{calign}
\end{minipage}
\hspace{\bndr}
&
\begin{minipage}{\wtext}\center quantum \\ element \end{minipage}\\
\hline
\text{function} &
\begin{minipage}{\wl\textwidth}
\begin{calign}\nonumber
\begin{tz}[zx, master, every to/.style={out=up, in=down},xscale=-1]
\draw (0,0) to (0,2) to [out=135] (-0.75,3);
\draw (0,2) to [out=45] (0.75, 3);
\node[zxnode=\zxwhite] at (0,1) {$f$};
\node[zxvertex=\zxwhite, zxup] at (0,2) {};
\end{tz}
=
\begin{tz}[zx, every to/.style={out=up, in=down},xscale=-1]
\draw (0,0) to (0,0.75) to [out=135] (-0.75,1.75) to (-0.75,3);
\draw (0,0.75) to [out=45] (0.75, 1.75) to +(0,1.25);
\node[zxnode=\zxwhite] at (-0.75,2) {$f$};
\node[zxnode=\zxwhite] at (0.75,2) {$f$};
\node[zxvertex=\zxwhite, zxup] at (0,0.75) {};
\end{tz}
\hspace{1cm}
\begin{tz}[zx,slave, every to/.style={out=up, in=down},xscale=-1]
\draw (0,0) to (0,2) ;
\node[zxnode=\zxwhite] at (0,1) {$f$};
\node[zxvertex=\zxwhite, zxup] at (0,2) {};
\end{tz}
=
\begin{tz}[zx,slave, every to/.style={out=up, in=down},xscale=-1]
\draw (0,0) to (0,0.75) ;
\node[zxvertex=\zxwhite, zxup] at (0,0.75) {};
\end{tz}
\\\nonumber
\begin{tz}[zx,slave, every to/.style={out=up, in=down},xscale=-1]
\draw (0,0) to (0,3);
\node[zxnode=\zxwhite] at (0,1.5) {$f^\dagger$};
\end{tz}
=~~
\begin{tz}[zx,slave,every to/.style={out=up, in=down},xscale=-1]
\draw (0,1.5) to (0,2) to [in=left] node[pos=1] (r){} (0.5,2.5) to [out=right, in=up] (1,2)  to (1,0);
\draw (-1,3) to [out=down,in=up] (-1,1) to [out=down, in=left] node[pos=1] (l){} (-0.5,0.5) to [out=right, in=down] (0,1) to (0,1.5);
\node[zxnode=\zxwhite] at (0,1.5) {$f$};
\node[zxvertex=\zxwhite,zxup] at (l.center){};
\node[zxvertex=\zxwhite,zxdown] at (r.center){};
\end{tz}
\end{calign}\\[-15pt]
\end{minipage}
\hspace{\bnd}
&
\hspace{\bndm}
\begin{minipage}{\wr\textwidth}
\begin{calign}\nonumber
\begin{tz}[zx,every to/.style={out=up, in=down},xscale=-1]
\draw (0,0) to (0,2) to [out=45] (0.75,3);
\draw (0,2) to [out=135] (-0.75,3);
\draw[arrow data={0.2}{>}, arrow data={0.8}{>}] (1.75,0) to [looseness=0.9] node[zxnode=\zxwhite, pos=0.5] {$P$} (-1.75,2.5) to (-1.75,3);
\node[zxvertex=\zxwhite, zxup] at (0,2){};
\end{tz}
=
\begin{tz}[zx,every to/.style={out=up, in=down},xscale=-1]
\draw (0,0) to (0,0.75) to [out=45] (0.75,1.75) to (0.75,3);
\draw (0,0.75) to [out=135] (-0.75,1.75) to (-0.75,3);
\draw[arrow data={0.2}{>}, arrow data={0.9}{>}] (1.75,0) to (1.75,0.75) to  [looseness=1.1, in looseness=0.9] node[zxnode=\zxwhite, pos=0.36] {$P$} node[zxnode=\zxwhite, pos=0.64] {$P$}(-1.75,3);
\node[zxvertex=\zxwhite, zxup] at (0,0.75){};
\end{tz}
\hspace{1cm}
\begin{tz}[zx,every to/.style={out=up, in=down},xscale=-1]
\draw (0,0) to (0,2.25);
\draw[arrow data={0.2}{>}, arrow data={0.8}{>}] (1,0) to [looseness=0.9] node[zxnode=\zxwhite, pos=0.5] {$P$} (-1,2.5) to (-1,3);
\node[zxvertex=\zxwhite] at (0,2.25){};
\end{tz}
=
\begin{tz}[zx,every to/.style={out=up, in=down},xscale=-1]
\draw (0,0) to (0,0.75);
\draw[arrow data={0.2}{>}, arrow data={0.9}{>}] (1.,0) to (1.,0.75) to   (-1,3);
\node[zxvertex=\zxwhite, zxup] at (0,0.75){};
\end{tz}
\\[2pt]\nonumber
\begin{tz}[zx,every to/.style={out=up, in=down},xscale=-0.8]
\draw [arrow data={0.2}{>},arrow data={0.8}{>}]  (0,0) to (2.25,3);
\draw (2.25,0) to node[zxnode=\zxwhite, pos=0.5] {$P^\dagger$} (0,3);
\end{tz}
=~~
\begin{tz}[zx,xscale=-0.6,yscale=-1]
\draw[arrow data={0.5}{<}] (0.25,-0.5) to (0.25,0) to [out=up, in=-135] (1,1);
\draw (1,1) to [out=135, in=right] node[zxvertex=\zxwhite, pos=1]{} (-0.3, 1.7) to [out=left, in=up] (-1.25,1) to (-1.25,-0.5);
\draw[arrow data={0.5}{>}] (1.75,2.5) to (1.75,2) to [out=down, in=45] (1,1);
\draw (1,1) to [out= -45, in= left] node[zxvertex=\zxwhite, pos=1] {} (2.3,0.3) to [out=right, in=down] (3.25,1) to (3.25,2.5);
\node [zxnode=\zxwhite] at (1,1) {$P$};
\end{tz} 
\end{calign}\\[-15pt]
\end{minipage}
\hspace{\bndr}
&
\begin{minipage}{\wtext}\center quantum \\ function \end{minipage} \\ \hline
\text{bijection} &
\begin{minipage}{\wl\textwidth}
\vspace{0.2cm}
\hspace{0.1cm}function $+$
\vspace{-0.1cm}
\begin{calign}\nonumber
\begin{tz}[zx, master, every to/.style={out=up, in=down},yscale=-1,xscale=-1]
\draw (0,0) to (0,2) to [out=135] (-0.75,3);
\draw (0,2) to [out=45] (0.75, 3);
\node[zxnode=\zxwhite] at (0,1) {$f$};
\node[zxvertex=\zxwhite, zxdown] at (0,2) {};
\end{tz}
=
\begin{tz}[zx, every to/.style={out=up, in=down},yscale=-1,xscale=-1]
\draw (0,0) to (0,0.75) to [out=135] (-0.75,1.75) to (-0.75,3);
\draw (0,0.75) to [out=45] (0.75, 1.75) to +(0,1.25);
\node[zxnode=\zxwhite] at (-0.75,2) {$f$};
\node[zxnode=\zxwhite] at (0.75,2) {$f$};
\node[zxvertex=\zxwhite, zxdown] at (0,0.75) {};
\end{tz}
\hspace{1cm}
\begin{tz}[zx,slave, every to/.style={out=up, in=down},yscale=-1,xscale=-1]
\draw (0,0) to (0,2) ;
\node[zxnode=\zxwhite] at (0,1) {$f$};
\node[zxvertex=\zxwhite, zxup] at (0,2) {};
\end{tz}
=
\begin{tz}[zx,slave, every to/.style={out=up, in=down},yscale=-1,xscale=-1]
\draw (0,0) to (0,0.75) ;
\node[zxvertex=\zxwhite, zxup] at (0,0.75) {};
\end{tz}
\end{calign}\\[-15pt]
\end{minipage}
\hspace{\bnd}
&
\hspace{\bndm}
\begin{minipage}{\wr\textwidth}
\vspace{0.2cm}
quantum function $+$
\vspace{-0.1cm}
\begin{calign}\nonumber
\begin{tz}[zx,every to/.style={out=up, in=down},scale=-1,xscale=-1]
\draw (0,0) to (0,2) to [out=45] (0.75,3);
\draw (0,2) to [out=135] (-0.75,3);
\draw[arrow data={0.2}{<}, arrow data={0.8}{<}] (1.75,0) to [looseness=0.9] node[zxnode=\zxwhite, pos=0.5] {$P$} (-1.75,2.5) to (-1.75,3);
\node[zxvertex=\zxwhite, zxdown] at (0,2){};
\end{tz}
=
\begin{tz}[zx,every to/.style={out=up, in=down},scale=-1,xscale=-1]
\draw (0,0) to (0,0.75) to [out=45] (0.75,1.75) to (0.75,3);
\draw (0,0.75) to [out=135] (-0.75,1.75) to (-0.75,3);
\draw[arrow data={0.2}{<}, arrow data={0.9}{<}] (1.75,0) to (1.75,0.75) to  [looseness=1.1, in looseness=0.9] node[zxnode=\zxwhite, pos=0.36] {$P$} node[zxnode=\zxwhite, pos=0.64] {$P$}(-1.75,3);
\node[zxvertex=\zxwhite, zxdown] at (0,0.75){};
\end{tz}
\hspace{1cm}
\begin{tz}[zx,every to/.style={out=up, in=down},scale=-1,xscale=-1]
\draw (0,0) to (0,2.25);
\draw[arrow data={0.2}{<}, arrow data={0.8}{<}] (1,0) to [looseness=0.9] node[zxnode=\zxwhite, pos=0.5] {$P$} (-1,2.5) to (-1,3);
\node[zxvertex=\zxwhite] at (0,2.25){};
\end{tz}
=
\begin{tz}[zx,every to/.style={out=up, in=down},scale=-1,xscale=-1]
\draw (0,0) to (0,0.75);
\draw[arrow data={0.2}{<}, arrow data={0.9}{<}] (1.,0) to (1.,0.75) to   (-1,3);
\node[zxvertex=\zxwhite] at (0,0.75){};
\end{tz}
\end{calign}\\[-15pt]
\end{minipage}
\hspace{\bndr}
&
\begin{minipage}{\wtext}\center quantum \\ bijection \end{minipage}   \\
\hline \nonumber
\begin{minipage}{1.5cm}\center graph \\homom.\end{minipage}&
\begin{minipage}{\wl\textwidth}
\vspace{0.2cm}
\hspace{0.01cm}function $+$
\vspace{-0.2cm}
\def\d{2}
\begin{calign}\nonumber
\begin{tz}[zx,yscale=0.8,xscale=-1]
\draw[string] (0,-0.5) to node[front,zxnode=\zxwhite,pos=0.5] {$f$} (0,5);\draw[string] (\d,-0.5) to node[front,zxnode=\zxwhite, pos=0.5] {$f$} (\d,5);
\draw[string] (0, 0.5) to (\d,0.5);
\node[zxvertex=\zxwhite] at (0,0.5){};
\node[zxvertex=\zxwhite] at (\d,0.5){};
\node[zxnode=\zxwhite] at (\d/2, 0.5 ) {$G$};
\end{tz}
~~=~~
\begin{tz}[zx,yscale=0.8,xscale=-1]
\draw[string] (0,-0.5) to node[front,zxnode=\zxwhite,pos=0.5] {$f$} (0,5);\draw[string] (\d,-0.5) to node[front,zxnode=\zxwhite, pos=0.5] {$f$} (\d,5);
\draw[string] (0, 0.5) to (\d,0.5);
\node[zxvertex=\zxwhite] at (0,0.5){};
\node[zxvertex=\zxwhite] at (\d,0.5){};
\node[zxnode=\zxwhite] at (\d/2, 0.5 ) {$G$};
\draw[string] (0, 4) to (\d,4);
\node[zxvertex=\zxwhite] at (0,4){};
\node[zxvertex=\zxwhite] at (\d,4){};
\node[zxnode=\zxwhite] at (\d/2, 4 ) {$H$};
\end{tz}
\end{calign}\\[-15pt]
\end{minipage}
\hspace{\bnd}
&
\hspace{\bndm}
\begin{minipage}{\wr\textwidth}
\vspace{0.2cm}
\hspace{-0.36cm}quantum function $+$
\vspace{-0.2cm}
\def\d{2}
\begin{calign}\nonumber
\begin{tz}[zx,yscale=0.8,xscale=-1]
\draw[string] (0,-0.5) to node[front,zxnode=\zxwhite,pos=0.57] {$P$} (0,5);\draw[string] (\d,-0.5) to node[front,zxnode=\zxwhite, pos=0.43] {$P$} (\d,5);
\draw[string] (0, 0.5) to (\d,0.5);
\node[zxvertex=\zxwhite] at (0,0.5){};
\node[zxvertex=\zxwhite] at (\d,0.5){};
\node[zxnode=\zxwhite] at (\d/2, 0.5 ) {$G$};
\draw[string, arrow data={0.08}{>},arrow data={0.5}{>}, arrow data ={0.95}{>}] (3, -0.5)  to (3,0.5) to [out=up, in=down] (-1.,4) to (-1.,5);
\end{tz}
~=~
\begin{tz}[zx,yscale=0.8,xscale=-1]
\draw[string] (0,-0.5) to node[front,zxnode=\zxwhite,pos=0.57] {$P$} (0,5);\draw[string] (\d,-0.5) to node[front,zxnode=\zxwhite, pos=0.43] {$P$} (\d,5);
\draw[string] (0, 0.5) to (\d,0.5);
\node[zxvertex=\zxwhite] at (0,0.5){};
\node[zxvertex=\zxwhite] at (\d,0.5){};
\node[zxnode=\zxwhite] at (\d/2, 0.5 ) {$G$};
\draw[string] (0, 4) to (\d,4);
\node[zxvertex=\zxwhite] at (0,4){};
\node[zxvertex=\zxwhite] at (\d,4){};
\node[zxnode=\zxwhite] at (\d/2, 4 ) {$H$};
\draw[string, arrow data={0.08}{>}, arrow data ={0.5}{>},arrow data ={0.95}{>}] (3, -0.5)  to (3,0.5) to [out=up, in=down] (-1.,4) to (-1.,5);
\end{tz}
\end{calign}\\[-15pt]
\end{minipage}
\hspace{\bndr}
&
\begin{minipage}{\wtext}\center quantum \\ graph\\homom. \end{minipage}\\
\hline\nonumber
\begin{minipage}{1.5cm}\center graph \\iso.\end{minipage} &
\begin{minipage}{\wl\textwidth}
\vspace{0.2cm}
\hspace{-0.0cm}bijection $+$
\vspace{-0.1cm}
\begin{calign}\nonumber
\begin{tz}[zx,every to/.style={out=up, in=down},scale=1,xscale=-1]
\draw (0,0) to (0,3);
\path (1.75,0) to (1.75,0.75) to  [looseness=1.1, in looseness=0.9]  node[zxnode=\zxwhite, pos=0.5] {$P$}(-1.75,3);
\node[zxnode=\zxwhite] at (0,0.9){$G$};
\end{tz}
=
\begin{tz}[zx,every to/.style={out=up, in=down},scale=1,xscale=-1]
\draw (0,0) to (0,3);
\path (1.75,0) to (1.75,0.75) to  [looseness=1.1, in looseness=0.9]  node[zxnode=\zxwhite, pos=0.5] {$H$}(-1.75,3);
\node[zxnode=\zxwhite] at (0,0.9){$P$};
\end{tz}
\end{calign}\\[-25pt]
\end{minipage}
\hspace{\bnd}
&
\hspace{\bndm}
\begin{minipage}{\wr\textwidth}
\vspace{0.2cm}
\hspace{-0.35cm}quantum bijection $+$
\vspace{-0.1cm}
\begin{calign}\nonumber
\begin{tz}[zx,every to/.style={out=up, in=down},scale=1,xscale=-1]
\draw (0,0) to (0,3);
\draw[arrow data={0.2}{>}, arrow data={0.9}{>}] (1.75,0) to (1.75,0.75) to  [looseness=1.1, in looseness=0.9]  node[zxnode=\zxwhite, pos=0.5] {$P$}(-1.75,3);
\node[zxnode=\zxwhite] at (0,0.9){$G$};
\end{tz}
\quad =\quad 
\begin{tz}[zx,every to/.style={out=up, in=down},scale=1,xscale=-1]
\draw (0,0) to (0,3);
\draw[arrow data={0.2}{>}, arrow data={0.8}{>}] (1.75,0) to [looseness=0.9] node[zxnode=\zxwhite, pos=0.5] {$P$} (-1.75,2.5) to (-1.75,3);
\node[zxnode=\zxwhite] at (0,2.35) {$H$};
\end{tz}
\end{calign}\\[-25pt]
\end{minipage}
\hspace{\bndr}
&
\begin{minipage}{\wtext}\begin{center} quantum \\ graph\\iso.\end{center} \end{minipage}\\ \nonumber&\nonumber&&\\ \hline\end{array}
\]
\vspace{-10pt}
\caption{Some classical concepts and their quantum analogues.}\label{fig:classicalquantum}
\end{figure}

\section{Background}
\label{sec:background}

\subsection{Gelfand duality}

Our approach to defining a theory of quantum sets and quantum functions is based on \emph{Gelfand duality}.\begin{theorem}[Gelfand duality]\label{thm:gelfandduality} The category of commutative $C^*$-algebras and \mbox{$*$-homo}-morphisms is equivalent to the opposite of the category of compact Hausdorff spaces and continuous functions.
\end{theorem}
\noindent The equivalence takes a space to the algebra of continuous complex-valued functions on that space; in the other direction, a $C^*$-algebra is taken to its spectrum. In some sense, therefore, one may consider the theory of compact Hausdorff spaces --- classical topology --- to be the theory of commutative $C^*$-algebras.

The idea of noncommutative geometry is to consider \emph{noncommutative} \mbox{$C^*$-algebras} in light of Gelfand duality, thus enabling us to study the `noncommutative' or `quantum' spaces to which they would be dual. In this paper, we consider only finite discrete spaces and finite-dimensional $C^*$-algebras. In this case, Theorem \ref{thm:gelfandduality} reduces to the following statement.
\begin{corollary}[Finite Gelfand duality]\label{cor:finiteGelfandduality} The category of finite-dimensional commutative $C^*$-algebras and $*$-homomorphisms is equivalent to the opposite of the category of finite sets and functions.
\end{corollary}

\subsection{The string diagram calculus}

In order to investigate finite Gelfand duality and its noncommutative generalisation, we make use of the string diagram calculus for monoidal categories. This calculus is well established and has been treated in detail elsewhere~\cite{Joyal1991a,Joyal1991b,Selinger2010,Coecke2010}; here we only provide a brief and informal introduction. We remark that this calculus is quite general. Although we only consider the graphical calculus of \Hilb , the category of finite-dimensional Hilbert spaces and linear maps, most of what we prove holds in the general setting of dagger compact categories~\cite{Kelly1972,Kelly1980,Abramsky2004}.

In the string diagram calculus\ignore{ for the monoidal category \Hilb }, wires correspond to finite-dimensional Hilbert spaces and boxes correspond to linear maps.\ignore{The horizontal direction in the diagram corresponds to monoidal product and the vertical direction to composition.} We read diagrams from bottom to top.\ignore{ For example, for $f: V_1 \to V_2$ and  $g:V_2 \to V_3$, the  composition $g \circ f: V_1 \to V_3$ and monoidal product $f \otimes g: V_1 \otimes  V_2 \to V_2 \otimes V_3$ are depicted as follows:}
Composition and tensor product are depicted as follows:
\begin{calign}
\begin{tz}[zx]
\draw (0,0) to (0,3.5);
\node[zxnode=\zxwhite] at (0,1.05) {$f$};
\node[zxnode=\zxwhite] at (0,2.45) {$g$};
\node[dimension, right] at (0,0){$V_1$};
\node[dimension, right] at (0, 1.75) {$V_2$};
\node[dimension, right] at (0, 3.5) {$V_3$};
\end{tz}
&
\begin{tz}[zx]
\draw (0,0) to (0,3.5);
\draw (1.5,0) to (1.5,3.5);
\node[zxnode=\zxwhite] at (0,1.75) {$f$};
\node[zxnode=\zxwhite] at (1.5,1.75) {$g$};
\node[dimension, right] at (0,0){$V_1$};
\node[dimension, right] at (1.5, 0) {$V_3$};
\node[dimension, right] at (0, 3.5) {$V_2$};
\node[dimension, right] at (1.5, 3.5) {$V_4$};
\end{tz}\\\nonumber
gf:V_1\to V_3
& 
f\otimes g: V_1\otimes V_3\to V_2\otimes V_4
\end{calign}
\ignore{
\begin{center}
\resizebox{!}{2.2cm}{
\input{diagcalcmonoidalcats.tex}}
\end{center}}
\ignore{The one-dimensional Hilbert space $\mathbb{C}$ is not drawn, since $\mathbb{C} \otimes V \simeq V$ for all Hilbert spaces $V$.}
\ignore{
There are three special families of linear maps which we depict topologically. Firstly, there is the swap map $\sigma: V \otimes W \to W \otimes V$, where $\sigma(v,w) = (w,v)$. We depict this as a crossing of wires:
\begin{equation}
\begin{tz}[zx,every to/.style={out=up, in=down}]
\draw (0,0) to (2,3);
\draw (2,0) to (0,3);
\node[dimension, right] at (2,0) {$W$};
\node[dimension,left] at (0,0) {$V$};
\node[dimension, right] at (2,3) {$V$};
\node[dimension, left] at (0,3) {$W$};
\end{tz}
\end{equation}%
\ignore{
\begin{center}
\resizebox{!}{2.2cm}{
\input{swapdiag.tikz}}
\end{center}}}
All finite-dimensional Hilbert spaces $V$ have dual spaces $V^*=\Hom(V,\mathbb{C})$, represented in the graphical calculus as an oriented wire with the opposite orientation as $V$. Duality is characterized by the following linear maps, here called \textit{cups and caps}:
\def\pv{\vphantom{V^*}}
\begin{calign}\label{eq:cupscapsHilb}
\begin{tz}[zx]
\draw[arrow data ={0.15}{<}, arrow data={0.89}{<}] (0,0) to [out=up, in=up, looseness=2.5] (2,0) ;
\node[dimension, right] at (2.05,0) {$\pv V$};
\node[dimension, left] at (0,0) {$V^*$};
\end{tz}
&
\begin{tz}[zx,scale=-1]
\draw[arrow data ={0.15}{>}, arrow data={0.89}{>}] (0,0) to [out=up, in=up, looseness=2.5] (2,0) ;
\node[dimension, left] at (2.0,0) {$\pv V$};
\node[dimension, right] at (0,0) {$V^*$};
\end{tz}
&
\begin{tz}[zx,xscale=-1]
\draw[arrow data ={0.15}{<}, arrow data={0.89}{<}] (0,0) to [out=up, in=up, looseness=2.5] (2,0) ;
\node[dimension, left] at (2.0,0) {$\pv V$};
\node[dimension, right] at (0,0) {$V^*$};
\end{tz}
&
\begin{tz}[zx,yscale=-1]
\draw[arrow data ={0.15}{>}, arrow data={0.89}{>}] (0,0) to [out=up, in=up, looseness=2.5] (2,0) ;
\node[dimension, right] at (2.05,0) {$\pv V$};
\node[dimension, left] at (0,0) {$V^*$};
\end{tz}\\\nonumber
f\otimes v \mapsto f(v) 
& 
~1\mapsto \mathbbm{1}_V
&
v\otimes f\mapsto f(v)
&
~~1\mapsto \mathbbm{1}_V
\end{calign}
To define the second and fourth map, we have identified $V\otimes V^* \cong V^*\otimes V \cong\End(V)$. It may be verified that these maps fulfill the following \textit{snake equations}:
\begin{calign}\label{eq:snake}
\begin{tz}[zx]
\draw[arrow data={0.15}{>}, arrow data={0.5}{>}, arrow data={0.9}{>}] (0,0) to (0,1) to [out=up, in=up, looseness=2] (1,1) to [out=down, in=down, looseness=2] (2,1) to (2,2);
\end{tz}
~~~=~~~
\begin{tz}[zx]
\draw[arrow data={0.5}{>}] (0,0) to (0,2);
\end{tz}
~~~= ~~~
\begin{tz}[zx,xscale=-1]
\draw[arrow data={0.15}{>}, arrow data={0.5}{>}, arrow data={0.9}{>}] (0,0) to (0,1) to [out=up, in=up, looseness=2] (1,1) to [out=down, in=down, looseness=2] (2,1) to (2,2);
\end{tz}
&
\begin{tz}[zx]
\draw[arrow data={0.15}{<}, arrow data={0.5}{<}, arrow data={0.9}{<}] (0,0) to (0,1) to [out=up, in=up, looseness=2] (1,1) to [out=down, in=down, looseness=2] (2,1) to (2,2);
\end{tz}
~~~ =~~~
\begin{tz}[zx]
\draw[arrow data={0.5}{<}] (0,0) to (0,2);
\end{tz}
~~~= ~~~
\begin{tz}[zx,xscale=-1]
\draw[arrow data={0.15}{<}, arrow data={0.5}{<}, arrow data={0.9}{<}] (0,0) to (0,1) to [out=up, in=up, looseness=2] (1,1) to [out=down, in=down, looseness=2] (2,1) to (2,2);
\end{tz}
\end{calign}
Together with the swap map $\sigma_{V,W}:\ignore{V\otimes W\to W\otimes V$,~$}v\otimes w\mapsto w\otimes v$, depicted as a crossing of wires, this leads to an extremely flexible topological calculus, allowing us to untangle arbitrary diagrams and straighten out any twists:
\begin{calign}
\begin{tz}[zx,scale=1]
	\begin{pgfonlayer}{nodelayer}
		\node [style=none] (0) at (0, -0) {};
		\node [style=none] (1) at (0, 3) {};
		\node [style=none] (2) at (1, -0) {};
		\node [style=none] (3) at (1, 3) {};
		\node [style=none] (4) at (2, -0) {};
		\node [style=none] (5) at (2, 3) {};
		\node [style=none] (6) at (3, -0) {};
		\node [style=none] (7) at (3, 3) {};
		\node [style=none] (8) at (2.25, 0.75) {};
		\node [style=none] (9) at (1.25, 1.25) {};
		\node [style=none] (10) at (3, 2) {};
		\node [style=none] (11) at (1, 2) {};
		\node [style=none] (12) at (1.75, 1.25) {};
		\node [style=none] (13) at (0.75, 0.5) {};
		\node [style=none] (14) at (0, 2) {};
		\node [style=none] (15) at (2, -0) {};
		\node [style=none] (16) at (2.25, 2) {};
		\node [style=none] (17) at (4, 1.5) {$=$};
		\node [style=none] (18) at (6, -0) {};
		\node [style=none] (19) at (7, -0) {};
		\node [style=none] (20) at (8, -0) {};
		\node [style=none] (21) at (5, -0) {};
		\node [style=none] (22) at (7, -0) {};
		\node [style=none] (23) at (7, 3) {};
		\node [style=none] (24) at (8, 3) {};
		\node [style=none] (25) at (6, 3) {};
		\node [style=none] (26) at (5, 3) {};
	\end{pgfonlayer}
	\begin{pgfonlayer}{edgelayer}
		\draw [thick, in=-90, out=27, looseness=0.75] (0.center) to (8.center);
		\draw [thick, in=-45, out=105, looseness=0.75] (8.center) to (9.center);
		\draw [thick, in=-90, out=105, looseness=0.75] (9.center) to (10.center);
		\draw [thick, in=-90, out=105, looseness=0.50] (10.center) to (1.center);
		\draw [thick, bend right=15, looseness=1.00] (3.center) to (11.center);
		\draw [thick, in=90, out=-90, looseness=0.75] (11.center) to (12.center);
		\draw [thick, in=-90, out=90, looseness=0.50] (13.center) to (12.center);
		\draw [thick, bend right, looseness=1.00] (13.center) to (2.center);
		\draw [thick, in=90, out=-153, looseness=0.75] (5.center) to (14.center);
		\draw [thick, in=135, out=-90, looseness=1.50] (14.center) to (4.center);
		\draw [thick, in=-127, out=90, looseness=1.50] (16.center) to (7.center);
		\draw [thick, in=111, out=-90, looseness=1.25] (16.center) to (6.center);
		\draw [style=simple] (21.center) to (26.center);
		\draw [style=simple] (18.center) to (25.center);
		\draw [style=simple] (19.center) to (23.center);
		\draw [style=simple] (20.center) to (24.center);
	\end{pgfonlayer}
\end{tz}
&
\begin{tz}[zx]
\draw[arrow data ={0.075}{>}, arrow data ={0.5}{>}, arrow data={0.925}{>}] (0,0) to [out=up, in=right] (-1.25, 2.25) to [out=left, in=left, looseness=1] (-1.25,0.75) to [out=right, in=down] (0,3);
\end{tz}
~~~=~~~\begin{tz}[zx]
\draw[arrow data={0.5}{>}] (0,0) to (0,3);
\end{tz}
~~~=~~~
\begin{tz}[zx,xscale=-1]
\draw[arrow data ={0.075}{>}, arrow data ={0.5}{>}, arrow data={0.925}{>}] (0,0) to [out=up, in=right] (-1.25, 2.25) to [out=left, in=left, looseness=1] (-1.25,0.75) to [out=right, in=down] (0,3);
\end{tz}%
\ignore{
\input{unwindtwists.tikz}}
\end{calign}
\ignore{We will refer to linear maps of the form $\C \to V$ and $V \to\C$ for a Hilbert space $V$ as \emph{states} and \emph{effects}, where the more traditional definition of  a state as a vector $\ket{\psi} \in V$ is retrieved as $\ket{\psi} = \psi(1)$.}
Given a linear map $f:V\to W$ between Hilbert spaces, we can express its adjoint \mbox{$f^\dagger:W \to V$} as a reflection of the corresponding diagram across a horizontal axis. This is justified, since the following holds:
\begin{calign}
\left(~~\begin{tz}[zx]
\draw[arrow data ={0.15}{<}, arrow data={0.89}{<}] (0,0) to [out=up, in=up, looseness=2.5] (2,0) ;
\end{tz}~~\right)^\dagger ~~~=~~~
\begin{tz}[zx,yscale=-1]
\draw[arrow data ={0.15}{>}, arrow data={0.89}{>}] (0,0) to [out=up, in=up, looseness=2.5] (2,0) ;
\end{tz}
&
\left(~~\begin{tz}[zx,yscale=-1]
\draw[arrow data ={0.15}{>}, arrow data={0.89}{>}] (0,0) to [out=up, in=up, looseness=2.5] (2,0) ;
\end{tz}~~\right)^\dagger
~~~=~~~
\begin{tz}[zx,xscale=1]
\draw[arrow data ={0.15}{<}, arrow data={0.89}{<}] (0,0) to [out=up, in=up, looseness=2.5] (2,0) ;
\end{tz}
\end{calign}
The following generalisation will be important in what follows.
\begin{definition}\label{def:daggerduality}Let $V$ and $W$ be Hilbert spaces. A \emph{dagger duality} between $V$ and $W$ is given by linear maps $\epsilon: W\otimes V \to \mathbb{C}$ and $\eta: \mathbb{C} \to V\otimes W$ fulfilling the snake equations~\eqref{eq:snake} and such that the following holds:
\begin{equation}\label{eq:daggerduality}
\begin{tz}[zx]
\draw (0,0) to (0,1.25);
\draw (1.5,0) to (1.5,1.25);
\draw (-0.2,0) rectangle (1.7, -0.8);
\node[scale=0.8] at (0.75,-0.4) {$\epsilon^\dagger$};
\node[dimension, left] at (0,1.25){$W$};
\node[dimension, right] at (1.5,1.25) {$V$};
\end{tz}
~~~=~~~
\begin{tz}[zx]
\draw (0,0) to [out=up, in=down] (1.5,1.25);
\draw (1.5,0) to [out=up, in=down](0,1.25);
\draw (-0.2,0) rectangle (1.7, -0.8);
\node[scale=0.8] at (0.75,-0.4) {$\eta$};
\node[dimension, left] at (0,1.25){$W$};
\node[dimension, right] at (1.5,1.25) {$V$};
\end{tz}
\end{equation}
\end{definition}
\noindent
The cups and caps defined in \eqref{eq:cupscapsHilb} are dagger duals. Dagger dualities are unique up to a unique unitary map~\cite[Section 7]{Selinger2007}, meaning that if $V$ and $W$ are dagger dual, then there is a unitary map $U: V^*\to W$ such that the following holds:
\begin{calign}
\begin{tz}[zx]
\clip (-0.6, 1.5) rectangle (2.1, -1.1);
\draw (0,0) to (0,1.25);
\draw (1.5,0) to (1.5,1.25);
\draw (-0.2,0) rectangle (1.7, -0.8);
\node[scale=0.8] at (0.75,-0.4) {$\eta$};
\node[dimension, left] at (0,1.25){$V$};
\node[dimension, right] at (1.5,1.25) {$W$};
\end{tz}
~~~=~~~
\begin{tz}[zx]
\clip (-0.6, 1.5) rectangle (2.1, -1.1);
\draw[arrow data ={0.2}{<}, arrow data={0.7}{<}] (0,1.25) to (0,0.25) to [out=down, in=down, looseness=2.5] (1.5,0.25) to (1.5,1.25);
\node[zxnode=\zxwhite] at (1.5, 0.5) {$U$};
\node[dimension, left] at (0,1.25){$V$};
\node[dimension, right] at (1.5,1.25) {$W$};
\end{tz}
&
\begin{tz}[zx,yscale=-1]
\clip (-0.6, 1.5) rectangle (2.1, -1.1);
\draw (0,0) to (0,1.25);
\draw (1.5,0) to (1.5,1.25);
\draw (-0.2,0) rectangle (1.7, -0.8);
\node[scale=0.8] at (0.75,-0.4) {$\epsilon$};
\node[dimension, left] at (0,1.25){$W$};
\node[dimension, right] at (1.5,1.25) {$V$};
\end{tz}
~~~=~~~
\begin{tz}[zx,yscale=-1]
\clip (-0.6, 1.5) rectangle (2.1, -1.1);
\draw[arrow data ={0.35}{<}, arrow data={0.8}{<}] (0,1.25) to (0,0.25) to [out=down, in=down, looseness=2.5] (1.5,0.25) to (1.5,1.25);
\node[zxnode=\zxwhite] at (0, 0.5) {$U^\dagger$};
\node[dimension, left] at (0,1.25){$W$};
\node[dimension, right] at (1.5,1.25) {$V$};
\end{tz}
\end{calign}

\ignore{
The final element of the graphical calculus for \Hilb\ is a depiction of the adjoint (conjugate transpose) of a linear map. In the diagrammatic notation, the adjoint is given by reflection of the diagram in a horizontal axis , since $(g \circ f)^{\dagger} = f^{\dagger} \circ g^{\dagger}$ and the following equalities hold:
\begin{center}
\resizebox{!}{1.5cm}{\input{adjointequalities.tikz}}
\end{center}}

\subsection{Diagrams for $C^*$-algebras}\label{sec:diagramsforC*}
Following Vicary~\cite{Vicary2010}, we define $C^*$-algebras in a categorical manner, as dagger Frobenius algebras in the category $\Hilb$.\ignore{ In particular, in all our definitions we refer to Hilbert spaces rather than vector spaces.} Of course, algebras can be defined on any vector space, but in order to discuss $C^*$-algebras we require the inner product. 

We will refrain from drawing an orientation on the wire corresponding to the Hilbert space on which the algebra is defined, for reasons which will soon become apparent.
\ignore{
On the wire corresponding to the Hilbert space on which the algebra is defined we will refrain from drawing arrows, for reasons which will soon become apparent. \ignore{The invisible arrow should always be taken as pointing upwards.}}

\begin{definition}\label{def:algebra}An \textit{algebra} is a Hilbert space $H$ with a multiplication and a unit map, depicted as follows:%
\begin{calign}
\begin{tz}[zx,master]
\coordinate (A) at (0,0);
\draw (0.75,1) to (0.75,2);
\mult{A}{1.5}{1}
\end{tz}
&
\begin{tz}[zx,slave]
\coordinate (A) at (0.75,2);
\unit{A}{1}
\end{tz}
\\[0pt]\nonumber
m:H\otimes H \to H& u: \mathbb{C} \to H 
\end{calign}\hspace{-0.2cm}
These maps  satisfy the following associativity and unitality equations:
\begin{calign}\label{eq:assocandunitality}
\begin{tz}[zx]
\coordinate(A) at (0.25,0);
\draw (1,1) to [out=up, in=-135] (1.75,2);
\draw (1.75,2) to [out=-45, in=up] (3.25,0);
\draw (1.75,2) to (1.75,3);
\mult{A}{1.5}{1}
\node[zxvertex=\zxwhite,zxdown] at (1.75,2){};
\end{tz}
\quad = \quad
\begin{tz}[zx,xscale=-1]
\coordinate(A) at (0.25,0);
\draw (1,1) to [out=up, in=-135] (1.75,2);
\draw (1.75,2) to [out=-45, in=up] (3.25,0);
\draw (1.75,2) to (1.75,3);
\mult{A}{1.5}{1}
\node[zxvertex=\zxwhite,zxdown] at (1.75,2){};
\end{tz}
&
\begin{tz}[zx]
\coordinate (A) at (0,0);
\draw (0,-0.25) to (0,0);
\draw (0.75,1) to (0.75,2);
\mult{A}{1.5}{1}
\node[zxvertex=\zxwhite,zxdown] at (1.5,0){};
\end{tz}
\quad =\quad
\begin{tz}[zx]
\draw (0,0) to (0,2);
\end{tz}
\quad= \quad
\begin{tz}[zx,xscale=-1]
\coordinate (A) at (0,0);
\draw (0,-0.25) to (0,0);
\draw (0.75,1) to (0.75,2);
\mult{A}{1.5}{1}
\node[zxvertex=\zxwhite,zxdown] at (1.5,0){};
\end{tz}
\end{calign}
Analogously, a \textit{coalgebra} is a Hilbert space $H$ with a coassociative comultiplication $\delta: H \to H\otimes H$ and a counit $\epsilon:H\to \mathbb{C}$. The adjoint of an algebra is a coalgebra.
\end{definition}
\noindent
Note that for the multiplication and unit maps of an algebra we simply draw white nodes rather than labelled boxes, for concision. Likewise, we draw the comultiplication and counit maps of the adjoint coalgebra as white nodes. Despite having the same label in the diagram, they can be easily distinguished by their type.

\begin{definition} A \textit{dagger Frobenius algebra} is an algebra where the algebra and adjoint coalgebra structures are related by the following Frobenius equation:
\begin{equation}\label{eq:Frobenius}
\begin{tz}[zx]
\draw (0,0) to [out=up, in=-135] (0.75,2) to (0.75,3);
\draw (0.75,2) to [out=-45, in=135] (2.25,1);
\draw (2.25,0) to (2.25,1) to [out=45, in=down] (3,3);
\node[zxvertex=\zxwhite,zxup] at (2.25,1){};
\node[zxvertex=\zxwhite,zxdown] at (0.75,2){};
\end{tz}
\quad = \quad
\begin{tz}[zx]
\coordinate (A) at (0,0);
\coordinate (B) at (0,3);
\draw (0.75,1) to (0.75,2);
\mult{A}{1.5}{1}
\comult{B}{1.5}{1}
\end{tz}
\quad = \quad 
\begin{tz}[zx]
\draw (0,0) to [out=up, in=-135] (0.75,2) to (0.75,3);
\draw (0.75,2) to [out=-45, in=135] (2.25,1);
\draw (2.25,0) to (2.25,1) to [out=45, in=down] (3,3);
\node[zxvertex=\zxwhite,zxup] at (2.25,1){};
\node[zxvertex=\zxwhite,zxdown] at (0.75,2){};
\end{tz}
\end{equation}
A Frobenius algebra is \textit{special}, \textit{symmetric} or \textit{commutative} if one of the following additional equations holds:
\begin{calign}\label{eq:special}
\begin{tz}[zx,every to/.style={out=up, in=down}]\draw (0,0) to (0,1) to [out=135] (-0.75,2) to [in=-135] (0,3) to (0,4);
\draw (0,1) to [out=45] (0.75,2) to [in=-45] (0,3);
\node[zxvertex=\zxwhite, zxup] at (0,1){};
\node[zxvertex=\zxwhite,zxdown] at (0,3){};\end{tz}
\quad = \quad
\begin{tz}[zx]
\draw (0,0) to +(0,4);
\end{tz}
&
\begin{tz}[zx,every to/.style={out=up, in=down}]
\draw (0.25,-1.5) to (1.75,0) to [in=-45] (1,1) to (1,1.5);
\draw (1.75,-1.5) to (0.25,0) to [in=-135] (1,1);
\node[zxvertex=\zxwhite, zxdown] at (1,1){};
\node[zxvertex=\zxwhite] at (1,1.5){};
\end{tz}
\quad = \quad
\begin{tz}[zx,every to/.style={out=up, in=down}]
\draw (1.75,-1.5) to (1.75,0) to [in=-45] (1,1) to (1,1.5);
\draw (0.25,-1.5) to (0.25,0) to [in=-135] (1,1);
\node[zxvertex=\zxwhite, zxdown] at (1,1){};
\node[zxvertex=\zxwhite] at (1,1.5){};
\end{tz}
&
\begin{tz}[zx,every to/.style={out=up, in=down}]
\draw (0.25,-1.5) to (1.75,0) to [in=-45] (1,1) to (1,2);
\draw (1.75,-1.5) to (0.25,0) to [in=-135] (1,1);
\node[zxvertex=\zxwhite, zxdown] at (1,1){};
\end{tz}
\quad = \quad
\begin{tz}[zx,every to/.style={out=up, in=down}]
\draw (1.75,-1.5) to (1.75,0) to [in=-45] (1,1) to (1,2);
\draw (0.25,-1.5) to (0.25,0) to [in=-135] (1,1);
\node[zxvertex=\zxwhite, zxdown] at (1,1){};
\end{tz}\\\nonumber
\text{a) special}&\text{b) symmetric}&\text{c) commutative}
\end{calign} 
\end{definition}
\noindent
In this work, we will focus on special symmetric and special commutative dagger Frobenius algebras, which we abbreviate as \F s and \CF s, respectively. 
\ignore{\begin{remark} Special symmetric dagger Frobenius algebras are also known as Q-systems \DRcomm{expand on this with some refs}
\end{remark}}%

Frobenius algebras are closely related to dualities. In particular, it is a direct consequence of~\eqref{eq:assocandunitality} and~\eqref{eq:Frobenius} that the following cups and caps fulfill the snake equations~\eqref{eq:snake}:
\begin{calign}\label{eq:cupcapfrob}
\begin{tz}[zx]
\clip (-0.1,0) rectangle (2.1,2.);
\draw (0,0) to [out=up, in=up, looseness=2] node[zxvertex=\zxwhite, pos=0.5]{} (2,0);
\end{tz}
~:=~
\begin{tz}[zx]
\clip (-0.1,0) rectangle (2.1,2.2);
\draw (0,0) to [out=up, in=up, looseness=2] node[front,zxvertex=\zxwhite, pos=0.5](A){} (2,0);
\draw[string] (A.center) to (1,1.8);
\node[zxvertex=\zxwhite] at (1,1.8) {};
\end{tz}
&
\begin{tz}[zx,yscale=-1]
\clip (-0.1,0) rectangle (2.1,2.);
\draw (0,0) to [out=up, in=up, looseness=2] node[zxvertex=\zxwhite, pos=0.5]{} (2,0);
\end{tz}
~:=~
\begin{tz}[zx,yscale=-1]
\clip (-0.1,0) rectangle (2.1,2.2);
\draw (0,0) to [out=up, in=up, looseness=2] node[front,zxvertex=\zxwhite, pos=0.5](A){} (2,0);
\draw[string] (A.center) to (1,1.8);
\node[zxvertex=\zxwhite] at (1,1.8) {};
\end{tz}
\end{calign}
\ignore{%
\begin{proposition}
Let $(H, \mu, \eta, \delta, \epsilon)$ be an \F. Let the `cup' and `cap' be defined as follows:
\begin{equation}\label{eqn:cupcapfrob}
\resizebox{!}{1.5cm}{\input{cupcapfrob.tikz}}
\end{equation}
Then the following `snake' equations hold: 
\begin{equation}
\resizebox{!}{1.5cm}{\input{frobsnake.tikz}}
\end{equation}
\end{proposition}
\begin{proof}
Can easily be shown using the equations~\eqref{eq:assocandunitality} and  equations (\ref{eq:Frobenius}).
\end{proof}}%
It follows that every Frobenius algebra is canonically self-dual, $A^*\cong A$; we therefore do not need to draw an orientation on the corresponding wire.
\ignore{
\begin{remark}\label{rem:noorientationonfrobwire}
Categorically, this means the cup and the cap defined in (\ref{eqn:cupcapfrob}) witness a self-duality $H \simeq H^*$. For this reason, we do not draw the arrows on $H$, since we assume the categorical dual on $H$ to be defined using the Frobenius structure, which means that $H$ is its own dual.
\end{remark}}%

A major reason for defining these structures is the fact that $\F$s coincide with finite-dimensional $C^*$-algebras.
\begin{theorem}[{\cite[Theorem 4.6.]{Vicary2010}}] Every finite-dimensional $C^*$-algebra has an inner product making it into a special symmetric dagger Frobenius algebra. Conversely, every $\F$ $A$ admits a norm such that the canonical involution, defined by its action on vectors $\ket{a}\in A$ as the following antihomomorphism, endows it with the structure of a $C^*$-algebra:%
\begin{calign}
\begin{tz}[zx,master]
\draw (0.25,3) to (0.25,1);
\node[zxnode=\zxwhite] at (0.25,1) {$a$};
\end{tz}
\quad \mapsto \quad
\begin{tz}[zx,slave]
\draw (0.25,3) to  (0.25,2) to [out=down, in=135] (1,1) to [out=45, in=down] (1.75,2);
\draw (1,1) to (1,0.5);
\node[zxnode=\zxwhite] at (1.75,2) {$a^\dagger$};
\node[zxvertex=\zxwhite,zxup] at (1,1){};
\node[zxvertex=\zxwhite] at (1,0.5){};
\end{tz}\\[-5pt]\nonumber\end{calign} 
\end{theorem}
\ignore{
\begin{theorem}\label{thm:frobalgsarecstaralgs}
The special symmetric dagger Frobenius algebras in \Hilb\ may all be endowed with up to two canonical $C^*$-algebra structures (each corresponding to a choice of involution), and each $C^*$-algebra may be endowed with a canonical Frobenius algebra structure, such that this correspondence induces an equality between the category of \F s with chosen involution and the category of finite-dimensional $C^*$-algebras.
\end{theorem}
\begin{proof}
See Theorem 4.7 in the paper of Vicary~\cite{Vicary2010}.  In particular, every finite-dimensional $C^*$-algebra has an inner product making it into a Hilbert space. In the other direction, let $A$ be a $\F$. Then one may define a $C^*$ algebra structure on $A$, where the antihomomorphism $*:A\to A$ of the $C^*$-algebra $A$ is the following function, defined by its action on vectors $\ket{a}\in A$.
\end{proof}}

\ignore{An advantage of using 
Since the \F s in $\Hilb$ are exactly the $C^*$-algebras, it is a matter of preference which starting point one takes. Here we start with \F s,}
\noindent
One advantage of explicitly using $\F$s over $C^*$-algebras is that $\F$s already contain `up-front' all emergent structures of finite-dimensional $C^*$-algebras, such as the comultiplication $\Delta = m^\dagger : H\to H \otimes H$; they are therefore more amenable to the purely compositional reasoning of the graphical calculus. Notions from the theory of \mbox{$C^*$-algebras}, such as $*$-homomorphisms, can be reformulated in the language of \F s.\looseness=-2

\begin{definition}\label{def:starhoms}A \textit{$*$-homomorphism} between \F s $A$ and $B$ is a linear map $f:A\to B$ satisfying the following equations:
\begin{calign}\label{eq:homo}
\begin{tz}[zx, master, every to/.style={out=up, in=down},yscale=-1]
\draw (0,0) to (0,2) to [out=135] (-0.75,3);
\draw (0,2) to [out=45] (0.75, 3);
\node[zxnode=\zxwhite] at (0,1) {$f$};
\node[zxvertex=\zxwhite, zxdown] at (0,2) {};
\end{tz}
=
\begin{tz}[zx, every to/.style={out=up, in=down},yscale=-1]
\draw (0,0) to (0,0.75) to [out=135] (-0.75,1.75) to (-0.75,3);
\draw (0,0.75) to [out=45] (0.75, 1.75) to +(0,1.25);
\node[zxnode=\zxwhite] at (-0.75,2) {$f$};
\node[zxnode=\zxwhite] at (0.75,2) {$f$};
\node[zxvertex=\zxwhite, zxdown] at (0,0.75) {};
\end{tz}
&
\begin{tz}[zx,slave, every to/.style={out=up, in=down},yscale=-1]
\draw (0,0) to (0,2) ;
\node[zxnode=\zxwhite] at (0,1) {$f$};
\node[zxvertex=\zxwhite, zxup] at (0,2) {};
\end{tz}
=
\begin{tz}[zx,slave, every to/.style={out=up, in=down},yscale=-1]
\draw (0,0) to (0,0.75) ;
\node[zxvertex=\zxwhite, zxup] at (0,0.75) {};
\end{tz}
&
\begin{tz}[zx,slave, every to/.style={out=up, in=down},scale=-1]
\draw (0,0) to (0,3);
\node[zxnode=\zxwhite] at (0,1.5) {$f^\dagger$};
\end{tz}
=~~
\begin{tz}[zx,slave,every to/.style={out=up, in=down},scale=-1,xscale=-1]
\draw (0,1.5) to (0,2) to [in=left] node[pos=1] (r){} (0.5,2.5) to [out=right, in=up] (1,2)  to [out=down, in=up] (1,0);
\draw (-1,3) to [out=down,in=up] (-1,1) to [out=down, in=left] node[pos=1] (l){} (-0.5,0.5) to [out=right, in=down] (0,1) to (0,1.5);
\node[zxnode=\zxwhite] at (0,1.5) {$f$};
\node[zxvertex=\zxwhite] at (l.center){};
\node[zxvertex=\zxwhite] at (r.center){};
\end{tz}
\end{calign}
A \textit{$*$-cohomomorphism} is a linear map $f:A\to B$ satisfying the following equations:
\begin{calign}\label{eq:cohomo}
\begin{tz}[zx, master, every to/.style={out=up, in=down}]
\draw (0,0) to (0,2) to [out=135] (-0.75,3);
\draw (0,2) to [out=45] (0.75, 3);
\node[zxnode=\zxwhite] at (0,1) {$f$};
\node[zxvertex=\zxwhite, zxup] at (0,2) {};
\end{tz}
=
\begin{tz}[zx, every to/.style={out=up, in=down}]
\draw (0,0) to (0,0.75) to [out=135] (-0.75,1.75) to (-0.75,3);
\draw (0,0.75) to [out=45] (0.75, 1.75) to +(0,1.25);
\node[zxnode=\zxwhite] at (-0.75,2) {$f$};
\node[zxnode=\zxwhite] at (0.75,2) {$f$};
\node[zxvertex=\zxwhite, zxup] at (0,0.75) {};
\end{tz}
&
\begin{tz}[zx,slave, every to/.style={out=up, in=down}]
\draw (0,0) to (0,2) ;
\node[zxnode=\zxwhite] at (0,1) {$f$};
\node[zxvertex=\zxwhite, zxup] at (0,2) {};
\end{tz}
=
\begin{tz}[zx,slave, every to/.style={out=up, in=down}]
\draw (0,0) to (0,0.75) ;
\node[zxvertex=\zxwhite, zxup] at (0,0.75) {};
\end{tz}
&
\begin{tz}[zx,slave, every to/.style={out=up, in=down}]
\draw (0,0) to (0,3);
\node[zxnode=\zxwhite] at (0,1.5) {$f^\dagger$};
\end{tz}
=~~
\begin{tz}[zx,slave,every to/.style={out=up, in=down},xscale=-1]
\draw (0,1.5) to (0,2) to [in=left] node[pos=1] (r){} (0.5,2.5) to [out=right, in=up] (1,2)  to [out=down, in=up] (1,0);
\draw (-1,3) to [out=down,in=up] (-1,1) to [out=down, in=left] node[pos=1] (l){} (-0.5,0.5) to [out=right, in=down] (0,1) to (0,1.5);
\node[zxnode=\zxwhite] at (0,1.5) {$f$};
\node[zxvertex=\zxwhite] at (l.center){};
\node[zxvertex=\zxwhite] at (r.center){};
\end{tz}
\end{calign}
A \emph{$*$-isomorphism} is a linear map $f:A\to B$ which is both a $*$-homomorphism and a $*$-cohomomorphism. 
\end{definition}
\noindent
Observe that the adjoint of a $*$-homomorphism is a $*$-cohomomorphism, that every {$*$-isomorphism} is unitary, and that every unitary $*$-homomorphism of \F s is a \mbox{$*$-isomorphism.} In particular, a $*$-isomorphism is precisely an invertible *-homomorphism (see Proposition~\ref{prop:Frobeniusinvertible} for the converse). 

\begin{proposition}[{\cite[Theorem 4.7]{Vicary2010}}] The notion of a $*$-homomorphism between $\F$s coincides with the notion of a  $*$-homomorphism between finite-dimensional \mbox{$C^*$-algebras.}
\end{proposition}
\ignore{
We record the following characterisation of $*$-isomorphisms for later.
\begin{proposition}
Every $*$-isomorphism between \F s is unitary, and every unitary $*$-homomorphism of \F s is a $*$-isomorphism.
\end{proposition}
\begin{proof}
Follows from the fact that the coalgebra structure is the adjoint of the algebra structure. 
\end{proof}
}
\subsection{Finite-dimensional Gelfand duality in diagrams}\label{sec:gelfanddiagram}
Having established the graphical calculus and the correspondence between finite-\\dimensional $C^*$-algebras and $\F$s, we now recall the graphical version of finite-dimensional Gelfand duality in the framework established by Coecke, Pavlovi{\'c} and Vicary~\cite{Coecke2009}. We first observe that every orthonormal basis on a Hilbert space $H$ defines a special commutative dagger Frobenius algebra on $H$.

\begin{example} \label{exm:Frob}Let $\left\{\ket{i}\right\}_{1\leq i\leq n}$ be an orthonormal basis of a Hilbert space $H$. Then the following multiplication and unit maps, together with their adjoints, form a special commutative dagger Frobenius algebra on $H$:
\begin{calign}\label{eq:classicalcopy}\begin{tz}[zx,master, scale=1.3]
\coordinate (A) at (0,0);
\draw (0.75,1) to (0.75,2);
\mult{A}{1.5}{1}
\end{tz} := ~\sum_{i=1}^{n} ~~~
\begin{tz}[zx,slave,scale=1.3]
\draw (0,0) to (0,0.5);
\draw (1.5,0) to (1.5,0.5);
\draw (0.75,1.5) to (0.75,2);
\node[zxnode=\zxwhite] at (0,0.5) {$i^\dagger$};
\node[zxnode=\zxwhite] at (1.5,0.5) {$i^\dagger$};
\node[zxnode=\zxwhite] at (0.75,1.5) {$i$};
\end{tz}
&
\begin{tz}[zx,slave,scale=1.3]
\draw (0.75,1) to (0.75,2);
\node[zxvertex=\zxwhite] at (0.75,1){};
\end{tz} := ~\sum_{i=1}^{n} 
\begin{tz}[zx,slave,scale=1.3]
\draw (0.75,1.25) to (0.75,2);
\node[zxnode=\zxwhite] at (0.75,1.25) {$i$};
\end{tz}\\[5pt]\nonumber
m: \ket{i} \otimes \ket{j} \mapsto \delta_{i,j} \ket{i} 
&
u: 1 \mapsto \sum_{i=1}^{n} \ket{i}
\end{calign}
\end{example}
\noindent
Conversely, every special commutative dagger Frobenius algebra $A$ gives rise to an orthonormal basis of $A$; the basis vectors are given by the copyable element of $A$, defined as follows:

\begin{definition}\label{def:copyablestates} A \textit{copyable element} of a $\CF$ $A$ is a $*$-cohomomorphism $\psi: \mathbb{C} \to A$; that is, a vector $\ket{\psi} \in A$ such that the following hold\footnote{In \Hilb, the last equation on the right is redundant, as it follows from the other equations; we include it for completeness.}:
\begin{calign}\label{eq:ordinaryelement}\begin{tz}[zx, master, every to/.style={out=up, in=down}]
\draw (0,1) to (0,2) to [out=135] (-0.75,3);
\draw (0,2) to [out=45] (0.75, 3);
\node[zxnode=\zxwhite] at (0,1) {$\psi$};
\node[zxvertex=\zxwhite, zxup] at (0,2) {};
\end{tz}
=~~~
\begin{tz}[zx,slave, every to/.style={out=up, in=down}]
\draw (-0.75,2) to (-0.75,3);
\draw (0.75, 2) to +(0,1.);
\node[zxnode=\zxwhite] at (-0.75,2) {$\psi$};
\node[zxnode=\zxwhite] at (0.75,2) {$\psi$};
\end{tz}
&
\begin{tz}[zx, every to/.style={out=up, in=down}]
\draw (0,1) to (0,2) ;
\node[zxnode=\zxwhite] at (0,1) {$\psi$};
\node[zxvertex=\zxwhite, zxup] at (0,2) {};
\end{tz}
~~=~~
\emptydiagram
&
\begin{tz}[zx,slave, every to/.style={out=up, in=down}]
\draw (0,0) to (0,1.5);
\node[zxnode=\zxwhite] at (0,1.5) {$\psi^\dagger$};
\end{tz}
=~~~
\begin{tz}[zx,slave,every to/.style={out=up, in=down},xscale=-1]
\draw (0,1.5) to (0,2) to [in=left] node[pos=1] (r){} (0.5,2.5) to [out=right, in=up] (1,2)  to [out=down, in=up] (1,0);
\node[zxnode=\zxwhite] at (0,1.5) {$\psi$};
\node[zxvertex=\zxwhite,zxdown] at (r.center){};
\end{tz}\\\nonumber
\end{calign}
\end{definition}
\noindent 
These copyable elements indeed form an orthonormal basis of the Hilbert space underlying $A$.
\begin{theorem}[{\cite[Theorem 5.1.]{Coecke2009}}] \label{thm:classificationONB}The copyable elements of a special commutative dagger Frobenius algebra $A$ form an orthonormal basis of $A$ for which the algebra is of the form given in Example~\ref{exm:Frob}.
\end{theorem}
\noindent
In other words, every special dagger commutative Frobenius algebra is of the form~\eqref{eq:classicalcopy} for some orthonormal basis on a Hilbert space.

Given a $\CF$ $A$, we denote its set of copyable elements by $\widehat{A}$. For $\CF$s $A$ and $B$, it can easily be verified that every function $\widehat{A}\to \widehat{B}$ gives rise to a $*$-cohomomorphism between $A$ and $B$ and that conversely every $*$-cohomomorphism $A\to B$ comes from such a function $\widehat{A} \to \widehat{B}$. Therefore, Theorem~\ref{thm:classificationONB} gives rise to the following Frobenius algebraic version of finite Gelfand duality:

\begin{corollary}[{\cite[Corollary 7.2.]{Coecke2009}}]The category of commutative special dagger Frobenius algebras and $*$-cohomomorphisms\footnote{Here, we use $*$-cohomomorphisms instead of the more conventional $*$-homomorphisms of Corollary~\ref{cor:finiteGelfandduality} to obtain an equivalence with the category of finite sets and not with its opposite. This is a recurring theme throughout this work; we use copyable elements instead of the equivalent characters and comodules instead of modules.} is equivalent to the category of finite sets and functions.
\end{corollary}
\noindent
Explicitly, this equivalence maps a $\CF$ $A$ to its set of copyable elements $\widehat{A}$ and a set $X$ to the algebra associated to the orthonormal basis $\{ \ket{x}~|~x\in X\}$ of the Hilbert space $\mathbb{C}^{|X|}$.
%
%
%
Under this correspondence, we may therefore consider the category of finite sets as `contained within $\Hilb$' using the following identification.

\begin{center}
\begin{tabular}{l | l}
$\Set$ & $\mathrm{Hilb}$\\
\hline
set of cardinality $n$ & \CF\ of dimension $n$ \\
elements of the set & copyable states of the \CF \\
functions & $*$-cohomomorphisms\\
bijections & $*$-isomorphisms\\
the one element set $\{*\}$ & the one-dimensional \CF\ $\mathbb{C}$
\end{tabular}
\end{center}

\begin{terminology}\label{not:setalgebra}
Throughout this paper, we will take pairs of words in this table to be synonymous. In particular, we will denote a set and its corresponding commutative algebra by the same symbol. It will always be clear from context whether we refer to the set $X$ or the algebra $X$. 
\end{terminology}

\noindent
The fact that the category of finite sets and functions can be faithfully embedded into $\Hilb$ will be central to our quantisation procedure described in the introduction. We formalise this important concept by adopting terminology initially used by Freyd and developed by Ad{\'{a}}mek and others~\cite{Adamek1990,Freyd1970}.
\begin{definition}\label{def:Hilbcategory} A \emph{concrete dagger category} is a pair $(\mathcal{C}, F)$, where $\mathcal{C}$ is a dagger category and $F:\mathcal{C} \to \Hilb$ is a faithful dagger functor, which we refer to as the \emph{forgetful functor}. 
\end{definition}
\noindent
In other words, Gelfand duality allows us to treat the category of finite sets and functions as a concrete dagger category. 
\begin{philosophy}\label{phil:concretedagger}All categories in this paper are concrete dagger categories, with a forgetful functor specifying an underlying Hilbert space for every object and an underlying linear map for every morphism. The quantum mechanical interpretation of our categories depends on this forgetful functor. We will formulate all categorical concepts in terms of concrete dagger categories, and ensure that they are compatible with the forgetful functor to \Hilb.\end{philosophy}

\section{Quantum sets and quantum functions}\label{sec:quantumset}
The fundamental idea of noncommutative topology is to generalise the correspondence between spaces and commutative algebras by considering noncommutative algebras in light of Gelfand duality. We will now begin our exploration of the world of finite-dimensional noncommutative algebras, or `finite quantum sets'.

\begin{terminology}\ignore{In light of Gelfand duality, we refer to special commutative dagger Frobenius algebras as sets, and --- by analogy --- refer to special symmetric dagger Frobenius algebras as \textit{quantum sets}.
}
By analogy with Gelfand duality, we think of a special symmetric dagger Frobenius algebra as being associated to an imagined finite \emph{quantum set}, just as a commutative special dagger Frobenius algebra is associated to a finite set.
We follow Terminology~\ref{not:setalgebra} and Wang~\cite[page 3]{Wang1998} in denoting both the algebra and its associated imagined quantum set by the same symbol. \ignore{It will always be clear from context whether we refer to the algebra or the associated imagined geometric object.}
\end{terminology}

\subsection{Quantum elements}\label{sec:qelem}

A set is completely determined by its elements. What is the appropriate notion of a \emph{quantum element} of a quantum set? In particular, is there a notion of quantum element such that a quantum set is completely determined by its quantum elements? Copyable states cannot play this role, since it can easily be verified that matrix algebras have no copyable states whatsoever. Adopting the ideas outlined in the introduction, we make the following definition, `quantising' the notion of a copyable state~\eqref{eq:ordinaryelement}.

\begin{definition}\label{def:quantumelement}
A \textit{quantum element} of a quantum set $A$ is a pair ($H,Q$), where $H$ is a Hilbert space and $Q: H\to A \otimes H$ is a linear map satisfying the following equations:
\begin{calign}\label{eq:quantumelement}
\begin{tz}[zx,every to/.style={out=up, in=down},xscale=-1]
\draw (0,0) to (0,2) to [out=45] (0.75,3);
\draw (0,2) to [out=135] (-0.75,3);
\draw[arrow data={0.2}{>}, arrow data={0.8}{>}] (1.75,0) to [looseness=0.9] node[zxnode=\zxwhite, pos=0.5](Q) {$Q$} (-1.75,2.5) to (-1.75,3);
\node[zxvertex=\zxwhite, zxup] at (0,2){};
\draw[white,double] (0,0) to (Q.center);
\end{tz}
=
\begin{tz}[zx,every to/.style={out=up, in=down},xscale=-1]
\draw (0.75,1.75) to (0.75,3);
\draw  (-0.75,1.75) to (-0.75,3);
\draw[arrow data={0.2}{>}, arrow data={0.9}{>}] (1.75,0) to (1.75,0.75) to  [looseness=1.1, in looseness=0.9] node[zxnode=\zxwhite, pos=0.36] {$Q$} node[zxnode=\zxwhite, pos=0.64] {$Q$}(-1.75,3);
\end{tz}
&
\begin{tz}[zx,every to/.style={out=up, in=down},xscale=-1]
\draw (0,0) to (0,2.25);
\draw[arrow data={0.2}{>}, arrow data={0.8}{>}] (1,0) to [looseness=0.9] node[zxnode=\zxwhite, pos=0.5](Q) {$Q$} (-1,2.5) to (-1,3);
\node[zxvertex=\zxwhite] at (0,2.25){};
\draw[white,double] (0,0) to (Q.center);
\end{tz}
=
\begin{tz}[zx,every to/.style={out=up, in=down},xscale=-1]
\draw[arrow data={0.2}{>}, arrow data={0.9}{>}] (1.,0) to [looseness=0.9]  (-1,2.5) to (-1,3);
\end{tz}
&
\begin{tz}[zx,every to/.style={out=up, in=down},xscale=-0.8]
\draw [arrow data={0.2}{>},arrow data={0.8}{>}]  (0,0) to (2.25,3);
\draw (2.25,0) to [in=-45] node[zxnode=\zxwhite, pos=1] (Q) {$Q^\dagger$} (1.125,1.5);
\end{tz}
~=~~
\begin{tz}[zx,xscale=-0.6,yscale=-1]
\draw[arrow data={0.5}{<}] (0.25,-0.5) to (0.25,0) to [out=up, in=-135] (1,1);
\draw[arrow data={0.5}{>}] (1.75,2.5) to (1.75,2) to [out=down, in=45] (1,1);
\draw (1,1) to [out= -45, in= left] node[zxvertex=\zxwhite, pos=1] {} (2.3,0.3) to [out=right, in=down] (3.25,1) to (3.25,2.5);
\node [zxnode=\zxwhite] at (1,1) {$Q$};
\end{tz} 
\end{calign}
\end{definition}
\begin{philosophy}\label{phil:quantumprooffromclassical} The axioms defining a quantum element~\eqref{eq:quantumelement} look just like the axioms defining an ordinary element~\eqref{eq:ordinaryelement} with an additional oriented wire. This will be a guiding principle for the graphical calculus in this work. Many of the calculations presented here were derived in the `classical' setting with an additional oriented wire only added later. Here, and in what follows, we always draw this additional wire with an orientation, while we draw the original wire --- carrying a Frobenius algebra structure --- without orientation (see also the discussion after~\eqref{eq:cupcapfrob}).
\ignore{
 Because the original wire is a Frobenius algebra, we draw it without an orientation following Remark \ref{rem:noorientationonfrobwire}; whereas we draw the additional wire, which need not carry an algebra structure, with an orientation.}%
\end{philosophy}

\begin{remark} We have drawn the Hilbert space wires from the bottom left to the top right. This is just a convention, and we could equally have defined quantum elements --- and later quantum functions --- using the opposite convention.
\end{remark}
\noindent
We will show in Proposition~\ref{prop:reconstruct} that every quantum set is completely determined by its quantum elements, thereby justifying Definition~\ref{def:quantumelement}.

\begin{remark}\label{rem:operationalinterpretationofquantumels}
 If $X$ is an ordinary set (that is, a $\CF$), then it follows that a quantum element of $X$ is a projective measurement with outcomes in $X$ (see Corollary~\ref{cor:quantumelementmeasurement}). The first diagram of~\eqref{eq:quantumelement} corresponds to orthogonality and idempotency; the second to completeness; and the third to self-adjointness. This diagrammatic representation of projective measurements has been known at least since the work of Coecke and Pavlovi{\'c}~\cite{Coecke2007}.\! A direct operational intepretation of~\eqref{eq:quantumelement} has recently been \mbox{proposed~\cite{Coecke2017_2}.}
\end{remark}

\begin{remark}\label{rem:interpretationquantumelement}
A (non-probabilistic, discrete) classical observable can be thought of both as a process $x:\{*\} \to \{*\} \times X$ producing an element of a set $X$ from a trivial system --- that is, a process picking out one and only one element --- or simply as an element $x\in X$.\looseness=-2

Similarly, a quantum observable can either be thought of as projective measurement $P:H\to X\otimes H$ producing an element of $X$ from an underlying quantum mechanical system $H$ or simply as a \emph{quantum element}  $P\in_Q X$, shifting attention away from the underlying Hilbert space.
This is similar in spirit to random variables in probability theory which are defined as functions $x$ from an underlying probability space to a set $X$, but are usually thought of as random elements of this set $x\in_R X$. 

This perspective is reflected in our quantisation approach:
By appending a Hilbert space wire to a diagram defining some set-theoretic concept, we retain any operational interpretation, only now allowing the relevant processes to make use of an underlying quantum mechanical system. This is similar to quantisation in the sense of nonlocal games, where classical concepts are quantised by formulating them as strategies for multi-player games and allowing the use of an additional shared quantum resource~\cite{Atserias2016}. As we will see in Section~\ref{sec:quantumperm} and~\ref{sec:qgt}, our approach to quantisation indeed leads to the same concepts as those appearing in the study of such nonlocal games.

\end{remark}
\noindent
Due to the presence of the additional Hilbert space wire introduced in our quantisation procedure, we are led to consider maps on this Hilbert space, which interact in a particular way with quantum elements. 
\begin{definition} An \textit{intertwiner} of quantum elements $(H,Q) \to (H', Q')$ is a linear map $f:H\to H'$ such that the following holds:%
\begin{equation}
\begin{tz}[zx,every to/.style={out=up, in=down},xscale=-1]
\draw (1,2)to [out=45] (2,3);
\draw[arrow data={0.35}{>}] (2,0) to node[zxnode=\zxwhite, pos=0.9] {$f$} (2,1);
\draw[string,arrow data={0.9}{>},arrow data={0.26}{>}](2,1) to node[zxnode=\zxwhite, pos=0.5]{$Q'$}  (0,3);
\end{tz}
\quad = \quad
\begin{tz}[zx,every to/.style={out=up, in=down},scale=-1,xscale=-1]
\draw (0,0) to (0,1) to [in=-135] (1,2);
\draw[arrow data={0.35}{<}] (2,0) to node[zxnode=\zxwhite, pos=0.9] {$f$} (2,1);
\draw[string,arrow data={0.9}{<},arrow data={0.26}{<}](2,1) to node[zxnode=\zxwhite, pos=0.5]{$Q$}  (0,3);
\end{tz}
\end{equation}
\end{definition}
Note that, in the case of classical elements~\eqref{eq:ordinaryelement}, all linear maps $\C \to \C$ are scalars and commute trivially with the element, providing us with no information about its structure; intertwiners therefore only become relevant in the quantum setting. 

\begin{example}
Let $Q$ be a quantum element of a classical set $X$; in other words, a projective measurement with outcomes in $X$ (see Remark~\ref{rem:operationalinterpretationofquantumels}). In this case, examples of intertwiners $Q\to Q$ are given by projectors onto subspaces which are left undisturbed by the measurement. We will show in Section~\ref{sec:semisimple} that all intertwiners between quantum elements may be understood in this way.
\end{example}
\noindent
As we will see, much of the structural difference between the classical and the quantum setting can be understood in terms of the existence or non-existence of such intertwiners. While ordinary elements form a set, quantum elements should properly be organized into a category, to keep track of the intertwiners. 
\begin{definition}\label{def:intertwinerofelements}For a quantum set $A$, we define the category \QEl(A) of quantum elements of $A$:
\vspace{-0.1cm}
\begin{itemize} \item \textbf{objects} are quantum elements $(H,Q)$ of $A$;
\item \textbf{morphisms} $(H,Q) \to (H', Q')$ are intertwiners of quantum elements. 
\end{itemize}
Composition of intertwiners is ordinary composition of linear maps.
\end{definition}
\noindent Every category of quantum elements comes with a forgetful functor $F:\QEl(A) \to \Hilb$ mapping a quantum element to its underlying Hilbert space and an intertwiner to the underlying linear map. This underlying structure makes $(\QEl(A),F)$ a concrete dagger category as in Definition~\ref{def:Hilbcategory}.

We now justify our definition of quantum elements by showing that a quantum set is completely determined by its concrete dagger category of quantum elements.

\ignore{
\begin{definition}\label{def:intertwinerofelements}Given a quantum set $A$, we define the concrete dagger category $(\QEl(A),F)$ of quantum elements of $A$:
\vspace{-0.1cm}
\begin{itemize} \item \textbf{objects} are quantum elements $(H,Q)$ of $A$;
\item \textbf{morphisms} $(H,Q) \to (H', Q')$ are intertwiners of quantum elements. 
\end{itemize}
Composition of intertwiners is ordinary composition of linear maps. The forgetful functor $F:\QEl(A) \to \Hilb$ maps a quantum element to its underlying Hilbert space and an intertwiner to its underlying linear map.
\end{definition}
}

\ignore{
In other words, we can think of $\QEl(A)$ as the quantum version of the set of elements of a commutative special Frobenius algebra.
}%
\begin{proposition}\label{prop:reconstruct} Up to isomorphism, a quantum set $A$ can be reconstructed from its category of elements $\QEl(A)$ and the forgetful functor $F:\QEl(A)\to \Hilb$.
\end{proposition}
\begin{proof} By definition, $\QEl(A)$ is the category of left comodules of the special symmetric dagger Frobenius algebra $A$ and as such equivalent to the category of modules of $A$. The proposition therefore follows from existing results on Tannaka duality, which states that a semisimple algebra can be reconstructed from its category of modules and a forgetful (or fibre) functor (cf. \cite{Joyal1991}).\end{proof}

\noindent
Explicitly, we can reconstruct the algebra $A$ as follows:%
\begin{calign}\label{eqn:quantumcopy}
\begin{tz}[zx,scale=1.25]
\draw (0.25,-0.25) to (0.25,0) to [out=up, in=-135] (1,1) to (1,2.25);
\draw (1.75,-0.25) to (1.75,0) to [out=up, in=-45] (1,1);
\node[zxvertex=\zxwhite, zxdown] at (1,1) {};
\end{tz}
= ~
\sum_{i} \frac{1}{\dim(H_i)}~~\begin{tz}[zx,scale=1.25,xscale=-1]
\draw (0.25,-0.25) to node[zxnode=\zxwhite, pos=0.46] {$i^\dagger$}(0.25,1);
\draw[string] (1.75,-0.25) to node[zxnode=\zxwhite, pos=0.82] {$i^\dagger$} (1.75,1);
\draw[string] (1,1) to node[zxnode=\zxwhite, pos=0.35] {$i$} (1,2.25);
\draw[string, fill=white, arrow data={0.}{>},arrow data={0.33}{>}, arrow data={0.6}{>},arrow data={0.8}{>}] (-0.5,1) to [out=down, in=-135] (-0.05,0.05) to [in looseness=0.3,out=45, in=down] (2.25,1.) to [out looseness=0.3,out=up, in=-45] (-0.05,1.95) to [out=135, in=up] (-0.5,1);
\end{tz}
&
\begin{tz}[zx]
\draw (1,1) to (1,2.25);
\node[zxvertex=\zxwhite] at (1,1){};
\end{tz}
~=~
\sum_i \frac{1}{\dim(H_i)}~
\begin{tz}[zx,xscale=-1]
\draw (1,1) to node[zxnode=\zxwhite, pos=0] {$i$} (1,2.25);
\draw[string,arrow data={0.17}{>}, arrow data={0.87}{>},arrow data={0.5}{>}] (1,1) to [looseness=3.5, out=135, in=135] (0.25,0.25) to [out=-45, in=-45, looseness=3.5] (1,1);
\end{tz}
\end{calign}
Here the sum ranges over the \textit{simple quantum elements}, which we define in Section~\ref{sec:semisimple}. The equations~\eqref{eqn:quantumcopy} can be understood as a quantisation of equations~\eqref{eq:classicalcopy}.

\subsection{Quantum functions}\label{sec:funct}

Having defined quantum elements of quantum sets, we now consider the appropriate notion of quantum functions between them. We define these by quantisation of \mbox{equation~\eqref{eq:cohomo}.}
\begin{definition}\label{def:quantumfunction} A \textit{quantum function} between quantum sets $A$ and $B$ is a pair ($H,P$), where $H$ is a Hilbert space and $P$ is a linear map $H\otimes A\! \to\! B \otimes H$ satisfying the \mbox{following:}\looseness=-2
\begin{calign}\label{eq:quantumfunction} \begin{tz}[zx,every to/.style={out=up, in=down},xscale=-1]
\draw (0,0) to (0,2) to [out=45] (0.75,3);
\draw (0,2) to [out=135] (-0.75,3);
\draw[arrow data={0.2}{>}, arrow data={0.8}{>}] (1.75,0) to [looseness=0.9] node[zxnode=\zxwhite, pos=0.5] {$P$} (-1.75,2.5) to (-1.75,3);
\node[zxvertex=\zxwhite, zxup] at (0,2){};
\end{tz}
=
\begin{tz}[zx,every to/.style={out=up, in=down},xscale=-1]
\draw (0,0) to (0,0.75) to [out=45] (0.75,1.75) to (0.75,3);
\draw (0,0.75) to [out=135] (-0.75,1.75) to (-0.75,3);
\draw[arrow data={0.2}{>}, arrow data={0.9}{>}] (1.75,0) to (1.75,0.75) to  [looseness=1.1, in looseness=0.9] node[zxnode=\zxwhite, pos=0.36] {$P$} node[zxnode=\zxwhite, pos=0.64] {$P$}(-1.75,3);
\node[zxvertex=\zxwhite, zxup] at (0,0.75){};
\end{tz}
&
\begin{tz}[zx,every to/.style={out=up, in=down},xscale=-1]
\draw (0,0) to (0,2.25);
\draw[arrow data={0.2}{>}, arrow data={0.8}{>}] (1,0) to [looseness=0.9] node[zxnode=\zxwhite, pos=0.5] {$P$} (-1,2.5) to (-1,3);
\node[zxvertex=\zxwhite] at (0,2.25){};
\end{tz}
=
\begin{tz}[zx,every to/.style={out=up, in=down},xscale=-1]
\draw (0,0) to (0,0.75);
\draw[arrow data={0.2}{>}, arrow data={0.9}{>}] (1.,0) to (1.,0.75) to   (-1,3);
\node[zxvertex=\zxwhite, zxup] at (0,0.75){};
\end{tz}
&
\begin{tz}[zx,every to/.style={out=up, in=down},xscale=-0.8]
\draw [arrow data={0.2}{>},arrow data={0.8}{>}]  (0,0) to (2.25,3);
\draw (2.25,0) to node[zxnode=\zxwhite, pos=0.5] {$P^\dagger$} (0,3);
\end{tz}
=~~
\begin{tz}[zx,xscale=-0.6,yscale=-1]
\draw[arrow data={0.5}{<}] (0.25,-0.5) to (0.25,0) to [out=up, in=-135] (1,1);
\draw (1,1) to [out=135, in=right] node[zxvertex=\zxwhite, pos=1]{} (-0.3, 1.7) to [out=left, in=up] (-1.25,1) to (-1.25,-0.5);
\draw[arrow data={0.5}{>}] (1.75,2.5) to (1.75,2) to [out=down, in=45] (1,1);
\draw (1,1) to [out= -45, in= left] node[zxvertex=\zxwhite, pos=1] {} (2.3,0.3) to [out=right, in=down] (3.25,1) to (3.25,2.5);
\node [zxnode=\zxwhite] at (1,1) {$P$};
\end{tz} 
\end{calign}
\ignore{
\[\begin{tz}[zx,every to/.style={out=up, in=down}]
\draw (0,0) to (2.25, 3);
\draw (1.5,4) to [out=down, in=135] (2.25,3) to [out=45, in=down] (3,4);
\draw[arrow data={0.2}{>},arrow data={0.8}{>}] (3,0) to node[zxnode=\zxwhite, pos=0.545] {$Q$} (0,3) to (0,4);
\node[zxvertex=\zxwhite, zxup] at (2.25,3) {};
\end{tz}
=
\begin{tz}[zx,every to/.style={out=up, in=down}]
\draw (0,0) to (0,1) to [out=135, in=down] (-0.75,1.8) to (1.5,4);
\draw (0,1) to [out=45] (0.75,1.8) to [out looseness=0.7]  (3,4);
\draw[arrow data={0.2}{>},arrow data={0.8}{>}] (3,0) to (3,2) to node[zxnode=\zxwhite , pos=0.35] {$Q$} node[zxnode=\zxwhite, pos=0.65]{$Q$} (0,4);
\node[zxvertex=\zxwhite, zxup] at (0,1) {};
\end{tz}
\hspace{1cm}
\]
}%
\end{definition}

\begin{remark}
The diagrammatic representation of quantum functions provides an interesting topological intuition: a quantum function behaves like a braiding or crossing between the directed and the undirected wire. From this perspective,~\eqref{eq:quantumfunction} allows us to pull the comultiplication and counit through the braiding. Note that we cannot yet pull the multiplication and unit through the braiding; adding these additional pull-throughs defines a quantum bijection, as will be seen in Section~\ref{sec:quantumbijection}. 
\end{remark}

\begin{remark}\label{rem:opintofquantumfunctions} We show in Corollary~\ref{cor:quantumelementmeasurement} that quantum functions $X\to Y$ between classical sets are families of projective measurements with outcomes in $Y$, controlled by the set $X$. We can think of this as a non-deterministic function $X \to Y$ which uses quantum measurements on an underlying Hilbert space to determine the output $y \in Y$ for a given input $x \in X$ (cf. Remark~\ref{rem:interpretationquantumelement}).
\end{remark}
\noindent
Having defined quantum functions, we make the following elementary observations.
\begin{proposition} A quantum element of a quantum set $A$ is a quantum function from the one-element {set~$\{*\}$} (or equivalently from the commutative special dagger Frobenius algebra $\mathbb{C}$) to the quantum set $A$.
\end{proposition}
\begin{proof} \eqref{eq:quantumelement} is a special case of \eqref{eq:quantumfunction} where the source is the trivial quantum set.
\end{proof}
\noindent
We also note that quantum functions map quantum elements to quantum elements:%
\begin{equation}\begin{tz}[zx,every to/.style={out=up, in=down},xscale=-1]
\draw (0,0) to (0,3.5);
\draw[arrow data={0.2}{>}, arrow data={0.8}{>}] (1.75,0) to [looseness=0.8] node[zxnode=\zxwhite, pos=0.46](Q) {$Q$} (-2.25,2.5) to (-2.25,3.5);
\draw[string,arrow data={0.35}{>}, arrow data={0.9}{>}] (3,0)  to (3,0.25)to [in looseness=0.7] node[zxnode=\zxwhite, pos=0.67](Q2) {$P$}  (-1,3.25) to (-1,3.5);
\draw[white,double] (0,0) to (Q.center);
\end{tz}
\end{equation}
In particular, a quantum function between quantum sets $A$ and $B$ induces a functor between their categories of quantum elements.

\begin{definition}\label{def:qfctdimension}We define the \textit{dimension} of a quantum function $(H,P)$ to be the dimension of the underlying Hilbert space $H$.
\end{definition}

\begin{remark}
A one-dimensional quantum function is an ordinary $*$-cohomomorphism of Frobenius algebras, as can be seen by comparing~\eqref{eq:cohomo} with~\eqref{eq:quantumfunction}. By Gelfand duality, a one-dimensional quantum function between classical sets is therefore just a function.
\end{remark}
\noindent
We can extend the notion of an intertwiner of quantum elements to all \mbox{quantum functions.}\looseness=-2

\begin{definition}\label{def:intertwiner} An \textit{intertwiner} of quantum functions $(H,P) \to (H', P')$ is a linear map $f:H\to H'$ such that the following holds:%
\begin{calign}
\begin{tz}[zx,every to/.style={out=up, in=down},xscale=-1]
\draw (0,0) to (0,1) to (2,3);
\draw[arrow data={0.35}{>}] (2,0) to node[zxnode=\zxwhite, pos=0.9] {$f$} (2,1);
\draw[string,arrow data={0.9}{>},arrow data={0.26}{>}](2,1) to node[zxnode=\zxwhite, pos=0.5]{$P'$}  (0,3);
\end{tz}
\quad = \quad
\begin{tz}[zx,every to/.style={out=up, in=down},scale=-1,xscale=-1]
\draw (0,0) to (0,1) to (2,3);
\draw[arrow data={0.35}{<}] (2,0) to node[zxnode=\zxwhite, pos=0.9] {$f$} (2,1);
\draw[string,arrow data={0.9}{<},arrow data={0.26}{<}](2,1) to node[zxnode=\zxwhite, pos=0.5]{$P$}  (0,3);
\end{tz}
\end{calign}
\end{definition}
\subsection{The 2-category $\QSet$}
We have seen in the discussion preceding Definition~\ref{def:intertwiner} that, while classical elements form a set, quantum elements should be organized into a \emph{category}, keeping track of the additional layer of structure introduced by the intertwiners. %
\ignore{We have seen that, while classical functions form a set, quantum functions should be organized into a \textit{category}, in order to keep track of the additional layer of structure introduced by the intertwiners. }We thus expect the quantum analogue of the category of sets and functions to be a \emph{2-category} of quantum sets and quantum functions, keeping track of the intertwiners between quantum functions. \ignore{We will show that in fact this is a concrete dagger 2-category.}

Overall, we observe that our approach to quantisation leads to \mbox{categorification.}

\begin{definition}\label{def:2catqset} The 2-category\footnote{If we use a strict version of the monoidal category $\Hilb$, then $\QSet$ is a strict 2-category.} $\QSet$ is built from the following structures:%
\begin{itemize}
\item \textbf{objects} are quantum sets $A,B$, ...;
\item \textbf{1-morphisms} $A\to B$ are quantum functions $(H,P): A\to B$;
\item \textbf{2-morphisms} $(H,P) \to (H', P')$ are intertwiners of quantum functions.
\end{itemize}
The composition of two quantum functions $(H,P):A\to B$ and $(H', Q): B \to C$ is a quantum function $(H' \otimes H, Q\circ P)$ defined as follows:
\begin{calign}
\begin{tz}[zx, every to/.style={out=up, in=down},xscale=-1] \label{eq:1composition}
\draw[arrow data={0.15}{>}, arrow data={0.8}{>}] (1.575,0) to (0.325,3.5);
\draw[string,arrow data={0.18}{>}, arrow data={0.85}{>}] (2.175,0) to (0.925,3.5);
\draw (0,0) to node[zxvertex=\zxwhite, pos=0.5] {$Q\circ P$} (2.5,3.5);
\node[dimension,right] at (4.4,0) {$H' \otimes H$};
\end{tz}
\! := \quad
\begin{tz}[zx, every to/.style={out=up, in=down},xscale=-1]
\draw (0,0) to (2.5,3.5);
\draw[arrow data={0.15}{>}, arrow data={0.8}{>}] (1.25,0) to node[zxnode=\zxwhite, pos=0.4] {$P$} (0,3.5);
\draw[string,arrow data={0.2}{>}, arrow data={0.9}{>}] (2.5,0) to node[zxnode=\zxwhite, pos=0.58] {$Q$} (1.25,3.5);
\node[dimension,right] at (1.25,0) {$H$};
\node[dimension, right] at (2.5,0) {$H'$};
\end{tz}
\end{calign}
Vertical and horizontal composition of 2-morphisms is defined as the ordinary composition and tensor product of linear maps, respectively.
\end{definition}
\noindent
As expected, we observe that $\QSet(*,A) = \QEl(A)$ where $*$ is the one-element set (cf. Terminology~\ref{not:setalgebra}).

\begin{remark}\label{rem:street} We denote by $\BHilb$ the \textit{delooping} of $\Hilb$; that is, the 2-category with a single object $*$ and endomorphism category $\Hom(*,*)=\Hilb$. We observe that $\QSet$ is a sub-2-category of Street's 2-category $\mathrm{CoMnd}(\BHilb)$ of comonads, comonad maps and comonad transformations in the 2-category $\BHilb$~\cite{Street1972}. Indeed, the first two equations of~\eqref{eq:quantumfunction} can be understood as defining a comonad map (an analogous graphical definition of the 2-category $\mathrm{Mnd}(\Cat)$ of monads in the 2-category of categories, functors and natural transformations is given in~\cite{Marsden2014,Hinze2016}).
\end{remark}
\begin{theorem}\label{thm:QSetdagger2category} $\QSet$ is a dagger 2-category. 
\end{theorem}
\begin{proof} Following Remark~\ref{rem:street}, $\QSet$ is a sub-2-category of $\mathrm{CoMnd}(\BHilb)$, which is well known to be a 2-category (in the strict case the proof goes back at least to Street~\cite{Street1972}). The dagger of a 2-morphism is defined to be the ordinary Hilbert space adjoint of the underlying linear map. This is well defined since the adjoint of an intertwiner $f:(H,P) \to (H', Q)$ is an intertwiner $f^\dagger: (H', Q) \to (H,P)$: 
\begin{equation}\label{eq:qsetdagger}\begin{tz}[zx,every to/.style={out=up, in=down},xscale=-1]
\draw (0,0) to (0,1) to (2,3);
\draw[string, arrow data={0.35}{>}] (2,0) to node[zxnode=\zxwhite, pos=0.9] {$f$} (2,1);
\draw[string,arrow data={0.9}{>},arrow data={0.26}{>}](2,1) to node[zxnode=\zxwhite, pos=0.5]{$Q$}  (0,3);
\end{tz}
~=~
\begin{tz}[zx,every to/.style={out=up, in=down},scale=-1,xscale=-1]
\draw (0,0) to (0,1) to (2,3);
\draw[arrow data={0.35}{<}] (2,0) to node[zxnode=\zxwhite, pos=0.9] {$f$} (2,1);
\draw[string,arrow data={0.9}{<},arrow data={0.26}{<}](2,1) to node[zxnode=\zxwhite, pos=0.5]{$P$}  (0,3);
\end{tz}
\quad
\stackrel{~(-)^\dagger}{\Leftrightarrow}
\quad
\begin{tz}[zx,every to/.style={out=up, in=down},yscale=-1,xscale=-1]
\draw (0,0) to (0,1) to (2,3);
\draw[arrow data={0.35}{<}] (2,0) to node[zxnode=\zxwhite, pos=0.9] {$f^\dagger$} (2,1);
\draw[string,arrow data={0.9}{<},arrow data={0.26}{<}](2,1) to node[zxnode=\zxwhite, pos=0.5]{$Q^\dagger$}  (0,3);
\end{tz}
=
\begin{tz}[zx,every to/.style={out=up, in=down},xscale=1]
\draw (0,0) to (0,1) to (2,3);
\draw[arrow data={0.35}{>}] (2,0) to node[zxnode=\zxwhite, pos=0.9] {$f^\dagger$} (2,1);
\draw[string,arrow data={0.9}{>},arrow data={0.26}{>}](2,1) to node[zxnode=\zxwhite, pos=0.5]{$P^\dagger$}  (0,3);
\end{tz}
\quad
\super{\eqref{eq:quantumfunction}}{\Leftrightarrow}
\quad
\begin{tz}[zx,every to/.style={out=up, in=down},scale=-1,xscale=-1]
\draw (0,0) to (0,1) to (2,3);
\draw[arrow data={0.35}{<}] (2,0) to node[zxnode=\zxwhite, pos=0.9] {$f^\dagger$} (2,1);
\draw[string,arrow data={0.9}{<},arrow data={0.26}{<}](2,1) to node[zxnode=\zxwhite, pos=0.5]{$Q$}  (0,3);
\end{tz}
~=~
\begin{tz}[zx,every to/.style={out=up, in=down},xscale=-1]
\draw (0,0) to (0,1) to (2,3);
\draw[arrow data={0.35}{>}] (2,0) to node[zxnode=\zxwhite, pos=0.9] {$f^\dagger$} (2,1);
\draw[string,arrow data={0.9}{>},arrow data={0.26}{>}](2,1) to node[zxnode=\zxwhite, pos=0.5]{$P$}  (0,3);
\end{tz}
\end{equation}
\end{proof}
The underlying Hilbert space of the composite of two quantum functions is the tensor product of the underlying Hilbert spaces of the quantum functions. Similarly, vertical and horizontal composition of 2-morphisms coincides with composition and tensor product of linear maps. We therefore have a forgetful 2-functor \begin{equation}\label{eq:fibrefunctor}F: \QSet \to \BHilb\end{equation}  from $\QSet$ to the delooping of $\Hilb$ (see Remark \ref{rem:street}). As outlined in Philosophy~\ref{phil:concretedagger}, the quantum mechanical interpretation of \QSet\ depends on this 2-functor; in particular, it makes $\QSet$ into a concrete dagger 2-category.

\begin{definition}\label{def:Hilb2category} A \emph{concrete dagger 2-category} is a pair $(\mathbb{B}, F)$, where $\mathbb{B}$ is a dagger 2-category and $F:\mathbb{B} \to \BHilb$ is a locally faithful dagger 2-functor (see Section~\ref{sec:background}). 
\end{definition}
\begin{remark}
In particular, a \emph{concrete dagger monoidal category} is a pair $(\mathcal{C}, F)$ where $\mathcal{C}$ is a dagger monoidal category and $F:\mathcal{C} \to \Hilb$ is a faithful dagger monoidal functor.
\end{remark}
\noindent In other words, a concrete dagger 2-category is a $2$-category for which every $\Hom$-category is a concrete dagger category in a compatible way. The 2-functor $F:\QSet\to \BHilb$ has in fact already appeared in Proposition~\ref{prop:reconstruct}.

\begin{remark}\label{rem:setsubset}
By Gelfand duality, every set is itself a quantum set and every function is a quantum function. We may therefore think of $\Set$ as contained in $\QSet$; there is a faithful inclusion $2$-functor $\Set\hookrightarrow \QSet$, where we think of $\Set$ as a $2$-category with only identity $2$-morphisms.

Given two functions $f$ and $g$ between sets, $\QSet(f,g) = \delta_{f,g} \mathbb{C}$.
\end{remark}

\begin{remark}\label{rem:symmmon2cat} Using the cartesian product of sets, the  category of sets and functions is a symmetric monoidal category. Similarly, it can be shown that \QSet, using the tensor product of the underlying algebras, becomes a symmetric monoidal 2-category.\ignore{\footnote{Working with a strict version of $\Hilb$} quasistrict symmetric Gray monoid.}
\end{remark}

\subsection{A universal property}\label{sec:universal}
In this section, we show that the categories of quantum functions $A \to B$ between quantum sets may be obtained via a universal construction, as categories of\ignore{ finite} quantum elements of the internal hom $[A,B]$ in the category of \emph{quantum spaces} --- the opposite of the category of $C^*$-algebras. Our definition of the quantum space of quantum functions between finite quantum sets is analogous to~\cite[Definition 3.1.]{Soltan2009}, and generalises the construction used to define quantum permutation groups~\cite{Wang1998}.

To present these results, we work with infinite-dimensional $C^*$-algebras. Various parts of our presentation in Sections~\ref{sec:background} and~\ref{sec:quantumset} do not apply here (for instance, there are no infinite-dimensional Frobenius algebras), and we will therefore need to modify some definitions. Any such modifications will apply only in this section.

It is an observation of Wang~\cite{Wang1998} that quantum functions between classical spaces should form a quantum set rather than a classical set. Let $\CAlg$ be the category of (possibly infinite-dimensional) $C^*$-algebras and $*$-homomorphisms. By analogy with Gelfand duality, we might treat the category $\CAlg^\op$ as the category of `quantum spaces'. However, given two finite sets $A$ and $B$, understood as commutative $C^*$-algebras, then $\CAlg^\op (A,B) =\Set(A,B)$ is just a set; from this perspective there are only classical functions between $A$ and $B$. This inspires the following definition, analogous to constructions in~\cite{Wang1998} and~\cite{Soltan2009}.     
\begin{definition}\label{def:internalhom} \ignore{Let $A$ and $B$ be finite-dimensional $C^*$-algebras (not necessarily commutative). }The \textit{quantum space of quantum functions} between two finite quantum sets $A$ and $B$ is the internal hom $[A,B]$ in $\CAlg^\op$; that is, the universal $C^*$-algebra such that there are bijections\footnote{All finite-dimensional $C^*$-algebras $A$ are \emph{nuclear}; that is for all $C^*$-algebras $C$, there is a unique norm on the algebraic tensor product; the $C^*$-tensor product $C\otimes A$ is defined as its completion.   }
\begin{equation}\label{eq:internalhom}\CAlg^\op (C, [A,B]) \cong \CAlg^\op (C\otimes A, B)\end{equation}
which are natural in $C\in \CAlg^\op$.
\end{definition}
\begin{remark}\label{rem:soltan}
So\l{}tan's \emph{quantum space of all maps} between two quantum spaces~\cite[Definition 3.1.]{Soltan2009} is precisely the internal hom in the category $\CAlg^\op$. It follows from~\cite[Theorem 3.3]{Soltan2009} that the quantum space of quantum functions between finite quantum sets always exists. 
\end{remark}

\noindent
The internal hom $[A,B]$ is in general an infinite-dimensional noncommutative $C^*$-algebra. In the finite-dimensional case, we showed that a quantum set could be recovered from its category of quantum elements. In the infinite-dimensional setting we cannot define the category of quantum elements as the category of comodules, due to the lack of a comultiplication on $[A,B]$. We can nevertheless define the category of quantum elements of $[A,B]$ as the \textit{opposite} of the category of modules of $[A,B]$; this coincides with the usual definition in the finite-dimensional case.\ignore{\footnote{Takesaki showed~\cite{Takesaki1967} that a general separable $C*$-algebra can be recovered from its category of (possibly infinite-dimensional) representations.} Since we here only treat finite dimensional Hilbert spaces, we make the following definition.}

\ignore{ In the finite-dimensional case, taking the adjoint induces an equivalence between the category of comodules and the opposite of the category of modules; we use this to define the category of quantum elements for an infinite-dimensional $C^*$-algebra. }
\begin{definition} The \textit{category of finite quantum elements} $\QEl(A)$ of a quantum space $A$ is the opposite of the category $\mathrm{Rep}_{\text{fd}}(A)$ of finite-dimensional $C^*$-representations\footnote{A $C^*$-representation of $A$ is a $*$-homomorphism $A\to \mathcal{B}(H)$, where $\mathcal{B}(H)$ is the $C^*$-algebra of bounded operators on a Hilbert space $H$.} of the $C^*$-algebra $A$.
\end{definition}
\noindent
We now show that our category of quantum functions between finite quantum sets $A$ and $B$ coincides with the category of quantum elements of the internal hom $[A,B]$.
\begin{theorem}\label{thm:universalprop} The category of finite quantum elements of $[A,B]$ is equivalent to the category of quantum functions $\QSet(A,B)$:
\begin{equation}\QEl([A,B]) \cong \QSet(A,B)\end{equation}
\end{theorem}
\begin{proof} We need to show that $\mathrm{Rep}_{\text{fd}}([A,B])^\op \cong \QSet(A,B)$. A representation of $[A,B]$ is a $*$-homomorphism $[A,B]\to \End(H)$ for some finite-dimensional Hilbert space. We have the following series of bijections: 
\begin{equation}\label{eq:universalsequence}\CAlg([A,B], \End(H)) = \CAlg^\op(\End(H),[A,B]) \end{equation}\[\hspace{1cm}\cong \CAlg^\op(\End(H)\otimes A, B) = \CAlg(B, \End(H)\otimes A) \]
The representation therefore corresponds to a $*$-homomorphism $B\to \End(H) \otimes A$. Since all algebras $A,B$ and $\End(H)$ are finite-dimensional, we can now use the Frobenius algebra formalism established in Section~\ref{sec:background}. Taking the Hilbert space adjoint of the homomorphism $B\to \End(H)\otimes A$ yields a $*$-cohomomorphism (in the sense of equation~\eqref{eq:cohomo}) $P:\End(H)\otimes A\to B$. Explicitly, this means:
\def\d{0.25}
\def\h{0.4}
\begin{calign}\nonumber
\begin{tz}[zx,master,every to/.style={out=up, in=down},scale=1.2]
\draw (0,0) to (0,0.5) to [out=135] (-0.75, 1.5) to (-0.75,1.7);
\draw (0, 0.5) to [out=45] (0.75,1.5) to (0.75,1.7);
\draw[arrow data={0.2}{>}] (-2*\d, 0)to   (-2*\d, \h) to(-0.75-2*\d, 1.7);
\draw[arrow data={0.12}{<}] (-\d, 0) to (-\d,\h) to  (0.75-\d, 1.7) to (0.75-\d, 3);
\draw (-0.75-\d, 1.8) to (-0.75-\d, 3);
\draw[arrow data={0.55}{>}] (-0.75-\d, 1.7) to [out=down, in=down,looseness=1.6] (0.75-2*\d, 1.7) ;
\node[zxnode=\zxwhite] at (0.75-\d, 1.8) {$P$};
\node[zxnode=\zxwhite] at (-0.75-\d, 1.8) {$P$};
\node[zxvertex=\zxwhite,zxup] at (0,0.5){};
\end{tz}
\quad =
\begin{tz}[zx,slave,every to/.style={out=up, in=down},scale=1.2]
\draw[arrow data={0.4}{>}] (-\d,0)  to (-\d,1);
\draw[arrow data={0.4}{<}] (0,0)  to (0,1);
\draw (\d,0)  to (\d,1);
\draw (0,1) to (0,2) to [out=135] (-0.75, 3);
\draw (0,2) to [out=45] (0.75,3);
\node[zxnode=\zxwhite] at (0,1) {$P$};
\node[zxvertex=\zxwhite,zxup] at (0,2){};
\end{tz}
&
\begin{tz}[zx,slave,every to/.style={out=up, in=down},scale=1.2]
\draw[arrow data={0.4}{>}] (-\d,0)  to (-\d,1);
\draw[arrow data={0.4}{<}] (0,0)  to (0,1);
\draw (\d,0)  to (\d,1);
\draw (0,1) to (0,2);
\node[zxnode=\zxwhite] at (0,1) {$P$};
\node[zxvertex=\zxwhite] at (0,2){};
\end{tz}
\quad=~
\begin{tz}[zx,slave,every to/.style={out=up, in=down},scale=1.2]
\draw[arrow data={0.2}{>},arrow data={0.9}{>}] (-\d,0)  to (-\d,0.5) to [out=up, in=up] (0,0.5) to (0,0);
\draw (\d,0)  to (\d,1);
\node[zxvertex=\zxwhite] at (\d,1){};
\end{tz}
&
\begin{tz}[zx,every to/.style={out=up, in=down},scale=1.2,yscale=-1]
\clip (-0.5, 0) rectangle (0.5,3);
\draw[arrow data={0.4}{<}] (-\d,0)  to (-\d,1.5);
\draw[arrow data={0.4}{>}] (0,0)  to (0,1.5);
\draw (\d,0)  to (\d,1.5);
\draw (0,1.5) to (0,3);
\node[zxnode=\zxwhite] at (0,1.5) {$P^\dagger$};
\end{tz}
~=~
\def\loose{3}
\begin{tz}[zx,every to/.style={out=up, in=down},scale=1.2]
\clip (-1.5, 0) rectangle (1.5,3);
\draw (0,1.5) to [out=up, in=up,looseness=\loose] node[zxvertex=\zxwhite, pos=0.5]{}(1,1.5) to (1,0);
\draw[string,arrow data={0.76}{<}] (-\d,1.5) to [out=down, in=down,looseness=\loose] (-1, 1.5) to (-1,3);
\draw[string,arrow data={0.8}{>}] (0,1.5) to [out=down, in=down,looseness=2.5] (-1-\d, 1.5) to (-1-\d,3);
\draw[string] (\d,1.5) to [out=down, in=down,looseness=4]node[zxvertex=\zxwhite, pos=0.5] {} (-1+\d, 1.5) to (-1+\d,3);
\node[zxnode=\zxwhite] at (0,1.5) {$P$};
\end{tz}
\end{calign}
But these are exactly the axioms of a quantum function \eqref{eq:quantumfunction} for the following morphism $H\otimes A \to B\otimes H$:
\[
\begin{tz}[zx,every to/.style={out=up, in=down},scale=1.2]
\draw (0,1.5) to (-0.75, 3);
\draw (0.75,0) to (\d-0.05, 1.5);
\draw[arrow data={0.35}{>}] (-0.75,0) to (-\d+0.1,1.5); 
\draw[arrow data={0.2}{>},arrow data={0.85}{>}] (0,1.5) to[out=down, in=down, looseness=5] (0.7, 1.5) to[out=up, in=down] (0.75, 3);
\node[zxnode=\zxwhite] at (0,1.5) {$P$};
\end{tz}
\]
Analogously, an intertwiner $f$ of $[A,B]$-actions is mapped to the intertwiner $f^\dagger$ of the corresponding quantum functions.
This defines the desired equivalence.
\end{proof}

\begin{remark} From this perspective, $\QSet$ provides a high-level approach to the representation theory of the algebras of quantum functions, in particular clarifying their compositional behaviour.
\end{remark}

\begin{remark}\label{rem:quantumfunctionalgebra} It is not hard to show that for classical sets $X$ and $Y$, the $C^*$-algebra $[X,Y]$ is the universal $C^*$-algebra with generators $p_{x,y}$ for $x\in X, y \in Y$ and the following relations for $x\in X, y \in Y$ (also see Theorem~\ref{thm:quantumfunctionsinprojectors} and Proposition~\ref{prop:Wang}):
\begin{calign}p_{x,y}^* = p_{x,y} = p_{x,y}^2 & p_{x,y} p_{x,y'} = \delta_{y,y'} p_{x,y} &\sum_{y\in Y}  p_{x,y}=1\end{calign}
For general finite-dimensional $C^*$-algebras $A$ and $B$,\! there is a similar \mbox{presentation of $[\!A,B]$.}\looseness=-2
\end{remark}
\begin{remark}\label{rem:problemTannaka}
We remark that the categories $\QSet(A,B)\cong \mathrm{Rep}_{\mathrm{fd}}([A,B])^{op}$ only keep track of the \emph{finite-dimensional} representation theory of the internal hom $[A,B]$. 

It is tempting to try to reconstruct the algebra $[A,B]$ from the abstract $C^*$-category $\QSet(A,B)$ and the fibre functor $\QSet([A,B])\to \Hilb$. However, if neither $A$ nor $B$ are the one-element set --- and if therefore $[A,B]$ is infinite-dimensional --- we will see in Section~\ref{sec:semisimple} that the set of isomorphism classes of simple objects of $\QSet(A,B)$ in uncountable and carries a natural non-discrete topology. This cannot be captured with the inherently discrete tools of (ordinary) category theory; we cannot reconstruct $[A,B]$ from the abstract category $\QSet(A,B) \to \Hilb$. This is analogous to the inability to recover a discrete group from its (naive) representation category.
\end{remark}

\section{Quantum bijections}\label{sec:quantumbijection}

\ignore{We now begin to explore the structure of the concrete dagger 2-category \QSet .}
The quantum graph isomorphisms of Atserias et al.~\cite{Atserias2016}, the quantum permutation groups of Wang~\cite{Wang1998}, and the quantum automorphism groups of Banica~\cite{Banica2005} are all intended as quantisations of classical bijections. In this section, we define within our $2$-categorical framework a notion of quantum bijection which we show to be equivalent to the quantum bijections of these authors. 

In particular, we show that quantum bijections can be exactly characterised as the \textit{dagger-dualisable} $1$-morphisms in the $2$-category $\QSet$. We emphasise that quantum bijections should not be thought of as invertible but merely as dualisable. In fact, categorical equivalences in $\QSet$ are just classical bijections, or more generally {$*$-isomorphisms}.

\ignore{
We show that categorical equivalences in $\QSet$ are classical bijections, or more generally $*$-isomorphisms; we then define quantum bijections as quantisations of these isomorphisms. We show that quantum bijections can be exactly characterised as the \textit{dagger-dualisable} $1$-morphisms in the $2$-category $\QSet$. In particular, we emphasise that quantum bijections should not be thought of as invertible but merely as dualisable. 
\ignore{We show that equivalences in $\QSet$ are classical bijections, or more generally ordinary $*$-isomorphisms between Frobenius algebras. We then define quantum bijections as a quantisation of these $*$-isomorphisms, prove that they are the dagger-dualisable, and relate them to Wang's quantum symmetry groups. Finally, we consider quantum bijections between classical sets and relate them to Atserias' projective permutation matrices.\DR{flag}}
}

\subsection{Equivalences in $\QSet$}
We first prove that equivalences in the sense of 2-category theory correspond to ordinary $*$-isomorphisms of Frobenius algebras; this shows that categorical equivalence is too strong a notion to characterise quantum bijections in $\QSet$.

Ordinary bijections --- and more generally ordinary $*$-isomorphisms between Frobenius algebras --- can be characterised by the following algebraic condition:%
\begin{proposition} \label{prop:Frobeniusinvertible}A $*$-cohomomorphism between symmetric special Frobenius algebras is invertible if and only if it is also a $*$-homomorphism.\footnote{
It is straightforward to show directly that a $*$-homomorphism which is also a $*$-cohomomorphism is invertible. However, it is harder to show that an invertible $*$-cohomomorphism is a $*$-homomorphism; this may be a new result.}%
\begin{calign}\label{eq:Frobeniusmonoidlemma}
\begin{tz}[zx, master, every to/.style={out=up, in=down},yscale=-1]
\draw (0,0) to (0,2) to [out=135] (-0.75,3);
\draw (0,2) to [out=45] (0.75, 3);
\node[zxnode=\zxwhite] at (0,1) {$f$};
\node[zxvertex=\zxwhite, zxdown] at (0,2) {};
\end{tz}
=
\begin{tz}[zx, every to/.style={out=up, in=down},yscale=-1]
\draw (0,0) to (0,0.75) to [out=135] (-0.75,1.75) to (-0.75,3);
\draw (0,0.75) to [out=45] (0.75, 1.75) to +(0,1.25);
\node[zxnode=\zxwhite] at (-0.75,2) {$f$};
\node[zxnode=\zxwhite] at (0.75,2) {$f$};
\node[zxvertex=\zxwhite, zxdown] at (0,0.75) {};
\end{tz}
&
\begin{tz}[zx,slave, every to/.style={out=up, in=down},yscale=-1]
\draw (0,0) to (0,2) ;
\node[zxnode=\zxwhite] at (0,1) {$f$};
\node[zxvertex=\zxwhite, zxup] at (0,2) {};
\end{tz}
=
\begin{tz}[zx,slave, every to/.style={out=up, in=down},yscale=-1]
\draw (0,0) to (0,0.75) ;
\node[zxvertex=\zxwhite, zxup] at (0,0.75) {};
\end{tz}
\end{calign}
\end{proposition}
\noindent
\begin{proof}
The proposition follows as a corollary from Theorem~\ref{thm:dual2}, as shown in Remark~\ref{rem:lemma42fromlemma46}.\looseness=-2
\end{proof}
\noindent
Now we show that equivalences in $\QSet$ are ordinary $*$-isomorphisms between Frobenius algebras.%
\begin{proposition} \label{prop:equivalence}Equivalences between quantum sets in $\QSet$ are $*$-isomorphisms of Frobenius algebras.
\end{proposition}
\begin{proof} Equivalences are preserved by the 2-functor $F:\QSet\to \BHilb$ of equation~\eqref{eq:fibrefunctor}. In particular, the underlying Hilbert space $H$ of such an equivalence is invertible; there is another Hilbert space $H'$ such that $H\otimes H' \cong \mathbb{C}$. Therefore, $H\cong \mathbb{C}$ and the equivalence is a quantum function with one-dimensional Hilbert space and thus an invertible $*$-cohomomorphism. It follows from Proposition~\ref{prop:Frobeniusinvertible} that invertible $*$-cohomomorphisms are $*$-isomorphisms. 
\ignore{In particular, given equivalent quantum sets $A$ and $B$, we have quantum functions $(P,H):A\to B$ and $(P',H'):B\to A$ and invertible intertwiner $\alpha:(P',H')\circ (P,H)\to (\id_A,\C)$. The image of $\alpha$ under $F$ gives us an isomorphism $H\otimes H' \cong \mathbb{C}$.
Therefore, $\dim(H)=\dim(H')=1$ and so $P$ and $P'$ are functions with $P\circ P'$ and $P'\circ P$ isomorphic to the identity. Thus the equivalence is an ordinary classical isomorphism of Frobenius algebras (or bijection of sets). }
\end{proof}
\noindent

\subsection{Two characterisations of quantum bijections}
\label{sec:twocharacterisations}
Having established that equivalences in $\QSet$ are ordinary $*$-isomorphisms, we must look elsewhere for an appropriate definition of quantum bijection. 
There are two possible approaches. The first is to quantise the algebraic characterisation of $*$-isomorphism from Proposition~\ref{prop:Frobeniusinvertible}, following the procedure outlined in Section~\ref{sec:quantumset}. The other more structural approach is to define quantum bijections in terms of a genuinely categorical notion of weakened equivalence --- the notion of \textit{dagger adjunction} or \textit{dagger duality}. 
We will show that both these approaches lead to the same definition of quantum bijection.

We first quantise the algebraic characterisation~\eqref{eq:Frobeniusmonoidlemma}.

\begin{definition}\label{def:quantumbijection}
A \textit{quantum bijection} between quantum sets $A$ and $B$ is a quantum function $(H,P)$ between $A$ and $B$ fulfilling the following additional equations:%
\begin{calign}\label{eq:monadmap}
\begin{tz}[zx,every to/.style={out=up, in=down},scale=-1,xscale=-1]
\draw (0,0) to (0,2) to [out=45] (0.75,3);
\draw (0,2) to [out=135] (-0.75,3);
\draw[arrow data={0.2}{<}, arrow data={0.8}{<}] (1.75,0) to [looseness=0.9] node[zxnode=\zxwhite, pos=0.5] {$P$} (-1.75,2.5) to (-1.75,3);
\node[zxvertex=\zxwhite, zxdown] at (0,2){};
\end{tz}
=
\begin{tz}[zx,every to/.style={out=up, in=down},scale=-1,xscale=-1]
\draw (0,0) to (0,0.75) to [out=45] (0.75,1.75) to (0.75,3);
\draw (0,0.75) to [out=135] (-0.75,1.75) to (-0.75,3);
\draw[arrow data={0.2}{<}, arrow data={0.9}{<}] (1.75,0) to (1.75,0.75) to  [looseness=1.1, in looseness=0.9] node[zxnode=\zxwhite, pos=0.36] {$P$} node[zxnode=\zxwhite, pos=0.64] {$P$}(-1.75,3);
\node[zxvertex=\zxwhite, zxdown] at (0,0.75){};
\end{tz}
&
\begin{tz}[zx,every to/.style={out=up, in=down},scale=-1,xscale=-1]
\draw (0,0) to (0,2.25);
\draw[arrow data={0.2}{<}, arrow data={0.8}{<}] (1,0) to [looseness=0.9] node[zxnode=\zxwhite, pos=0.5] {$P$} (-1,2.5) to (-1,3);
\node[zxvertex=\zxwhite] at (0,2.25){};
\end{tz}
=
\begin{tz}[zx,every to/.style={out=up, in=down},scale=-1,xscale=-1]
\draw (0,0) to (0,0.75);
\draw[arrow data={0.2}{<}, arrow data={0.9}{<}] (1.,0) to (1.,0.75) to   (-1,3);
\node[zxvertex=\zxwhite] at (0,0.75){};
\end{tz}
\end{calign}
\end{definition}
\noindent
Alternatively, quantum bijections can be characterised as duals in $\QSet$.
\def\*F{{^*\! F}}%
\def\F*{F^*}%
\begin{definition}\label{def:duals} A $1$-morphism $F:A\to B$ in a 2-category has a \textit{right dual} $\F*:B\to A$ if there are 2-morphisms $\epsilon_R:F{\circ}\F*\to 1_{B}$ (counit) and $\eta_R:1_A\to \F*{\circ} F$ (unit) fulfilling the triangle equations\footnote{The triangle equations given here are strict but can straightforwardly be weakened to apply to weak 2-categories.}:
\begin{calign}\left(\epsilon_R \circ 1_{F}\right)\left(1_{F}\circ \eta_R\right) = 1_{F}
&
\left(1_{\F*}\circ \epsilon_R \right)\left(\eta_R\circ 1_{\F*}\right) = 1_{\F*}
\end{calign}
A $1$-morphism has a \textit{left dual} $\*F:B\to A$ if there are 2-morphisms $\epsilon_L:\*F {\circ}F\to 1_{A}$ and $\eta_L:1_B\to F {\circ}\*F$ fulfilling the triangle equations:
\begin{calign}\left(\epsilon_L \circ 1_{\*F}\right)\left(1_{\*F}\circ \eta_L\right) = 1_{\*F}
&
\left(1_F\circ \epsilon_L \right)\left(\eta_L\circ 1_{F}\right) = 1_{F}
\end{calign}
In a dagger 2-category, a $1$-morphism has a \textit{dual} $\overline{F}:B\to A$ if $\overline{F}$ is both a left and a right dual and the corresponding 2-morphisms are related as follows:
\begin{calign} \epsilon_L^\dagger = \eta_R & \eta_L^\dagger = \epsilon_R
\end{calign}
\end{definition}
\begin{remark}\label{rem:daggerdual}
In a dagger $2$-category, every left (or right) dual is automatically a dual in the sense of Definition~\ref{def:duals}. We refer to a $1$-morphism that has a dual as \textit{dualisable}.\end{remark}
\noindent
As discussed in Philosophy~\ref{phil:concretedagger}, $\QSet$ should be thought of not only as a 2-category but as a concrete dagger 2-category (see Definition~\ref{def:Hilb2category}), equipped with its forgetful $2$-functor $\QSet \to \BHilb$. In particular, the notion of a dual in $\QSet$ should be compatible with the underlying Hilbert space structure. This leads to the following definition.
\begin{definition} \label{def:daggerdualisable}In a concrete dagger 2-category $(\mathbb{B}, F)$, a $1$-morphism $S$ is \emph{dagger dualisable} if it has a dual $\conj{S}$ such that the underlying duality between $F(S)$ and $F(\conj{S})$ in $\Hilb$ is a dagger duality in the sense of Definition~\ref{def:daggerduality}.

We will refer to such a $1$-morphism $\conj{S}$ as a \emph{dagger dual} of $S$.
\end{definition}

\noindent
To characterise dagger-dualisable quantum functions, we adopt terminology from \cite{Vicary2012hq, Reutter2016} and make the following definition.
\begin{definition}\label{def:biinvertible}Let $A,B$ and $H$ be finite-dimensional Hilbert spaces. A linear map \\\mbox{$P:H\otimes A \to B \otimes H$} is \textit{bi-invertible} if it is invertible and if
\begin{equation}\label{eq:biinv}\begin{tz}[zx,xscale=0.6,yscale=1,scale=1]
\draw(0.25,-0.5) to (0.25,0) to [out=up, in=-135] (1,1);
\draw [arrow data={0.34}{>}]  (1,1) to [out=135, in=right]  (-0.3, 1.7) to [out=left, in=up] (-1.25,1) to (-1.25,-0.5);
\draw(1.75,2.5) to (1.75,2) to [out=down, in=45] (1,1);
\draw[arrow data={0.32}{<}] (1,1) to [out= -45, in= left]   (2.3,0.3) to [out=right, in=down] (3.25,1) to (3.25,2.5);
\node [zxnode=\zxwhite] at (1,1) {$P^{-1}$};
\end{tz}
~~\text{ is inverse to }~~
\begin{tz}[zx,xscale=0.6,yscale=-1,scale=1]
\draw(0.25,-0.5) to (0.25,0) to [out=up, in=-135] (1,1);
\draw [arrow data={0.34}{<}]  (1,1) to [out=135, in=right]  (-0.3, 1.7) to [out=left, in=up] (-1.25,1) to (-1.25,-0.5);
\draw(1.75,2.5) to (1.75,2) to [out=down, in=45] (1,1);
\draw[arrow data={0.32}{>}] (1,1) to [out= -45, in= left]   (2.3,0.3) to [out=right, in=down] (3.25,1) to (3.25,2.5);
\node [zxnode=\zxwhite] at (1,1) {$P$};
\end{tz}
\end{equation}
\end{definition}
\noindent
\ignore{Using the graphical calculus, it can be shown that \eqref{eq:biinv} can equivalently be expressed as follows:
\begin{equation}\label{eq:biinv2}\begin{tz}[zx,xscale=0.6,yscale=1]
\draw[arrow data={0.5}{>}] (0.25,-0.5) to (0.25,0) to [out=up, in=-135] (1,1);
\draw (1,1) to [out=135, in=right] node[zxvertex=\zxwhite, pos=1]{} (-0.3, 1.7) to [out=left, in=up] (-1.25,1) to (-1.25,-0.5);
\draw[arrow data={0.5}{<}] (1.75,2.5) to (1.75,2) to [out=down, in=45] (1,1);
\draw (1,1) to [out= -45, in= left] node[zxvertex=\zxwhite, pos=1] {} (2.3,0.3) to [out=right, in=down] (3.25,1) to (3.25,2.5);
\node [zxnode=\zxwhite] at (1,1) {$P^{-1}$};
\end{tz} 
~~\text{ is inverse to }~~
\begin{tz}[zx,xscale=0.6,yscale=-1]
\draw[arrow data={0.5}{<}] (0.25,-0.5) to (0.25,0) to [out=up, in=-135] (1,1);
\draw (1,1) to [out=135, in=right] node[zxvertex=\zxwhite, pos=1]{} (-0.3, 1.7) to [out=left, in=up] (-1.25,1) to (-1.25,-0.5);
\draw[arrow data={0.5}{>}] (1.75,2.5) to (1.75,2) to [out=down, in=45] (1,1);
\draw (1,1) to [out= -45, in= left] node[zxvertex=\zxwhite, pos=1] {} (2.3,0.3) to [out=right, in=down] (3.25,1) to (3.25,2.5);
\node [zxnode=\zxwhite] at (1,1) {$P$};
\end{tz} 
\end{equation}}

\noindent
We now show that quantum bijections as in Definition~\ref{def:quantumbijection} are precisely dagger-dualisable quantum functions.
\begin{theorem}\label{thm:dual2}For a quantum function $(H,P)$ between quantum sets $A$ and $B$, the following are equivalent:
\begin{enumerate}
\item   $(H,P)$ is dagger-dualisable in $\QSet$.
\item The underlying linear map $P:H\otimes A \to B \otimes H$ is bi-invertible.
\item The underlying linear map $P:H\otimes A \to B\otimes H$ is unitary.
\item  $(H,P)$ is a quantum bijection as in Definition~\ref{def:quantumbijection}.
\end{enumerate}
If these properties hold, then the quantum function $(H^*,\overline{P}): B\to A$ whose underlying linear map $\overline{P}:H^* \otimes B \to A \otimes H^*$ is defined as follows, is a dagger dual of $(H,P)$:
\begin{equation}\label{eq:daggerdual}
\begin{tz}[zx,every to/.style={out=up, in=down},xscale=0.8]
\draw [arrow data={0.2}{<},arrow data={0.8}{<}]  (0,0) to (2.25,3);
\draw (2.25,0) to node[zxnode=\zxwhite, pos=0.5] {$\overline{P}$} (0,3);
\end{tz}
~~~:=~~~
\begin{tz}[zx,xscale=0.6,xscale=1]
\draw(0.25,-0.5) to (0.25,0) to [out=up, in=-135] (1,1);
\draw [arrow data={0.34}{>}]  (1,1) to [out=135, in=right]  (-0.3, 1.7) to [out=left, in=up] (-1.25,1) to (-1.25,-0.5);
\draw(1.75,2.5) to (1.75,2) to [out=down, in=45] (1,1);
\draw[arrow data={0.32}{<}] (1,1) to [out= -45, in= left]   (2.3,0.3) to [out=right, in=down] (3.25,1) to (3.25,2.5);
\node [zxnode=\zxwhite] at (1,1) {$P^\dagger$};
\end{tz} 
~~~=~~~
\begin{tz}[zx,xscale=0.6,xscale=1]
\draw[arrow data={0.68}{>}]  (4.75,2.5) to (4.75, 0.5) to [out=down, in=right] (2.65,-0.6) to [out=left, in=down] (0.55,0.5);
\draw (0.55,0.5) to [out=up, in=-135] (1,1);
\draw  (1,1) to [out=135, in=right]  (-0.3, 1.7) to [out=left, in=up]  node[zxvertex=\zxwhite, pos=0] {}(-1.25,1) to (-1.25,-0.5);
\draw[string, arrow data={0.68}{<}] (-2.75, -0.5) to (-2.75,1.5) to [out=up, in=left] (-0.65, 2.6) to [out=right, in=up] (1.45,1.5);
\draw  (1.45,1.5) to [out=down, in=45] (1,1);
\draw (1,1) to [out= -45, in= left]  node[zxvertex=\zxwhite, pos=1] {}  (2.3,0.3) to [out=right, in=down] (3.25,1) to (3.25,2.5);
\node [zxnode=\zxwhite] at (1,1) {$P$};
\end{tz} 
\end{equation}
The units and counits witnessing the duality between $(H,P)$ and $(H^*, \conj{P})$ are given by the standard cups and caps in $\Hilb$~\eqref{eq:cupscapsHilb}. Moreover, $(H^*, \conj{P})$ is the unique dagger dual with these units and counits.

\end{theorem}
\begin{proof}
We prove the implications $1.\Rightarrow 2.\Rightarrow 4.\Rightarrow 3.\Rightarrow 1.$  

$1. \Rightarrow 2.$ Suppose that $P$ has a dagger dual $Q'$. Since a dagger duality in $\QSet$ induces a dagger duality of the underlying Hilbert spaces and since dagger dualities are unique up to unitary isomorphism, the underlying Hilbert space of $Q'$ is unitarily isomorphic to the dual space $H^*$. Conjugating $Q'$ by this unitary isomorphism leads to a quantum function $Q$, dagger dual to $P$, whose underlying Hilbert space is exactly $H^*$ and such that the underlying linear maps of the unit and counit of the adjunction are the `standard' cups and caps in $\Hilb$~\eqref{eq:cupscapsHilb}. In other words, a quantum function $P:H\otimes A \to B \otimes H$ has a dagger dual, if and only if there is a quantum function $Q:H^* \otimes B \to A \otimes H^*$ such that the following holds:
\def\scl{0.75}
\begin{calign}\label{eq:rightdual}
\begin{tz}[zx,scale=\scl,xscale=-1]
\clip (-2.05,-0.05) rectangle (4.05,4.05);
\draw[arrow data={0.1}{>},arrow data={0.499}{>}, arrow data={0.9}{>}] (0,0) to [out=up, in=up, looseness=6.] (2,0);
\draw (-2,0) to [out=up, in=down]  node[zxnode=\zxwhite, pos=0.4] {$P$} node[zxnode=\zxwhite, pos=0.6] {$Q$}(4,4);
\end{tz}
=
\begin{tz}[zx,scale=\scl,xscale=-1]
\clip (-2.05,-0.05) rectangle (4.05,4.05);
\draw[arrow data={0.2}{>}, arrow data={0.8}{>}] (0,0) to [out=up, in=up, looseness=2.] (2,0);
\draw (-2,0) to [out=up, in=down]  (4,4);
\end{tz}
&
\begin{tz}[zx,scale=-1,scale=\scl,xscale=-1]
\clip (-2.05,-0.05) rectangle (4.05,4.05);
\draw[arrow data={0.1}{<},arrow data={0.499}{<}, arrow data={0.9}{<}] (0,0) to [out=up, in=up, looseness=6.] (2,0);
\draw (-2,0) to [out=up, in=down]  node[zxnode=\zxwhite, pos=0.4] {$P$} node[zxnode=\zxwhite, pos=0.6] {$Q$}(4,4);
\end{tz}
=
\begin{tz}[zx,scale=-1,scale=\scl,xscale=-1]
\clip (-2.05,-0.05) rectangle (4.05,4.05);
\draw[arrow data={0.2}{<}, arrow data={0.8}{<}] (0,0) to [out=up, in=up, looseness=2.] (2,0);
\draw (-2,0) to [out=up, in=down]  (4,4);
\end{tz}
\end{calign}
\begin{calign}\label{eq:leftdual}
\begin{tz}[zx,master,scale=-1,scale=\scl,xscale=-1]
\clip (-2.05,-0.05) rectangle (4.05,4.05);
\draw[arrow data={0.1}{>},arrow data={0.499}{>}, arrow data={0.9}{>}] (0,0) to [out=up, in=up, looseness=6.] (2,0);
\draw (-2,0) to [out=up, in=down]  node[zxnode=\zxwhite, pos=0.4] {$Q$} node[zxnode=\zxwhite, pos=0.6] {$P$}(4,4);
\end{tz}
=
\begin{tz}[zx,scale=-1,scale=\scl,xscale=-1]
\clip (-2.05,-0.05) rectangle (4.05,4.05);
\draw[arrow data={0.2}{>}, arrow data={0.8}{>}] (0,0) to [out=up, in=up, looseness=2.] (2,0);
\draw (-2,0) to [out=up, in=down]  (4,4);
\end{tz}
&
\begin{tz}[zx,scale=1,scale=\scl,xscale=-1]
\clip (-2.05,-0.05) rectangle (4.05,4.05);
\draw[arrow data={0.1}{<},arrow data={0.499}{<}, arrow data={0.9}{<}] (0,0) to [out=up, in=up, looseness=6.] (2,0);
\draw (-2,0) to [out=up, in=down]  node[zxnode=\zxwhite, pos=0.4] {$Q$} node[zxnode=\zxwhite, pos=0.6] {$P$}(4,4);
\end{tz}
=
\begin{tz}[zx,scale=1,scale=\scl,xscale=-1]
\clip (-2.05,-0.05) rectangle (4.05,4.05);
\draw[arrow data={0.2}{<}, arrow data={0.8}{<}] (0,0) to [out=up, in=up, looseness=2.] (2,0);
\draw (-2,0) to [out=up, in=down]  (4,4);
\end{tz}
\end{calign}
Equations~\eqref{eq:rightdual} and~\eqref{eq:leftdual} can equivalently be expressed as stating the following:
\begin{calign}\nonumber 
\begin{tz}[zx,xscale=0.6,yscale=-1,scale=0.8]
\draw(0.25,-0.5) to (0.25,0) to [out=up, in=-135] (1,1);
\draw [arrow data={0.34}{>}]  (1,1) to [out=135, in=right]  (-0.3, 1.7) to [out=left, in=up] (-1.25,1) to (-1.25,-0.5);
\draw(1.75,2.5) to (1.75,2) to [out=down, in=45] (1,1);
\draw[arrow data={0.32}{<}] (1,1) to [out= -45, in= left]   (2.3,0.3) to [out=right, in=down] (3.25,1) to (3.25,2.5);
\node [zxnode=\zxwhite] at (1,1) {$Q$};
\end{tz}
\text{ is inverse to }
\begin{tz}[zx,every to/.style={out=up, in=down},xscale=0.8,scale=0.8]
\draw [arrow data={0.2}{>},arrow data={0.8}{>}]  (0,0) to (2.25,3);
\draw (2.25,0) to node[zxnode=\zxwhite, pos=0.5] {$P$} (0,3);
\end{tz}
&
\begin{tz}[zx,every to/.style={out=up, in=down},xscale=0.8,scale=0.8]
\draw [arrow data={0.2}{<},arrow data={0.8}{<}]  (0,0) to (2.25,3);
\draw (2.25,0) to node[zxnode=\zxwhite, pos=0.5] {$Q$} (0,3);
\end{tz}
\text{ is inverse to }
\begin{tz}[zx,xscale=0.6,yscale=-1,scale=0.8]
\draw(0.25,-0.5) to (0.25,0) to [out=up, in=-135] (1,1);
\draw [arrow data={0.34}{<}]  (1,1) to [out=135, in=right]  (-0.3, 1.7) to [out=left, in=up] (-1.25,1) to (-1.25,-0.5);
\draw(1.75,2.5) to (1.75,2) to [out=down, in=45] (1,1);
\draw[arrow data={0.32}{>}] (1,1) to [out= -45, in= left]   (2.3,0.3) to [out=right, in=down] (3.25,1) to (3.25,2.5);
\node [zxnode=\zxwhite] at (1,1) {$P$};
\end{tz}
\end{calign}
\noindent
We have therefore shown that $P$ is bi-invertible.

$2.\Rightarrow 4.$ Suppose that $P:H\otimes A \to B \otimes H$ is a bi-invertible quantum function. We now demonstrate that it fulfills the second equation in~\eqref{eq:monadmap}:

\[
\def\l{2.25}
\def\h{1}
\def\top{5.5}
\def\bottom{-2}
\def\bottomright{2}
\def\topleft{-0.5}
\begin{tz}[zx,every to/.style={out=up, in=down},xscale=-1]
\draw[string] (0.75,0.25) to node[zxnode=\zxwhite, pos=0.26] {$P$} (0.75,\top);
\node[zxvertex=\zxwhite] at (0.75,0.25){};
\draw[string,arrow data={0.2}{>}, arrow data={0.8}{>}] (\bottomright, \bottom) to (\topleft, \top);
\end{tz}
~=~~
\def\l{2.25}
\def\h{1}
\def\top{5.5}
\def\bottom{-2}
\def\bottomright{2}
\def\topleft{-0.5}
\begin{tz}[zx,every to/.style={out=up, in=down},xscale=-1]
\node at (0,0) (C){};
\draw (C.center) to [out=down, in=down,looseness=2] node[zxvertex=\zxwhite, pos=0.5] {} (-\l, 0) to [out=up, in=down] (-2.25,\h);
\draw[string] (C.center) to [out=45, in=down] ([xshift=0.75cm, yshift=1cm]C.center) to node[zxnode=\zxwhite, pos=0.16] {$P$} (0.75,\top);
\draw[string] (C.center) to [out=135, in=down]  ([xshift=-0.75cm, yshift=1cm]C.center) to (-0.75,\h) to [out=up, in=down] (-2.25, \h+1) to [out=up, in=up] node[zxvertex=\zxwhite, pos=0.5]{} (-0.75,\h+1) to [out=down, in=up] (-2.25,\h);
\node[zxvertex=\zxwhite,zxup] at (C.center){};
\draw[string,arrow data={0.2}{>}, arrow data={0.8}{>}] (\bottomright, \bottom) to (\topleft, \top);
\end{tz}
~~=~
\def\l{2.25}
\def\h{3.5}
\def\top{5.5}
\def\bottom{-2}
\def\bottomright{2}
\def\topleft{-0.2}
\def \middle{2.3}
\begin{tz}[zx,every to/.style={out=up, in=down},xscale=-1]
\node at (0,0) (C){};
\draw (C.center) to [out=down, in=down,looseness=2] node[zxvertex=\zxwhite, pos=0.5] {} (-\l, 0) to [out=up, in=down] (-2.25,\h);
\draw[string] (C.center) to [out=45, in=down] ([xshift=0.75cm, yshift=1cm]C.center) to node[zxnode=\zxwhite, pos=0.0] {$P$} (0.75,\top);
\draw[string] (C.center) to [out=135, in=down]  ([xshift=-0.75cm, yshift=1cm]C.center) to node[zxnode=\zxwhite, pos=0.18] {$P$}node[zxnode=\zxwhite, pos=0.85] {$P^{-1}$}(-0.75,\h) to [out=up, in=down] (-2.25, \h+1) to [out=up, in=up] node[zxvertex=\zxwhite, pos=0.5]{} (-0.75,\h+1) to [out=down, in=up] (-2.25,\h);
\node[zxvertex=\zxwhite,zxup] at (C.center){};
\draw[string,arrow data={0.15}{>}, arrow data={0.64}{>},arrow data={0.9}{>}] (\bottomright, \bottom) to (\bottomright, -0.1) to (-1.5, \middle) to  (\topleft, \h+0.5) to (\topleft, \top);
\end{tz}
~~=~
\def\l{2.25}
\def\h{1.5}
\def\top{5.5}
\def\bottom{-2}
\def\bottomright{2}
\def\topleft{-0.2}
\def \middle{0.2}
\begin{tz}[zx,every to/.style={out=up, in=down},xscale=-1]
\node at (0,0) (C){};
\draw (C.center) to [out=down, in=down,looseness=2] node[zxnode=\zxwhite, pos=0.15] {$P$} node[zxvertex=\zxwhite, pos=0.5] {} (-\l, 0) to [out=up, in=down] (-2.25,\h);
\draw[string] (C.center) to [out=45, in=down] ([xshift=0.75cm, yshift=1cm]C.center) to  (0.75,\top);
\draw[string] (C.center) to [out=135, in=down]  ([xshift=-0.75cm, yshift=1cm]C.center) to node[zxnode=\zxwhite, pos=-0.2] {$P^{-1}$}(-0.75,\h) to [out=up, in=down] (-2.25, \h+1) to [out=up, in=up] node[zxvertex=\zxwhite, pos=0.5]{} (-0.75,\h+1) to [out=down, in=up] (-2.25,\h);
\node[zxvertex=\zxwhite,zxup] at (C.center){};
\draw[string,arrow data={0.15}{>}, arrow data={0.4}{>},arrow data={0.9}{>}] (\bottomright, \bottom) to  (-1., \middle) to  (\topleft, \h+0.5) to (\topleft, \top);
\end{tz}
~~=~
\def\l{2.25}
\def\h{3.5}
\def\top{5.5}
\def\bottom{-2}
\def\bottomright{2}
\def\topleft{-0.2}
\def \middle{2.3}
\def\xgap{-0.9cm}
\def\loose{3}
\begin{tz}[zx,every to/.style={out=up, in=down},xscale=-1]
\node at (0,0) (C){};
\draw (C.center) to [out=down, in=down,looseness=2] node[zxvertex=\zxwhite, pos=0.5] {} (-\l, 0) to [out=up, in=down] (-2.25,\h);
\draw[string] (C.center) to [out=45, in=down] ([xshift=0.75cm, yshift=1cm]C.center) to (0.75,\top);
\draw[string] (C.center) to [out=135, in=down]  ([xshift=-0.75cm, yshift=1cm]C.center) to node[zxnode=\zxwhite, pos=0.1](b) {$P^{-1}$}node[zxnode=\zxwhite, pos=0.7](t) {$P$}(-0.75,\h) to [out=up, in=down] (-2.25, \h+1) to [out=up, in=up] node[zxvertex=\zxwhite, pos=0.5]{} (-0.75,\h+1) to [out=down, in=up] (-2.25,\h);
\draw[arrow data={0.5}{>}][string] (t.center) to [looseness=\loose,out=135, in=up]([xshift=\xgap]t.center) to ([xshift=\xgap]b.center) to [looseness=\loose,out=down, in=-135] (b.center);
\node[zxvertex=\zxwhite,zxup] at (C.center){};
\draw[string,arrow data={0.25}{>}] (\bottomright, \bottom) to [in=-45] (t.center);
\draw[string,arrow data = {0.7}{>}] (b.center) to [out=45]  (\topleft, \top);
\end{tz}
\]

\def\l{2.25}
\def\h{3}
\def\top{5.5}
\def\bottom{-2}
\def\bottomright{2}
\def\topleft{-0.5}
\def\loop{-0.5cm}
\def\looploose{2.5}
\[=~\begin{tz}[zx,every to/.style={out=up, in=down},xscale=-1]
\node at (0,0) (C){};
\draw (C.center) to [out=down, in=down,looseness=2] node[zxvertex=\zxwhite, pos=0.5] {} (-\l, 0) to [out=up, in=down] (-2.25,\h);
\draw[string] (C.center) to [out=45, in=down] ([xshift=0.75cm, yshift=1cm]C.center) to (0.75,\top);
\draw[string] (C.center) to [out=135, in=down]  ([xshift=-0.75cm, yshift=1cm]C.center) to (-0.75,\h) to [out=up, in=down] (-2.25, \h+1) to [out=up, in=up] node[zxvertex=\zxwhite, pos=0.5]{} (-0.75,\h+1) to [out=down, in=up] (-2.25,\h);
\node[zxvertex=\zxwhite,zxup] at (C.center){};
\coordinate (loopcenter) at (0.05,2);
\draw[string,arrow data={0.2}{>}, arrow data={0.8}{>}] (\bottomright, \bottom) to  [in=-45] (loopcenter) to [looseness=\looploose,out=135, in=up] ([xshift=\loop]loopcenter) to [looseness=\looploose,out=down, in=-135] (loopcenter) to [out=45] (\topleft, \top);
\end{tz}
~~=~
\def\l{2.25}
\def\h{1}
\def\top{5.5}
\def\bottom{-2}
\def\bottomright{2}
\def\topleft{-0.5}
\begin{tz}[zx,every to/.style={out=up, in=down},xscale=-1]
\draw[string] (0.75,4) to (0.75,\top);
\node[zxvertex=\zxwhite] at (0.75,4){};
\draw[string,arrow data={0.2}{>}, arrow data={0.8}{>}] (\bottomright, \bottom) to (\topleft, \top);
\end{tz}
\]
Here, the first equation follows from the axioms of symmetric special Frobenius algebras, the second equation is invertibility of $P$, the third equation uses the fact that $P$ is a quantum function, the fourth equation follows from the graphical calculus moving $P$ along the right wire onto the top of $P^{-1}$. Finally, in the first equation of the second line, we use bi-invertibility of $P$. 

We now show that if $P$ is invertible, then the second equation in \eqref{eq:monadmap} implies the first equation. 
\[
\def\scl{0.85}
\begin{tz}[zx,scale=-1,scale=\scl,xscale=-1]
\clip (-3.05,-0.2) rectangle (1.8,4.2);
\draw[arrow data={0.1}{<},arrow data={0.5}{<}, arrow data={0.95}{<}] (1.75,0) to [out=up,in=-45] (1,1) to (-1,3) to [out=135, in=down] (-1.75,4);
\draw (0.25,0) to [out=up, in=-135] (1,1) to [in looseness=2.5,out=45, in=65] node[pos=0.72](c){} (-1,3) to [out=-135, in=right] (-2,2) to [out=left, in=down] (-3,4) ;
\node[zxnode=\zxwhite] at (1,1) {$P$};
\node[zxnode=\zxwhite] at (-1,3) {$P$};
\node[zxvertex=\zxwhite] at (c.center){};
\node[zxvertex=\zxwhite] at (-2,2){};
\end{tz}
~~\superequals{\eqref{eq:quantumfunction}~\&~\eqref{eq:cupcapfrob}}~~~~
\begin{tz}[zx,scale=-1,scale=\scl,xscale=-1]
\draw  (1.75,0) to (1.75,1) to[out=up, in=-45] (1,2) to [out=-135, in=up] (0.25,1) to [looseness=1.5,out=down,  in=down] node[pos=0.5] (b){}(-1.25,1) to (-1.25,4);
\draw (1,2) to [looseness=0.4,out=up, in=down] (1.75,3.25);
\draw[arrow data={0.4}{<},arrow data={0.9}{<}] (3.5,0) to [looseness=0.5,out=up,in=down]node[zxnode=\zxwhite, pos=0.64] {$P$} (0.5,4);
\node[zxvertex=\zxwhite] at (b.center){};
\node[zxvertex=\zxwhite,zxup] at (1,2){};
\node[zxvertex=\zxwhite] at (1.75,3.25){};
\end{tz}
~~\superequalseq{eq:monadmap}~~
\begin{tz}[zx,scale=-1,scale=\scl,xscale=-1]
\draw (0,0) to (0, 2) to [looseness=1.5,out=up, in=up] node[zxvertex=\zxwhite,pos=0.5]{} (-1.5,2) to [looseness=1.5,out=down, in=down]node[zxvertex=\zxwhite,pos=0.5]{} (-3,2) to (-3,4);
\draw[string,arrow data={0.5}{<}] (1.5,0) to + (0,4);
\end{tz}
~~=~~
\begin{tz}[zx,xscale=1,scale=\scl]
\draw (0,0) to + (0,4);
\draw[string,arrow data={0.5}{>}] (1.5,0) to + (0,4);
\end{tz}
\hspace{0.5cm} \Rightarrow \hspace{0.5cm}
\begin{tz}[zx,every to/.style={out=up, in=down},xscale=0.8,scale=1]
\draw [arrow data={0.2}{>},arrow data={0.8}{>}]  (0,0) to (2.25,3);
\draw (2.25,0) to node[zxnode=\zxwhite, pos=0.5] {$P^{-1}$} (0,3);
\end{tz}
~~=~~
\begin{tz}[zx,xscale=-0.6,yscale=-1,scale=\scl]
\draw[arrow data={0.5}{<}] (0.25,-0.5) to (0.25,0) to [out=up, in=-135] (1,1);
\draw (1,1) to [out=135, in=right] node[zxvertex=\zxwhite, pos=1]{} (-0.3, 1.7) to [out=left, in=up] (-1.25,1) to (-1.25,-0.5);
\draw[arrow data={0.5}{>}] (1.75,2.5) to (1.75,2) to [out=down, in=45] (1,1);
\draw (1,1) to [out= -45, in= left] node[zxvertex=\zxwhite, pos=1] {} (2.3,0.3) to [out=right, in=down] (3.25,1) to (3.25,2.5);
\node [zxnode=\zxwhite] at (1,1) {$P$};
\end{tz} 
\]
Since $P^{-1}$ is also a left inverse of $P$, this implies the following:
\def\scl{0.7}
\begin{equation}\label{eq:proofbijection}\begin{tz}[zx,scale=\scl,xscale=1]
\clip (-2.05,-0.05) rectangle (4.05,4.05);
\draw(0,0) to [out=up, in=up, looseness=6.] node[zxvertex=\zxwhite,pos=0.5] {} (2,0);
\draw [string,arrow data={0.1}{>},arrow data={0.9}{>}] (-2,0) to [out=up, in=down]  node[zxnode=\zxwhite, pos=0.4] {$P$} node[zxnode=\zxwhite, pos=0.6] {$P$}(4,4);
\end{tz}
=
\begin{tz}[zx,scale=\scl,xscale=1]
\clip (-2.05,-0.05) rectangle (4.05,4.05);
\draw (0,0) to [out=up, in=up, looseness=2.] node[zxvertex=\zxwhite, pos=0.5] {}(2,0);
\draw[arrow data={0.2}{>}, arrow data={0.8}{>}] (-2,0) to [out=up, in=down]  (4,4);
\end{tz}
\end{equation}
We can then prove the first equation as follows:
\[
\begin{tz}[zx,every to/.style={out=up, in=down},scale=-1,xscale=-1]
\draw (0,0) to (0,2) to [out=45] (0.75,3);
\draw (0,2) to [out=135] (-0.75,3);
\draw[arrow data={0.2}{<}, arrow data={0.8}{<}] (1.75,0) to [looseness=0.9] node[zxnode=\zxwhite, pos=0.5] {$P$} (-1.75,2.5) to (-1.75,3);
\node[zxvertex=\zxwhite, zxdown] at (0,2){};
\end{tz}
~~\superequalseq{eq:Frobenius}~~
\begin{tz}[zx,every to/.style={out=up, in=down},scale=-1,xscale=-1]
\draw (0,0) to (0,1.7) to [in=-135] (0.5,2.3) to (0.5,3);
\draw (0.5,2.3) to [out=-45, in=up] (1,1.7) to [out=down, in=down, looseness=1.5] node[zxvertex=\zxwhite, pos=0.5] {}(1.75,1.7) to (1.75,3);
\draw[string,arrow data={0.2}{<}, arrow data={0.8}{<}] (1.75,0) to [looseness=0.9] node[zxnode=\zxwhite, pos=0.5] {$P$} (-1.75,2.5) to (-1.75,3);
\node[zxvertex=\zxwhite, zxup] at (0.5,2.3){};
\end{tz}
~~~~~\superequals{\eqref{eq:quantumfunction}~\&~\eqref{eq:proofbijection}}~~
\begin{tz}[zx,every to/.style={out=up, in=down},scale=-1,xscale=-1]
\draw (-1,0) to (-1,0.7) to [in=-135] node[pos=1] (a) {}(-0.5,1.3) to (-0.5,3);
\draw[string] (-0.5,1.3) to [out=-45, in=up] (0,0.7) to [out=down, in=down, looseness=1.5] node[zxvertex=\zxwhite, pos=0.5] {}(0.75,0.7) to (0.75,3);
\draw[string,arrow data={0.2}{<}, arrow data={0.9}{<}] (1.75,0) to (1.75,0.75) to  [looseness=1.1, in looseness=0.9] node[zxnode=\zxwhite, pos=0.34] {$P$} node[zxnode=\zxwhite, pos=0.6] {$P$}(-1.75,3);
\node[zxvertex=\zxwhite,zxup] at (a.center){};
\end{tz}
~~\superequalseq{eq:Frobenius}~~
\begin{tz}[zx,every to/.style={out=up, in=down},scale=-1,xscale=-1]
\draw (0,0) to (0,0.75) to [out=45] (0.75,1.75) to (0.75,3);
\draw (0,0.75) to [out=135] (-0.75,1.75) to (-0.75,3);
\draw[arrow data={0.2}{<}, arrow data={0.9}{<}] (1.75,0) to (1.75,0.75) to  [looseness=1.1, in looseness=0.9] node[zxnode=\zxwhite, pos=0.36] {$P$} node[zxnode=\zxwhite, pos=0.64] {$P$}(-1.75,3);
\node[zxvertex=\zxwhite, zxdown] at (0,0.75){};
\end{tz}
\]

$4.\Rightarrow 3.$ \[
\begin{tz}[zx,every to/.style={out=up, in=down},scale=1,xscale=-1]
\draw (0.25,0) to [out=up, in=-135] (1,1) to [looseness=1.5,out=45, in=-45] (1,3) to [out=135, in=down] (0.25,4);
\draw[arrow data={0.1}{>}, arrow data={0.5}{>}, arrow data={0.95}{>}] (1.75,0) to [out=up, in=-45] (1,1) to [looseness=1.5, out=135, in=-135] (1,3) to [out=45, in=down] (1.75,4);
\node[zxnode=\zxwhite] at (1,1) {$P$};
\node[zxnode=\zxwhite] at (1,3) {$P^\dagger$};
\end{tz}
~~\superequalseq{eq:quantumfunction}~~
\begin{tz}[zx,xscale=-1]
\clip (-3.05,-0.2) rectangle (1.8,4.2);
\draw[arrow data={0.1}{>},arrow data={0.5}{>}, arrow data={0.95}{>}] (1.75,0) to [out=up,in=-45] (1,1) to (-1,3) to [out=135, in=down] (-1.75,4);
\draw (0.25,0) to [out=up, in=-135] (1,1) to [in looseness=2.5,out=45, in=65] node[pos=0.72](c){} (-1,3) to [out=-135, in=right] (-2,2) to [out=left, in=down] (-3,4) ;
\node[zxnode=\zxwhite] at (1,1) {$P$};
\node[zxnode=\zxwhite] at (-1,3) {$P$};
\node[zxvertex=\zxwhite] at (c.center){};
\node[zxvertex=\zxwhite] at (-2,2){};
\end{tz}
~~\superequalseq{eq:monadmap}~~
\begin{tz}[zx,xscale=-1]
\draw  (1.75,0) to (1.75,1) to[out=up, in=-45] (1,2) to [out=-135, in=up] (0.25,1) to [looseness=1.5,out=down,  in=down] node[pos=0.5] (b){}(-1.25,1) to (-1.25,4);
\draw (1,2) to [looseness=0.4,out=up, in=down] (1.75,3.25);
\draw[arrow data={0.4}{>},arrow data={0.9}{>}] (3.5,0) to [looseness=0.5,out=up,in=down]node[zxnode=\zxwhite, pos=0.64] {$P$} (0.5,4);
\node[zxvertex=\zxwhite] at (b.center){};
\node[zxvertex=\zxwhite,zxdown] at (1,2){};
\node[zxvertex=\zxwhite] at (1.75,3.25){};
\end{tz}
~~\superequalseq{eq:quantumfunction}~~
\begin{tz}[zx,xscale=-1]
\draw (0,0) to (0, 2) to [looseness=1.5,out=up, in=up] node[zxvertex=\zxwhite,pos=0.5]{} (-1.5,2) to [looseness=1.5,out=down, in=down]node[zxvertex=\zxwhite,pos=0.5]{} (-3,2) to (-3,4);
\draw[string,arrow data={0.5}{>}] (1.5,0) to + (0,4);
\end{tz}
~~=~~
\begin{tz}[zx,xscale=-1]
\draw (0,0) to + (0,4);
\draw[string,arrow data={0.5}{>}] (1.5,0) to + (0,4);
\end{tz}
\]
\[
\begin{tz}[zx,every to/.style={out=up, in=down},xscale=1]
\draw (0.25,0) to [out=up, in=-135] (1,1) to [looseness=1.5,out=45, in=-45] (1,3) to [out=135, in=down] (0.25,4);
\draw[arrow data={0.1}{>}, arrow data={0.5}{>}, arrow data={0.95}{>}] (1.75,0) to [out=up, in=-45] (1,1) to [looseness=1.5, out=135, in=-135] (1,3) to [out=45, in=down] (1.75,4);
\node[zxnode=\zxwhite] at (1,1) {$P^\dagger$};
\node[zxnode=\zxwhite] at (1,3) {$P$};
\end{tz}
~~\superequalseq{eq:quantumfunction}~~
\begin{tz}[zx,scale=-1,xscale=-1]
\clip (-3.05,-0.2) rectangle (1.8,4.2);
\draw[arrow data={0.1}{<},arrow data={0.5}{<}, arrow data={0.95}{<}] (1.75,0) to [out=up,in=-45] (1,1) to (-1,3) to [out=135, in=down] (-1.75,4);
\draw (0.25,0) to [out=up, in=-135] (1,1) to [in looseness=2.5,out=45, in=65] node[pos=0.72](c){} (-1,3) to [out=-135, in=right] (-2,2) to [out=left, in=down] (-3,4) ;
\node[zxnode=\zxwhite] at (1,1) {$P$};
\node[zxnode=\zxwhite] at (-1,3) {$P$};
\node[zxvertex=\zxwhite] at (c.center){};
\node[zxvertex=\zxwhite] at (-2,2){};
\end{tz}
~~\superequalseq{eq:quantumfunction}~~
\begin{tz}[zx,scale=-1,xscale=-1]
\draw  (1.75,0) to (1.75,1) to[out=up, in=-45] (1,2) to [out=-135, in=up] (0.25,1) to [looseness=1.5,out=down,  in=down] node[pos=0.5] (b){}(-1.25,1) to (-1.25,4);
\draw (1,2) to [looseness=0.4,out=up, in=down] (1.75,3.25);
\draw[arrow data={0.4}{<},arrow data={0.9}{<}] (3.5,0) to [looseness=0.5,out=up,in=down]node[zxnode=\zxwhite, pos=0.64] {$P$} (0.5,4);
\node[zxvertex=\zxwhite] at (b.center){};
\node[zxvertex=\zxwhite,zxup] at (1,2){};
\node[zxvertex=\zxwhite] at (1.75,3.25){};
\end{tz}
~~\superequalseq{eq:monadmap}~~
\begin{tz}[zx,scale=-1,xscale=-1]
\draw (0,0) to (0, 2) to [looseness=1.5,out=up, in=up] node[zxvertex=\zxwhite,pos=0.5]{} (-1.5,2) to [looseness=1.5,out=down, in=down]node[zxvertex=\zxwhite,pos=0.5]{} (-3,2) to (-3,4);
\draw[string,arrow data={0.5}{<}] (1.5,0) to + (0,4);
\end{tz}
~~=~~
\begin{tz}[zx]
\draw (0,0) to + (0,4);
\draw[string,arrow data={0.5}{>}] (1.5,0) to + (0,4);
\end{tz}
\]

$3. \Rightarrow 1.$ Suppose that $P:H \otimes A \to B \otimes H$ is unitary. Then, the following linear map is a quantum function:
\[
\begin{tz}[zx,every to/.style={out=up, in=down},xscale=0.8,scale=0.8]
\draw [arrow data={0.2}{<},arrow data={0.8}{<}]  (0,0) to (2.25,3);
\draw (2.25,0) to node[zxnode=\zxwhite, pos=0.5] {$Q$} (0,3);
\end{tz}
~:=~
\begin{tz}[zx,xscale=0.6,yscale=1,scale=0.8]
\draw(0.25,-0.5) to (0.25,0)     to [out=up, in=-135] (1,1);
\draw [arrow data={0.34}{>}]  (1,1) to [out=135, in=right]  (-0.3, 1.7) to [out=left, in=up] (-1.25,1) to (-1.25,-0.5);
\draw(1.75,2.5) to (1.75,2) to [out=down, in=45] (1,1);
\draw[arrow data={0.32}{<}] (1,1) to [out= -45, in= left]   (2.3,0.3) to [out=right, in=down] (3.25,1) to (3.25,2.5);
\node [zxnode=\zxwhite] at (1,1) {$P^{\dagger}$};
\end{tz}
~=~
\begin{tz}[zx,xscale=0.6,xscale=1,scale=0.8]
\draw[arrow data={0.68}{>}]  (4.75,2.5) to (4.75, 0.5) to [out=down, in=right] (2.65,-0.6) to [out=left, in=down] (0.55,0.5);
\draw (0.55,0.5) to [out=up, in=-135] (1,1);
\draw  (1,1) to [out=135, in=right]  (-0.3, 1.7) to [out=left, in=up]  node[zxvertex=\zxwhite, pos=0] {}(-1.25,1) to (-1.25,-0.5);
\draw[string, arrow data={0.68}{<}] (-2.75, -0.5) to (-2.75,1.5) to [out=up, in=left] (-0.65, 2.6) to [out=right, in=up] (1.45,1.5);
\draw  (1.45,1.5) to [out=down, in=45] (1,1);
\draw (1,1) to [out= -45, in= left]  node[zxvertex=\zxwhite, pos=1] {}  (2.3,0.3) to [out=right, in=down] (3.25,1) to (3.25,2.5);
\node [zxnode=\zxwhite] at (1,1) {$P$};
\end{tz} \]
For example, the counit condition of $Q$ follows from precomposing the counit condition of the quantum function $P$ with $P^\dagger$:
\[\begin{tz}[zx,every to/.style={out=up, in=down},xscale=-1]
\draw (0,0) to (0,2.25);
\draw[arrow data={0.2}{>}, arrow data={0.8}{>}] (1,0) to [looseness=0.9] node[zxnode=\zxwhite, pos=0.5] {$P$} (-1,2.5) to (-1,3);
\node[zxvertex=\zxwhite] at (0,2.25){};
\end{tz}
~~\superequalseq{eq:quantumfunction}~~
\begin{tz}[zx,every to/.style={out=up, in=down},xscale=-1]
\draw (0,0) to (0,0.75);
\draw[arrow data={0.2}{>}, arrow data={0.9}{>}] (1.,0) to (1.,0.75) to   (-1,3);
\node[zxvertex=\zxwhite, zxup] at (0,0.75){};
\end{tz}
\hspace{0.75cm}\super{\text{$P$ unitary}}{\Leftrightarrow} \hspace{0.75cm}
\begin{tz}[zx,every to/.style={out=up, in=down},xscale=1]
\draw (0,0) to (0,0.75);
\draw[arrow data={0.2}{>}, arrow data={0.9}{>}] (1.,0) to (1.,3) ;
\node[zxvertex=\zxwhite, zxup] at (0,0.75){};
\end{tz}
~~=~~
\begin{tz}[zx,every to/.style={out=up, in=down},xscale=-1]
\draw[arrow data ={0.15}{>}, arrow data={0.8}{>}] (0,0) to (1,2) to (1,3);
\draw (1,0) to (0,2);
\node[zxnode=\zxwhite] at (0.5,1) {$P^\dagger$};
\node[zxvertex=\zxwhite] at (0,2){};
\end{tz}
\hspace{0.75cm}\Leftrightarrow \hspace{0.75cm}
\begin{tz}[zx,every to/.style={out=up, in=down},xscale=-1]
\draw (0,0) to (0,0.75);
\draw[arrow data={0.2}{<}, arrow data={0.9}{<}] (1.,0) to (1.,0.75) to   (-1,3);
\node[zxvertex=\zxwhite, zxup] at (0,0.75){};
\end{tz}
~~=~~
\begin{tz}[zx,every to/.style={out=up, in=down},xscale=-1]
\draw (0,0) to (0,2.25);
\draw[arrow data={0.2}{<}, arrow data={0.8}{<}] (1,0) to [looseness=0.9] node[zxnode=\zxwhite, pos=0.5] {$Q$} (-1,2.5) to (-1,3);
\node[zxvertex=\zxwhite] at (0,2.25){};
\end{tz}
\]
The comultiplication condition of $Q$ can be proven similarly. The last equation in~\eqref{eq:quantumfunction} follows from a direct graphical argument, proving that $Q$ is indeed a quantum function. 

Moreover, unitarity of $P$ implies that $P$ and $Q$ fulfill equations~\eqref{eq:rightdual}. We note that equation~\eqref{eq:leftdual} is the dagger of equation~\eqref{eq:rightdual} and therefore redundant. This follows directly from the fact that $P$ and $Q$ are quantum functions and $\QSet$ is a dagger 2-category (see equation~\eqref{eq:qsetdagger}). Therefore, $Q$ is a dagger dual of $P$ in $\QSet$.\qedhere
\end{proof}

\begin{remark}\label{rem:lemma42fromlemma46} Proposition~\ref{prop:Frobeniusinvertible} follows as the one-dimensional case of Theorem~\ref{thm:dual2}. In fact, following our philosophy outlined in Section~\ref{sec:quantumset}, several of the steps of the proof of Theorem~\ref{thm:dual2} were first devised for the simpler case of Proposition~\ref{prop:Frobeniusinvertible}. The oriented wire was added in later, in an essentially unique way.
\end{remark}

\begin{remark} It is a direct consequence of Theorem~\ref{thm:dual2} that a quantum function $(H,P)$ is dagger-dualisable (or equivalently bi-invertible) if and only if it is biunitary as defined in~\cite{Reutter2016}. 
\end{remark}

\noindent
We denote the 2-category of quantum sets, quantum bijections and intertwiners by $\QBij$.\looseness=-2

\begin{remark}\label{rem:monoidalsubcats}
Composition of quantum functions makes the category $\QBij(B,B)$ of \emph{quantum permutations} on a quantum set $B$ into a monoidal dagger category with dualisable objects\ignore{\footnote{We will show in Section~\ref{sec:semisimple} that these categories are moreover \emph{semisimple} with a possibly infinite number of simple objects. For a finite number of simple objects such categories are known as \emph{unitary fusion category}.}} (see Theorem~\ref{thm:dual2}). In particular, the categories of quantum permutations $\QBij([n], [n])$ on an $n$-element set $[n]$ can be thought of as a quantum version of the symmetric group $S_n$. 
\end{remark}

\subsection{Quantum bijections in noncommutative topology}
\label{sec:Wang}
The monoidal categories $\QBij(B,B)$ of quantum bijections on a quantum set $B$ have been considered in finite noncommutative topology.

In~\cite{Wang1998}, Wang introduced `quantum symmetry groups of finite quantum spaces' as non-commutative variants of the symmetric groups $S_n$. We now show that our category of quantum permutations on a quantum set is precisely the category of finite-dimensional representations of the Hopf $C^*$-algebra corresponding to Wang's quantum symmetry group. In other words, our quantum permutations are quantum elements of this quantum group (cf. Theorem~\ref{thm:universalprop}).

\begin{prop}\label{prop:Wang} For a quantum set $B$, $\QBij(B,B)$ is the category of finite-dimensional representations of Wang's `quantum symmetry group' algebra~$A_{aut}(B)$.
\end{prop}
\begin{proof} We only prove the proposition for classical sets $[n]$; the proof for general quantum sets is completely analogous, only involving more indices.

In~\cite[Theorem 3.1]{Wang1998}, Wang defines quantum permutation groups of classical finite sets $[n]$ as $C^*$-algebras $A_{aut}([n])$ with generators $a_{i,j}$ $(i,j=1,\ldots n)$ and relations:
\begin{calign}\label{eq:Wangcondition}
a_{i,j}^2=a_{i,j} = a_{i,j}^* 
&
\sum_{i=1}^n a_{i,j} = 1, ~~\forall 1\leq j\leq n
&
\sum_{j=1}^n a_{i,j} =1, ~~\forall 1\leq i\leq n
\end{calign}
Noting that projectors summing to the identity are mutually orthogonal and comparing the relations~\eqref{eq:Wangcondition} with Theorem~\ref{thm:quantumfunctionsinprojectors} shows that our categories $\QBij([n], [n])$ are the categories of finite-dimensional representations of $A_{aut}([n])$.
\end{proof}

Proposition~\ref{prop:Wang} is not surprising: In Section~\ref{sec:universal}, we showed that $\QSet(B,B)$ is the category of finite-dimensional representations of the internal hom $[B,B]$ and in Section~\ref{sec:quantumbijection} we prove that $\QBij(B,B)$ is the full subcategory of dagger-dualisable objects in $\QSet(B,B)$. We therefore expect $\QBij(B,B)$ to arise as the category of finite-dimensional representations of some internal `automorphism Hopf $C^*$-algebra', which is exactly how Wang defines his compact quantum group algebra $A_{aut}(B)$. 
\ignore{
That this is Wang's algebra is unsurprising, given the results of Section \ref{sec:universal}. Indeed, we showed that $\QSet(A,A)$ is the category of finite-dimensional representations of the internal Hom $[A,A]$. $\QBij(A,A)$ is the subcategory of dagger-dualisable objects in this category; one would therefore expect $\QBij(A,A)$ to be obtained as the category of representations of some universal `dualisable' quotient of $[A,A]$, which is Wang's compact quantum group.}

\begin{remark}\label{rem:Wangnotdiscrete} As in Remark~\ref{rem:problemTannaka}, we cannot recover $A_{aut}(B)$ from the fibre functor $\QBij(B,B) \to \Hilb$. Indeed, from the perspective of locally compact quantum groups (see e.g.~\cite{Maes1998,Neshveyev2013}), we can think of $\QBij(B,B)$ as the category of finite-dimensional representations\ignore{corepresentations of the underlying Hopf $C^*$-algebra} of the discrete quantum group dual to Wang's compact quantum symmetry group.
Similarly to the case of ordinary discrete groups this discrete quantum group cannot generally be recovered from its (naive) representation category.
\end{remark}

\subsection{Quantum bijections between classical sets}\label{sec:quantumperm}

We now focus on quantum functions and quantum bijections between classical sets and show that we recover the \textit{magic unitaries} of Banica et al.~\cite{Banica2007} and the \textit{projective permutation matrices} of Atserias et al.~\cite{Atserias2016}. We first show that we can express quantum functions between classical sets as families of projectors satisfying certain \mbox{orthogonality conditions.}\looseness=-2
\begin{theorem}\label{thm:quantumfunctionsinprojectors} 
A quantum function $X\to Y$ between classical sets $X$ and $Y$ is exactly a family of projectors $\{P_{x,y}\}_{x\in X, y \in Y}$ on a Hilbert space $H$ such that the following holds, for all $x\in X$ and $y_1,y_2\in Y$:
\begin{calign}\label{eq:PPM1}P_{x,y_1} P_{x,y_2} = \delta_{y_1,y_2} P_{x,y_1} 
& 
\sum_{y\in Y} P_{x,y} = \mathbbm{1}_H
\end{calign}
A quantum bijection $X\to Y$ between classical sets $X$ and $Y$ is exactly a family of projectors $\{P_{x,y}\}_{x\in X, y \in Y}$ satisfying \eqref{eq:PPM1} and the following additional conditions, for all $x_1,x_2\in X$ and $y\in Y$:
\begin{calign}\label{eq:PPM2} 
P_{x_1,y} P_{x_2,y} = \delta_{x_1,x_2} P_{x_1,y} 
& 
\sum_{x\in X} P_{x,y} = \mathbbm{1}_H
\end{calign}
\end{theorem}
\begin{proof}
Every classical set $X$ corresponds to a commutative Frobenius algebra defined in Example~\ref{exm:Frob}, the elements of $X$ form a basis of copyable elements of  this algebra. The linear map $P:H\otimes X \to Y \otimes H$ may therefore be expressed in this basis:
:%
\[
\begin{tz}[zx]
\draw[arrow data={0.2}{>}, arrow data={0.8}{>}] (0,0) to (0,3);
\node[zxnode=\zxwhite] at (0,1.5) {$P_{x,y}$};
\end{tz}
~~:=
\begin{tz}[zx,every to/.style={out=up, in=down},xscale=0.8]
\draw [arrow data={0.2}{>},arrow data={0.8}{>}]  (0,0) to (2.25,3);
\draw (1.75,0.5) to node[zxnode=\zxwhite, pos=0.5] {$P$} (0.5,2.5);
\node[zxnode=\zxwhite] at (1.75,0.5){$x$};
\node[zxnode=\zxwhite] at (0.5,2.5){$y$};
\end{tz}\]
As an example, the first equation of \eqref{eq:quantumfunction}, expanded in the classical basis, becomes \mbox{$\delta_{y,y'}P_{x,y} = P_{x,y} P_{x,y'}$:}
\[
\delta_{y,y'}
\begin{tz}[zx,every to/.style={out=up, in=down},xscale=-1]
\draw (0,0) to (0,2.5);
\draw[arrow data={0.2}{>}, arrow data={0.8}{>}] (1.75,0) to [looseness=0.9] node[zxnode=\zxwhite, pos=0.5] {$P$} (-1.75,2.5) to (-1.75,3);
\node[zxnode=\zxwhite] at (0,2.5){$y$};
\node[zxnode=\zxwhite] at (0,0) {$x$};
\end{tz}
=
\begin{tz}[zx,every to/.style={out=up, in=down},xscale=-1]
\draw (0,0) to (0,2) to [out=45] (0.75,3);
\draw (0,2) to [out=135] (-0.75,3);
\draw[arrow data={0.2}{>}, arrow data={0.8}{>}] (1.75,0) to [looseness=0.9] node[zxnode=\zxwhite, pos=0.5] {$P$} (-1.75,2.5) to (-1.75,3);
\node[zxvertex=\zxwhite, zxup] at (0,2){};
\node[zxnode=\zxwhite] at (-0.75,3) {$y$};
\node[zxnode=\zxwhite] at (0.75,3) {$y'$};
\node[zxnode=\zxwhite] at (0,0) {$x$};
\end{tz}
\superequalseq{eq:quantumfunction}
\begin{tz}[zx,every to/.style={out=up, in=down},xscale=-1]
\draw (0,0) to (0,0.75) to [out=45] (0.75,1.75) to (0.75,3);
\draw (0,0.75) to [out=135] (-0.75,1.75) to (-0.75,3);
\draw[arrow data={0.2}{>}, arrow data={0.9}{>}] (1.75,0) to (1.75,0.75) to  [looseness=1.1, in looseness=0.9] node[zxnode=\zxwhite, pos=0.36] {$P$} node[zxnode=\zxwhite, pos=0.64] {$P$}(-1.75,3);
\node[zxvertex=\zxwhite, zxup] at (0,0.75){};
\node[zxnode=\zxwhite] at (-0.75,3) {$y$};
\node[zxnode=\zxwhite] at (0.75,3) {$y'$};
\node[zxnode=\zxwhite] at (0,0) {$x$};
\end{tz}
=
\begin{tz}[zx,every to/.style={out=up, in=down},xscale=-1]
\draw  (0.75,0.85) to (0.75,3);
\draw (-0.75,0.85) to (-0.75,3);
\draw[arrow data={0.2}{>}, arrow data={0.9}{>}] (1.75,0) to (1.75,0.75) to  [looseness=1.1, in looseness=0.9] node[zxnode=\zxwhite, pos=0.36] {$P$} node[zxnode=\zxwhite, pos=0.64] {$P$}(-1.75,3);
\node[zxnode=\zxwhite] at (-0.75,3) {$y$};
\node[zxnode=\zxwhite] at (0.75,3) {$y'$};
\node[zxnode=\zxwhite] at (0.75,0.85) {$x$};
\node[zxnode=\zxwhite] at (-0.75,0.85) {$x$};
\end{tz}
\]
All other expressions are similar translations of \eqref{eq:quantumfunction} and \eqref{eq:monadmap}.
\end{proof}
It is natural to arrange these projectors in an $|X| \times |Y|$ matrix. From this perspective, equation \eqref{eq:PPM1} states that the projectors along each row form a complete orthogonal family, while equation \eqref{eq:PPM2} requires this for each column.
In the work of Banica et al, matrices of projectors obeying both the row and the column equations are called \textit{magic unitaries}~\cite{Banica2007}, while in the work of Atserias et al they are called \textit{projective permutation matrices}~\cite{Atserias2016}. In this paper, we adopt the latter terminology. 

Theorem~\ref{thm:quantumfunctionsinprojectors} has the following immediate corollary.
\begin{corollary}\label{cor:quantumelementmeasurement}A quantum function between sets $X\to Y$ is a family of projective measurements with outcomes in $Y$, controlled by the set $X$. In particular, a quantum element of a set $X$ is a projective measurement with outcomes in $X$.
\end{corollary}
\begin{remark}\label{rem:ppmcomposition}
At the level of matrices of projectors, the composition of quantum functions $P:X\to Y$ and $Q:Y \to Z$ (see equation~\eqref{eq:1composition}) takes the following form:
\begin{calign}(Q\circ P)_{x,z} =\sum_{y \in Y} Q_{y,z}\otimes P_{x,y}\end{calign}
This coincides with composition of quantum functions in the Kleisli category of Abramsky et al~\cite{Abramsky2017}.
\end{remark}
\noindent
Examples of projective permutation matrices are provided by the first author's \textit{quantum Latin squares}~\cite{Musto2015,Musto2016}; more examples are given in Section~\ref{sec:semisimple}. 

We now show that quantum bijections only exist between classical sets of the \mbox{same cardinality.}
\begin{proposition} If there is a quantum bijection from a classical set $X$ to a classical set $Y$, then $|X| = |Y|$. In other words, every projective permutation matrix is square.
\end{proposition}
\begin{proof} By Theorem~\ref{thm:dual2}, a quantum bijection $P: H\otimes X \to Y \otimes H$ is unitary. In particular, $|X|\dim(H) = \dim(X\otimes H) =  \dim(H\otimes Y) =  |Y| \dim(H)$. 
\end{proof}

\begin{remark}\label{rem:teleportation}
We remark that there are quantum functions and quantum bijections between noncommutative algebras. We suggest that these structures might also play a role in quantum information theory.  For example, unitary error bases\footnote{A \textit{unitary error basis} is a basis of unitary operators on a finite-dimensional Hilbert space, orthogonal with respect to the trace inner product.} --- providing the basic data for quantum teleportation and dense coding schemes~\cite{Werner2001} --- give rise to quantum bijections from the matrix algebra $\mathrm{Mat}_n$ to the set $[n^2]$ with $n^2$ elements. Similar correspondences have been noted in the work of Stahlke~\cite{Stahlke2016}.
\end{remark}

\section{Quantum graph theory}\label{sec:qgt}
In the first sections of this paper, we have presented a quantisation of the theory of finite sets and functions. In the following, we continue this general approach and quantise the theory of finite undirected graphs, unifying and extending the work of various authors on nonlocal games~\cite{Cameron2007,Mancinska2016,Atserias2016}, noncommutative topology~\cite{Banica2005,Banica2009,Banica2007_2,Bichon2003},  operator algebras~\cite{Weaver2010,Kuperberg2012,Weaver2015} and quantum information~\cite{
Duan2013,Stahlke2016}.

By analogy with classical graphs, we define a quantum graph as a quantum set of vertices together with a quantum adjacency matrix.\footnote{In Section~\ref{app:quantumrelation}, we show that these quantum graphs also arise from quantising the edge relation on the vertex set of a graph. This relational approach is closer to the work of previous authors, whereas the adjacency matrix approach has a stronger resemblance to classical graph theory.} We define a quantum graph homomorphism as a quantum function preserving the quantum adjacency matrix, and show that quantum graphs and quantum graph homomorphisms naturally form a concrete dagger $2$-category $\QGraph$. We then define quantum graph isomorphisms, and show that these are precisely the dagger-dualisable $1$-morphisms in $\QGraph$. \ignore{This fact may be used to classify graph pairs exhibiting quantum advantage in the quantum graph isomorphism game of Atserias et al~\cite{Atserias2016}, as will be shown in future work~\cite{paper1b}.}
\ignore{Here we content ourselves with recovering basic definitions, leaving a deeper study of quantum graph theory for future work.\looseness=-2}

For a classical graph $G$, we denote its set of vertices by $V_G$; for vertices $v,w\in V_G$ we write $v\sim_G w$ if $v$ and $w$ are connected in $G$. We consider only undirected graphs.

\subsection{Quantum graphs via adjacency matrices}
\label{sec:qgraphsbyadjmats}
A classical graph $G$ can be described by its vertex set $V_G$ and its \emph{adjacency matrix}, which is a linear map $G:V_G\to V_G$ (see Terminology~\ref{not:setalgebra}) with coefficients $G_{v,w}=1$ if $v\sim_G w$ and $G_{v,w} = 0$ otherwise.
 \ignore{is defined by its vertex set and \textit{adjacency matrix}, which is self-adjoint, symmetric, and has entries valued in the set $\{0,1\}$, where all entries are 1 along the diagonal.} %
 
This suggests the following definition of a quantum graph as a quantum set equipped with a quantum adjacency matrix. 
\ignore{with coefficients $G_{v,w}=1$ if $v\sim_G w$ and $G_{v,w} = 0$ otherwise.}
\begin{definition}\label{def:quantumgraphsbyadjmats}
A \emph{quantum graph} is a pair $(A,G)$ of a quantum set $A$ (the \emph{quantum set of vertices}) and a self-adjoint linear map $G:A \to A$ (the \emph{quantum adjacency matrix}) satisfying the following equations:%
\begin{calign}\label{eq:propadjacency}
\begin{tz}[zx]
\draw (0,0) to (0,0.5) to [out=135, in=-135,looseness=1.5] node[zxnode=\zxwhite, pos=0.5] {$G$} (0,2.5) to (0,3);
\draw[string] (0,0.5) to [out=45, in=-45,looseness=1.5] node[zxnode=\zxwhite, pos=0.5] {$G$} (0,2.5);
\node[zxvertex=\zxwhite,zxup] at (0,0.5){};
\node[zxvertex=\zxwhite,zxdown] at (0,2.5){};
\end{tz}\quad  = \quad 
\begin{tz}[zx]
\draw (0,0) to node[zxnode=\zxwhite, pos=0.5] {$G$} (0,3);
\end{tz}
&
\begin{tz}[zx]
\draw (1,0) to (1,1.5) to [out=up, in=up, looseness=2.5]node[zxvertex=\zxwhite, pos=0.5] {} (0,1.5) to [out=down, in=down, looseness=2.5]node[zxvertex=\zxwhite, pos=0.5] {} (-1,1.5) to (-1,3);
\node[zxnode=\zxwhite] at (0,1.5) {$G$};
\end{tz}
\quad = \quad 
\begin{tz}[zx]
\draw (0,0) to (0,3);
\node[zxnode=\zxwhite] at (0,1.5) {$G$};
\end{tz}%
\ignore{%
&
\begin{tz}[zx]
\draw (0,0) to (0,0.5) to [out=135, in=-135,looseness=1.5] node[zxnode=\zxwhite, pos=0.5] {$G$} (0,2.5) to (0,3);
\draw[string] (0,0.5) to [out=45, in=-45,looseness=1.5](0,2.5);
\node[zxvertex=\zxwhite,zxup] at (0,0.5){};
\node[zxvertex=\zxwhite,zxdown] at (0,2.5){};
\end{tz}\quad  = \quad 
\begin{tz}[zx]
\draw (0,0) to (0,3);
\end{tz}}
\end{calign}
\end{definition}
\noindent
For a classical set $A=V_G$, this reduces to the definition of an adjacency matrix $\{G_{v,w}\}_{v,w\in V_G}$; from left to right, the conditions state that $G_{v,w}^2 = G_{v,w}$, and that $G_{v,w} = G_{w,v}$.

A quantum graph is \emph{reflexive} or \emph{irreflexive} if one of the following additional equations holds:
\begin{equation*}
\includegraphics[trim={1cm 0 1cm .35cm},valign=c]{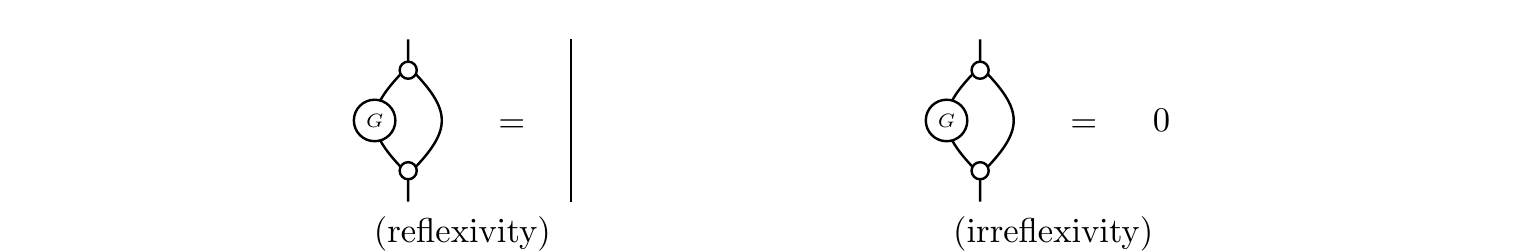}
\end{equation*}
\noindent 
For classical graphs this corresponds to $G_{v,v}=1$ or $G_{v,v}=0$, respectively. For a classical set, the definition of an irreflexive quantum graph therefore reduces to the standard definition of an adjacency matrix of a simple graph.

\begin{remark}\label{rem:literatureqgraph}There are many related definitions of noncommutative or quantum graphs in the literature~\cite{Kuperberg2012,Weaver2010,Weaver2015,Duan2013}, with applications in quantum error correction \cite{Weaver2015} and zero-error communication~\cite{Duan2013}. In Section~\ref{app:quantumrelation}, we comment on how our quantum graphs correspond to these previous definitions. In particular, we prove the following:\looseness=-2
\begin{itemize}
\item Our reflexive quantum graphs coincide with Weaver's finite-dimensional quantum graphs~\cite{Weaver2015}, defined in terms of symmetric and reflexive quantum relations~\cite{Kuperberg2012,Weaver2010}.
\item Our reflexive quantum graphs $(\mathrm{Mat}_n, G)$ on matrix algebras coincide with Duan, Severini and Winter's noncommutative graphs~\cite{Duan2013}.
\end{itemize}

\end{remark}

\subsection{Quantum graph homomorphisms}
\label{sec:quantumgraphhomandiso}
Having defined quantum graphs, we now quantise graph homomorphisms. As with quantum functions and quantum sets, we have classical homomorphisms between quantum graphs, quantum homomorphisms between classical graphs and quantum homomorphisms between quantum graphs. 
We show that our quantum homomorphisms and quantum isomorphisms between classical graphs coincide with the quantum graph homomorphisms of Man\v{c}inska et al~\cite{Mancinska2016} and the quantum graph isomorphisms of Atserias et al~\cite{Atserias2016}. We also show that our classical and quantum homomorphisms between quantum graphs are `pure' versions of Stahlke's quantum graph homomorphisms~\cite{Stahlke2016}. 

To quantise graph homomorphisms, we first express them via Gelfand duality in terms of string diagrams in the category of finite-dimensional Hilbert spaces and linear maps. For a quantum graph --- a quantum set $A$ with quantum adjacency matrix $G:A\to A$ --- we introduce the following notation:
\begin{equation}\label{eq:graphprojector1}\begin{tz}[zx]
\draw[string] (0, 0) to (0,3);
\draw (2.5,0) to (2.5, 3);
\draw (0, 1.5) to (2.5, 1.5) ;
\node[zxvertex=\zxwhite] at (0,1.5){};
\node[zxvertex=\zxwhite] at (2.5,1.5){};
\node[zxnode=\zxwhite] at (1.25, 1.5 ) {$G$};
\end{tz}
~~~:=~~~
\begin{tz}[zx,xscale=2.5/3]
\draw (-1.5,0) to (-1.5,1.5) to [out=up, in=-135] (-0.75, 2.5) to (-0.75,3);
\draw (-0.75, 2.5) to [out=-45, in=up] (0, 1.5) to [out=down, in=135] (0.75,0.5) to (0.75,0);
\draw (0.75,0.5) to [out=45, in=down] (1.5,1.5) to (1.5,3);
\node[zxnode=\zxwhite] at (0, 1.5) {$G$};
\node[zxvertex=\zxwhite, zxdown] at (-0.75, 2.5) {};
\node[zxvertex=\zxwhite, zxup] at (0.75, 0.5) {};
\end{tz}
~~~\superequalseq{eq:propadjacency}~~~
\begin{tz}[zx,xscale=-2.5/3]
\draw (-1.5,0) to (-1.5,1.5) to [out=up, in=-135] (-0.75, 2.5) to (-0.75,3);
\draw (-0.75, 2.5) to [out=-45, in=up] (0, 1.5) to [out=down, in=135] (0.75,0.5) to (0.75,0);
\draw (0.75,0.5) to [out=45, in=down] (1.5,1.5) to (1.5,3);
\node[zxnode=\zxwhite] at (0, 1.5) {$G$};
\node[zxvertex=\zxwhite, zxdown] at (-0.75, 2.5) {};
\node[zxvertex=\zxwhite, zxup] at (0.75, 0.5) {};
\end{tz}
\end{equation}
If $(V_G,G)$ is a classical graph, it can easily be verified that this map~\eqref{eq:graphprojector1} is a projector onto the following subset (or equivalently onto the commutative subalgebra corresponding to this subset):
\begin{equation}\label{eq:projectorsubset1}\left\{ (v,w) ~|~ v \sim_G w \right\} \subseteq V_G \times V_G\end{equation}
We can use these projectors to express the notion of a classical graph homomorphism diagrammatically:

\begin{proposition} Let $(V_G,G)$ and $(V_H,H)$ be classical graphs. Under Gelfand duality, a graph homomorphism $G\to H$ corresponds to a $*$-cohomomorphism $f:V_G\to V_H$ fulfilling the following equation:
\def\d{2}
\begin{equation}\label{eq:classicalgraphhomomorphism}\begin{tz}[zx,yscale=0.8]
\draw[string] (0,-0.5) to node[front,zxnode=\zxwhite,pos=0.5] {$f$} (0,5);\draw[string] (\d,-0.5) to node[front,zxnode=\zxwhite, pos=0.5] {$f$} (\d,5);
\draw[string] (0, 0.5) to (\d,0.5);
\node[zxvertex=\zxwhite] at (0,0.5){};
\node[zxvertex=\zxwhite] at (\d,0.5){};
\node[zxnode=\zxwhite] at (\d/2, 0.5 ) {$G$};
\end{tz}
~~~=~~~
\begin{tz}[zx,yscale=0.8]
\draw[string] (0,-0.5) to node[front,zxnode=\zxwhite,pos=0.5] {$f$} (0,5);\draw[string] (\d,-0.5) to node[front,zxnode=\zxwhite, pos=0.5] {$f$} (\d,5);
\draw[string] (0, 0.5) to (\d,0.5);
\node[zxvertex=\zxwhite] at (0,0.5){};
\node[zxvertex=\zxwhite] at (\d,0.5){};
\node[zxnode=\zxwhite] at (\d/2, 0.5 ) {$G$};
\draw[string] (0, 4) to (\d,4);
\node[zxvertex=\zxwhite] at (0,4){};
\node[zxvertex=\zxwhite] at (\d,4){};
\node[zxnode=\zxwhite] at (\d/2, 4 ) {$H$};
\end{tz}
~~
\begin{tz}[zx]
\end{tz}
\end{equation}
\end{proposition}
\begin{proof} Since~\eqref{eq:graphprojector1} is a projector onto the subset~\eqref{eq:projectorsubset1}, equation~\eqref{eq:classicalgraphhomomorphism} is simply expressing that if $v\sim_G w$, then $f(v) \sim_H f(w)$.
\end{proof}
\noindent
We quantise this notion following our usual philosophy.
\begin{definition}\label{def:quantumgraphhom}
Let $(A,G)$ and $(B,H)$ be quantum graphs. A \textit{quantum homomorphism} $(H,P) : (A,G) \to (B,H)$ is a quantum function $(H,P): A \to B$ fulfilling the following additional equation:
\def\d{2}
\begin{equation}\label{eq:quantumgraphhomomorphism}\begin{tz}[zx,yscale=0.8,xscale=-1]
\draw[string] (0,-0.5) to node[front,zxnode=\zxwhite,pos=0.57] {$P$} (0,5);\draw[string] (\d,-0.5) to node[front,zxnode=\zxwhite, pos=0.43] {$P$} (\d,5);
\draw[string] (0, 0.5) to (\d,0.5);
\node[zxvertex=\zxwhite] at (0,0.5){};
\node[zxvertex=\zxwhite] at (\d,0.5){};
\node[zxnode=\zxwhite] at (\d/2, 0.5 ) {$G$};
\draw[string, arrow data={0.08}{>},arrow data={0.5}{>}, arrow data ={0.95}{>}] (3, -0.5)  to (3,0.5) to [out=up, in=down] (-1.,4) to (-1.,5);
\end{tz}
~=~
\begin{tz}[zx,yscale=0.8,xscale=-1]
\draw[string] (0,-0.5) to node[front,zxnode=\zxwhite,pos=0.57] {$P$} (0,5);\draw[string] (\d,-0.5) to node[front,zxnode=\zxwhite, pos=0.43] {$P$} (\d,5);
\draw[string] (0, 0.5) to (\d,0.5);
\node[zxvertex=\zxwhite] at (0,0.5){};
\node[zxvertex=\zxwhite] at (\d,0.5){};
\node[zxnode=\zxwhite] at (\d/2, 0.5 ) {$G$};
\draw[string] (0, 4) to (\d,4);
\node[zxvertex=\zxwhite] at (0,4){};
\node[zxvertex=\zxwhite] at (\d,4){};
\node[zxnode=\zxwhite] at (\d/2, 4 ) {$H$};
\draw[string, arrow data={0.08}{>}, arrow data ={0.5}{>},arrow data ={0.95}{>}] (3, -0.5)  to (3,0.5) to [out=up, in=down] (-1.,4) to (-1.,5);
\end{tz}
\end{equation}
\end{definition}
\vspace{5pt}
\begin{remark}
It is possible to have ordinary homomorphisms between quantum graphs (that is, {$*$-cohomomorphisms} obeying \eqref{eq:classicalgraphhomomorphism} or equivalently one-dimensional quantum homomorphisms), quantum homomorphisms between classical graphs and quantum homomorphisms between quantum graphs.
\end{remark}
\noindent
Expressing a quantum homomorphism between classical graphs in terms of its \mbox{underlying} family of projectors $\{P_{v,w}\!\}_{v\in V_G, w\in V_H}$, fulfilling~\eqref{eq:PPM1}, this condition becomes:\looseness=-2
\begin{equation}\label{eq:qgraphhomprojcond} v\sim_G v' \text{ and }w\not\sim_H w'\hspace{0.75cm} \Rightarrow \hspace{0.75cm} P_{v',w'} P_{v,w} = 0
\end{equation}%

\begin{remark}\label{rem:quantumhomMancinska}Man\v{c}inska and Roberson~\cite{Mancinska2016} define quantum graph homomorphisms between classical graphs $G$ and $H$ as perfect quantum strategies --- strategies involving measurements on a shared entangled resource --- for a certain bipartite nonlocal \emph{graph homomorphism game}. This work generalises the earlier quantum graph colouring game~\cite{Avis2006} and the corresponding notion of the quantum chromatic number of a graph~\cite{Cameron2007}.\footnote{The quantum chromatic number of a graph $G$ is the smallest $n$ for which there exists a quantum graph homomorphism $G$ to $K_n$, the complete graph on $n$ vertices.}

In~\cite{Mancinska2016}, the existence of such a perfect strategy is expressed in terms of a family of projectors $\{P_{v,w}\}_{v\in V_G, w\in V_H}$ which form a quantum homomorphism in the sense of Definition~\ref{def:quantumgraphhom}.

\begin{proposition}[{\cite[Lemma 2.3]{Mancinska2016}}]\label{prop:mancinskagraphhom} Given graphs $G$ and $H$, there is a perfect quantum strategy for the graph homomorphism game defined in~\cite{Mancinska2016}, if and only if there is a nonzero family of projectors $\{P_{v,w}\}_{v\in V_G, w\in V_H}$ fulfilling equation~\eqref{eq:PPM1} and~\eqref{eq:qgraphhomprojcond} --- or equivalently, if and only if there is a nonzero quantum homomorphism $ (H, P): (V_G,G)\to (V_H,H)$.
\end{proposition}

\end{remark}

\begin{remark} Quantum functions are themselves perfect quantum strategies for a \emph{function game}, whose deterministic classical strategies correspond to classical functions between finite sets. In this game, a verifier sends an element of a finite set $X$ to Alice and Bob, who in turn return an element of another set $Y$.  A strategy is perfect if Alice and Bob return the same element $y\in Y $ whenever they receive the same element $x\in X$. Of course, every such game admits a classical strategy.
\end{remark}

\ignore{
A normal form strategy is defined by the dimension of the shared state and the projective measurements performed by each party, which must obey the conditions given in the following proposition.
\begin{proposition}[\cite{Mancinska2016}, Lemma 2.3]\label{prop:mancinskagraphhom}
A quantum graph homomorphism between two graphs $G$ and $H$ exists if and only if there exists a system of projectors $P_{x,y}$ for $x \in V_G, y \in V_H$ satisfying the following conditions:
\begin{align*}
&\sum_{y \in V_H} P_{x,y} = \mathbbm{1} \\
&P_{x,y} P_{x,y'} = \delta_{y,y'} P_{x,y}\\
& x\sim_G x' \text{ \emph{and} }y\not\sim_H y'\hspace{0.1cm} \Rightarrow \hspace{0.1cm} P_{x',y'} P_{x,y} = 0
\end{align*}
\end{proposition}
We call systems of projectors satisfying the conditions of Proposition \ref{prop:mancinskagraphhom} \emph{normal form quantum graph homomorphisms}. They are perfect strategies for the graph homomorphism game when the parties share a maximally entangled state whose dimension is the dimension of the underlying Hilbert space of the projectors.
\begin{proposition}
Our quantum graph homomorphisms are precisely normal form quantum graph homomorphisms in the sense of the definition of Man\v{c}inska et al (Proposition \ref{prop:mancinskagraphhom}).
\end{proposition}
\begin{proof}
Clear; the first two equations of Proposition \ref{prop:mancinskagraphhom} define a quantum function between classical sets in its representation as a system of projectors \eqref{eq:PPM1}, and the last equation is precisely \eqref{eq:qgraphhomprojcond}.
\end{proof}
\noindent}

\begin{remark}\label{rem:stahlkegraphhoms}
Stahlke~\cite{Stahlke2016} defines homomorphisms~\cite[Definition 7]{Stahlke2016} and quantum homomorphisms~\cite[Definition 15]{Stahlke2016} (there called entanglement assisted homomorphisms) between quantum graphs on endomorphism algebras (cf. Definition~\ref{def:operatorsystem} and Proposition~\ref{prop:quantumgraphDuanthesame}), and relates them to zero-error strategies for quantum source-channel coding in various scenarios.

Stahlke's notions of homomorphism and quantum homomorphism are defined in terms of `mixed' completely positive trace preserving (CPTP) maps which, if restricted to the `pure' setting of $*$-homomorphisms, agree with our definitions.\footnote{More precisely, the pure $*$-homomorphism version of Stahlke's entanglement assisted homomorphisms agrees with our quantum homomorphisms if the entangled resource is a maximally entangled state; or, in other words, if the positive operator $\Lambda$ used to define entanglement assisted homomorphisms (see~\cite[Definition 15]{Stahlke2016}) is the identity on some finite-dimensional Hilbert space $V$.}

\ignore{
 between quantum graphs on matrix algebras 
There is another notion of quantum graph homomorphism due to Stahlke~\cite{Stahlke2016}, defined between the quantum graphs of Duan et al~\cite{Duan2013} (Definition \ref{def:duanquantumgraphs}) which we proved in Proposition \ref{prop:duandefissameasours} to be equivalent to our own. These quantum graph homomorphisms were shown to correspond to zero-error strategies for quantum source-channel coding in various scenarios~\cite{Stahlke2016}, and, in the case of classical graphs, the relationship between this definition and that of Man\v{c}inska et al has been studied by Ortiz and Paulsen~\cite{Ortiz2016}. Although Stahlke's notion is closely related to our own, a comparison is rendered somewhat complex by the fact that the graph homomorphisms of Stahlke are `mixed' CP maps, whereas we restrict to `pure' *-homomorphisms. Since we will not need Stahlke's notion of graph homomorphism for the immediate application of results in this paper to a classification of pseudotelepathic graphs, to appear in a separate paper~\cite{paper1b}, we will reserve consideration of the relationship between this important notion and our own for future work.}
\end{remark}

\subsection{Quantum graph isomorphisms}\label{sec:qisosofgraphs}
Recall that a classical \textit{graph isomorphism} is an invertible graph homomorphism whose inverse is also a graph homomorphism. This condition can be more concisely expressed as follows.

\begin{proposition} Let $G$ and $H$ be classical graphs. Under Gelfand duality, a graph isomorphism $G\to H$ corresponds to a $*$-isomorphism $f:V_G \to V_H$ fulfilling the \mbox{following equation:}
\begin{equation}\label{eq:classicalgraphiso}
\begin{tz}[zx,every to/.style={out=up, in=down},scale=1]
\draw (0,0) to (0,3);
\path (1.75,0) to (1.75,0.75) to  [looseness=1.1, in looseness=0.9]  node[zxnode=\zxwhite, pos=0.5] {$f$}(-1.75,3);
\node[zxnode=\zxwhite] at (0,0.9){$G$};
\end{tz}
~~~= ~~~
\begin{tz}[zx,every to/.style={out=up, in=down},scale=1]
\draw (0,0) to (0,3);
\path (1.75,0) to (1.75,0.75) to  [looseness=1.1, in looseness=0.9]  node[zxnode=\zxwhite, pos=0.5] {$H$}(-1.75,3);
\node[zxnode=\zxwhite] at (0,0.9){$f$};
\end{tz}
\end{equation}%
\end{proposition}%
\begin{proof} This is well-known, but will also follow as a corollary of Theorem~\ref{thm:dualisablequantumgraph}. \end{proof}

As in Section~\ref{sec:twocharacterisations}, we can define a quantum graph isomorphism either as a quantum bijection fulfilling a quantised version of~\eqref{eq:classicalgraphiso}, or as a dualisable quantum graph homomorphism in the 2-category \QGraph\ of quantum graphs, quantum graph homomomorphisms and intertwiners, which we will shortly define. In fact, we will show in Theorem~\ref{thm:dualisablequantumgraph} that, as with quantum bijections, both these approaches lead to the same notion of quantum graph isomorphism.

\begin{definition}\label{def:quantumgraphiso} Let $(A,G)$ and $(A',G')$ be quantum graphs. A \textit{quantum isomorphism} $(H,P): (A,G)\to (A',G')$ is a quantum bijection $(H,P): A \to A'$ fulfilling the following additional equation:%
\begin{calign}\label{eq:quantumgraphisomorphism}
\begin{tz}[zx,every to/.style={out=up, in=down},scale=1,xscale=-1]
\draw (0,0) to (0,3);
\draw[arrow data={0.2}{>}, arrow data={0.9}{>}] (1.75,0) to (1.75,0.75) to  [looseness=1.1, in looseness=0.9]  node[zxnode=\zxwhite, pos=0.5] {$P$}(-1.75,3);
\node[zxnode=\zxwhite] at (0,0.9){$G$};
\end{tz}
\quad =\quad 
\begin{tz}[zx,every to/.style={out=up, in=down},scale=1,xscale=-1]
\draw (0,0) to (0,3);
\draw[arrow data={0.2}{>}, arrow data={0.8}{>}] (1.75,0) to [looseness=0.9] node[zxnode=\zxwhite, pos=0.5] {$P$} (-1.75,2.5) to (-1.75,3);
\node[zxnode=\zxwhite] at (0,2.35) {$G'$};
\end{tz}
\end{calign}
\end{definition}
\noindent
Expressing quantum isomorphisms between classical graphs in terms of their underlying projective permutation matrix $\{P_{v,w}\}_{v\in V_G, w\in V_{G'}}$, which fulfils~\eqref{eq:PPM1} and~\eqref{eq:PPM2}, the quantum isomorphism condition \eqref{eq:quantumgraphisomorphism} becomes:
\begin{equation}\label{eq:ppmquantumgraphcondition1}
\forall~a\in V_{G},~b \in V_{G'}:\hspace{1cm}\sum_{i \in \mathrm{nbh}_{G}(a)} P_{i,b} = \sum_{j \in \mathrm{nbh}_{G'}(b)} P_{a,j}
\end{equation}
Here, $\mathrm{nbh}_{G}(a)$ and $\mathrm{nbh}_{G'}(b)$ denote the set of neighboring vertices of $a$ and $b$ in $G$ and $G'$, respectively. Equivalently, this condition can be expressed as follows: 
\begin{equation}\label{eq:ppmquantumgraphcondition2}\text{If }\left(v\sim_{G}v'\text{ and }w\not\sim_{G'}w'\right)\text{ or }\left(v\not\sim_{G} v'\text{ and }w\sim_{G'} w'\right)\hspace{0.5cm} \Rightarrow\hspace{0.5cm} P_{v',w'} P_{v,w} = 0
\end{equation}

\noindent
The equivalence between~\eqref{eq:ppmquantumgraphcondition1} and~\eqref{eq:ppmquantumgraphcondition2} was proven by Atserias et al.~\cite[Theorem 5.8]{Atserias2016}, but will also follow from Theorem~\ref{thm:dualisablequantumgraph} and the fact that a quantum isomorphism is a dualisable quantum homomorphism whose dual $\overline{P}$ is also a quantum homomorphism. 

\begin{remark}\label{rem:Atseriasquantumiso}
Atserias et al~\cite{Atserias2016} define quantum graph isomorphisms between classical graphs as perfect quantum strategies for the \emph{graph isomorphism game}, which is a symmetric version of the graph homomorphism game of Man\v{c}inska and Roberson (see Remark~\ref{rem:quantumhomMancinska}). 
\end{remark}
\noindent As with quantum graph homomorphisms, Atserias et al. show that such perfect quantum strategies exist if and only if there is a projective permutation matrix forming a quantum isomorphism in the sense of Definition~\ref{def:quantumgraphiso}.

\begin{proposition}[{\cite[Theorem 5.4]{Atserias2016}}]\label{prop:Atserias}Given classical graphs $G$ and $G'$, there is a perfect quantum strategy for the graph isomorphism game defined in~\cite{Atserias2016}, if and only if there is a nonzero family of projectors $\{P_{v,w}\}_{v\in V_G, w\in V_{G'}}$ fulfilling equations~\eqref{eq:PPM1},~\eqref{eq:PPM2} and~\eqref{eq:ppmquantumgraphcondition2} --- or equivalently, if and only if there is a nonzero quantum isomorphism $(H,P):(V_G,G) \to (V_{G'},G')$.\looseness=-2
\end{proposition}

\subsection{The 2-category $\QGraph$}\label{sec:qgraph}
We now show that these new structures again form a 2-category under composition, which we call $\QGraph$.
\begin{definition}\label{def:2catqgraph} The 2-category $\QGraph$ is built from the following structures:
\begin{itemize}
\item \textbf{objects} are quantum graphs $(A,G),(A',G')$, ...;
\item \textbf{1-morphisms} $(A,G)\!\begin{tikzpicture}\node at (0,0){};\draw[->](0,0)--(0.25,0);\end{tikzpicture}\, (\!A',G')$ are quantum homomorphisms $(H,P)\!:\!\!(A,G)\!\begin{tikzpicture}\node at (0,0){};\draw[->](0,0)--(0.25,0);\end{tikzpicture}(\!A',G')$;\looseness=-5
\item \textbf{2-morphisms} $(H,P)\!\begin{tikzpicture}\node at (0,0){};\draw[->](0,0)--(0.25,0);\end{tikzpicture}\, (\!H', P')$ are intertwiners of the underlying quantum functions (see Definition~\ref{def:intertwiner}).
\end{itemize}
\end{definition}
\noindent
As with $\QBij \subset \QSet$, we may define a subcategory $\QIso \subset \QGraph $ of quantum graphs, quantum graph isomorphisms and their intertwiners.

As with $\QSet$, there is a forgetful dagger $2$-functor $F:\QGraph \to \BHilb$ mapping each quantum homomorphism to its underlying Hilbert space and each intertwiner to its corresponding linear map. This $2$-functor makes $\QGraph$ a concrete dagger $2$-category $(\QGraph, F)$ (see Definition~\ref{def:Hilb2category}).

\begin{remark}\label{rem:last} Every quantum homomorphism between quantum graphs $(A,G)$ and $(B,H)$ is a quantum function between the underlying quantum sets of vertices $A$ and $B$. Since an intertwiner of quantum homomorphisms is exactly an intertwiner of the underlying quantum functions, it follows that the category $\QGraph((A,G),(B,H))$ is a full subcategory of the category $\QSet(A,B)$. In other words, there is a forgetful $2$-functor $\QGraph \to \QSet$ which is locally fully faithful. \end{remark}
\subsection{Dualisable morphisms in \QGraph }
We now show that a quantum isomorphism is nothing other than a dagger-dualisable (see Definition~\ref{def:daggerdualisable}) quantum homomorphism.
\begin{theorem}\label{thm:dualisablequantumgraph} A quantum homomorphism is dagger-dualisable in $\QGraph$ if and only if it is a quantum isomorphism. 
\end{theorem}

\begin{proof} Let $(A,G)$ and $(A',G')$ be quantum graphs and let $(H,P)$ be a dagger-dualisable quantum homomorphism between them. In particular, $(H,P)$ is a dagger-dualisable quantum function --- and therefore by Theorem~\ref{thm:dual2} a quantum bijection. Moreover, the dual $(H^*,\conj{P})$, defined in \eqref{eq:daggerdual}, is a quantum homomorphism from $(A',G')$ to $(A,G)$, i.e. it fulfills the following equation:
\def\d{2}
\begin{equation}\label{eq:dualisquantumhom}
\begin{tz}[zx,yscale=0.8,xscale=-1]
\draw[string] (0,-0.5) to node[front,zxnode=\zxwhite,pos=0.57] {$\conj{P}$} (0,5);
\draw[string] (\d,-0.5) to node[front,zxnode=\zxwhite, pos=0.43] {$\conj{P}$} (\d,5);
\draw[string] (0, 0.5) to (\d,0.5);
\node[zxvertex=\zxwhite] at (0,0.5){};
\node[zxvertex=\zxwhite] at (\d,0.5){};
\node[zxnode=\zxwhite] at (\d/2, 0.5 ) {$G'$};
\draw[string, arrow data={0.08}{<}, arrow data ={0.5}{<},arrow data ={0.95}{<}] (3, -0.5)  to (3,0.5) to [out=up, in=down] (-1.,4) to (-1.,5);
\end{tz}
~~~=~~~
\begin{tz}[zx,yscale=0.8,xscale=-1]
\draw[string] (0,-0.5) to node[front,zxnode=\zxwhite,pos=0.57] {$\conj{P}$} (0,5);
\draw[string] (\d,-0.5) to node[front,zxnode=\zxwhite, pos=0.43] {$\conj{P}$} (\d,5);
\draw[string] (0, 0.5) to (\d,0.5);
\node[zxvertex=\zxwhite] at (0,0.5){};
\node[zxvertex=\zxwhite] at (\d,0.5){};
\node[zxnode=\zxwhite] at (\d/2, 0.5 ) {$G'$};
\draw[string] (0, 4) to (\d,4);
\node[zxvertex=\zxwhite] at (0,4){};
\node[zxvertex=\zxwhite] at (\d,4){};
\node[zxnode=\zxwhite] at (\d/2, 4 ) {$G$};
\draw[string, arrow data={0.08}{<}, arrow data ={0.5}{<},arrow data ={0.95}{<}] (3, -0.5)  to (3,0.5) to [out=up, in=down] (-1.,4) to (-1.,5);
\end{tz}
\end{equation}
Recalling~\eqref{eq:daggerdual} and using \eqref{eq:propadjacency}, we rotate this equation by 180 degree to the left and obtain: 
\def\d{2}
\[
\begin{tz}[zx,yscale=0.8,xscale=-1]
\draw[string] (0,-0.5) to node[front,zxnode=\zxwhite,pos=0.57] {$P$} (0,5);\draw[string] (\d,-0.5) to node[front,zxnode=\zxwhite, pos=0.43] {$P$} (\d,5);
\draw[string] (0, 4) to (\d,4);
\node[zxvertex=\zxwhite] at (0,4){};
\node[zxvertex=\zxwhite] at (\d,4){};
\node[zxnode=\zxwhite] at (\d/2, 4 ) {$G'$};
\draw[string, arrow data={0.08}{>}, arrow data ={0.5}{>},arrow data ={0.95}{>}] (3, -0.5)  to (3,0.5) to [out=up, in=down] (-1.,4) to (-1.,5);
\end{tz}
~~~=~~~
\begin{tz}[zx,yscale=0.8,xscale=-1]
\draw[string] (0,-0.5) to node[front,zxnode=\zxwhite,pos=0.57] {$P$} (0,5);\draw[string] (\d,-0.5) to node[front,zxnode=\zxwhite, pos=0.43] {$P$} (\d,5);
\draw[string] (0, 0.5) to (\d,0.5);
\node[zxvertex=\zxwhite] at (0,0.5){};
\node[zxvertex=\zxwhite] at (\d,0.5){};
\node[zxnode=\zxwhite] at (\d/2, 0.5 ) {$G$};
\draw[string] (0, 4) to (\d,4);
\node[zxvertex=\zxwhite] at (0,4){};
\node[zxvertex=\zxwhite] at (\d,4){};
\node[zxnode=\zxwhite] at (\d/2, 4 ) {$G'$};
\draw[string, arrow data={0.08}{>}, arrow data ={0.5}{>},arrow data ={0.95}{>}] (3, -0.5)  to (3,0.5) to [out=up, in=down] (-1.,4) to (-1.,5);
\end{tz}
~~~\superequalseq{eq:quantumgraphhomomorphism}~~~
\begin{tz}[zx,yscale=0.8,xscale=-1]
\draw[string] (0,-0.5) to node[front,zxnode=\zxwhite,pos=0.57] {$P$} (0,5);\draw[string] (\d,-0.5) to node[front,zxnode=\zxwhite, pos=0.43] {$P$} (\d,5);
\draw[string] (0, 0.5) to (\d,0.5);
\node[zxvertex=\zxwhite] at (0,0.5){};
\node[zxvertex=\zxwhite] at (\d,0.5){};
\node[zxnode=\zxwhite] at (\d/2, 0.5 ) {$G$};
\draw[string, arrow data={0.08}{>},arrow data={0.5}{>}, arrow data ={0.95}{>}] (3, -0.5)  to (3,0.5) to [out=up, in=down] (-1.,4) to (-1.,5);
\end{tz}
\]
Contracting with the unit on the bottom left and the counit on the top right recovers condition~\eqref{eq:quantumgraphisomorphism}. 

For the converse, we note that a quantum isomorphism in the sense of Definition~\ref{def:quantumgraphiso} is in particular a quantum bijection and therefore commutes both with comultiplication~\eqref{eq:quantumfunction}, multiplication~\eqref{eq:monadmap} and the graphs~\eqref{eq:quantumgraphisomorphism}. Therefore, it is clearly a quantum homomorphism~\eqref{eq:quantumgraphhomomorphism} whose dual is also a quantum homomorphism~\eqref{eq:dualisquantumhom}.
\end{proof}

\begin{remark}
Following Remark~\ref{rem:monoidalsubcats}, we can think of the monoidal categories\\ $\QIso((A,G),(A,G))$ of quantum automorphisms of a quantum graph $(A,G)$ as a noncommutative or quantum version of the automorphism group of this quantum graph.\looseness=-2
\end{remark}

\begin{remark}\label{rem:Banicagraph}
Based on Wang's `quantum symmetry groups'~\cite{Wang1998} (see Section~\ref{sec:Wang}), Banica~\cite{Banica2005} introduced\footnote{An earlier, related but different definition is due to Bichon~\cite{Bichon2003}.} `quantum automorphism groups' of finite classical graphs as noncommutative variants of their automorphism groups. \ignore{These noncommutative algebras have been studied by Banica, Bichon and others~\cite{Banica2005,Banica2007_2,Banica2007_3,Bichon2003}.}
We now show that our quantum graph automorphism categories are the categories of finite-dimensional representations of the Hopf $C^*$-algebras corresponding to these quantum automorphism groups. This can again be understood as asserting that our quantum graph automorphisms are quantum elements of these quantum groups (cf. Section~\ref{sec:Wang} and Theorem~\ref{thm:universalprop}).
\end{remark}

\begin{prop}\label{prop:Banica} Let $(V_G,G)$ be a classical graph. The category $\QIso((V_G,G), (V_G,G))$ is the category of finite-dimensional representations of Banica's `quantum automorphism group' algebra $A_{aut}(G)$ of the graph $G$ (see e.g.~\cite[Definition 2.1]{Banica2007_2}).
\end{prop}
\begin{proof} The algebra $A_{aut}(G)$ is defined as the $C^*$-algebra with generators $a_{i,j}~(i,j=1,\ldots,n)$ and relations (cf. Theorem 2.1 and Example 2.2 in~\cite{Banica2005}):
\begin{calign}\nonumber
a_{i,j}^2=a_{i,j} = a_{i,j}^* 
&
\sum_{i=1}^n a_{i,j} = 1, ~~\forall 1\leq j\leq n
&
\sum_{j=1}^n a_{i,j} =1, ~~\forall 1\leq i\leq n
\end{calign}
\vspace{-0.35cm}
\begin{calign}\nonumber
\sum_{k=1}^n G_{i,k} ~a_{k,j} ~=~ \sum_{k=1}^n a_{i,k} ~G_{k,j}
\end{calign}
Here $n$ is the number of vertices of the graph $G$ and $\{G_{i,j}\}_{1\leq i,j \leq n}$ is the adjacency matrix of $G$.\looseness=-2

Similar to the proof of Proposition~\ref{prop:Wang}, comparing these relations with equations~\eqref{eq:PPM1},~\eqref{eq:PPM2} and~\eqref{eq:ppmquantumgraphcondition1} shows that $\QIso((V_G,G),(V_G,G))$ is the category of finite-dimensional representations of $A_{aut}(G)$. \ignore{These finite-dimensional representations are therefore all unitary, since their underlying quantum bijections are.}
\end{proof}

\section{The semisimple structure of quantum functions}\label{sec:semisimple}

\ignore{We have not yet explored the significance of the 2-morphisms (intertwiners) in \QSet\ and \QGraph .}
In this section, we investigate more closely the categories of quantum functions $\QSet(A,B)$ between quantum sets $A$ and $B$ and show that they are \emph{semisimple}. This sheds light on the role played by the intertwiners and is crucial to our understanding of the structure of quantum functions, particularly with a view towards distinguishing genuinely quantum functions from essentially classical ones.

Semisimplicity is a categorical structure which commonly arises in representation theory. The paradigmatic example of a semisimple category is the category of finite-dimensional representations of a finite group; in this category there are a number of \emph{irreducible} representations and every representation splits as a direct sum of these `simple objects'. As we will now show, categories of quantum functions are similar; every quantum function from $A$ to $B$ may be decomposed into a direct sum of irreducible quantum functions.

\ignore{In Section~\ref{sec:dirsumqfcts} we discuss what semisimplicity means in $\QSet$ in concrete terms. We state the main result of this section in Theorem~\ref{thm:semisimple} but for reasons of readability we defer the proof until Section~\ref{app:semisimplicity}. }

\ignore{n Section \ref{sec:dirsumqfcts} we show how semisimplicity gives us a full description of the intertwiners in \QSet\ and \QGraph , and allows us to distinguish quantum functions and graph homomorphisms which are `genuinely quantum' from those which are not. In Section \ref{sec:quantumautomorphismcategories}, we use semisimplicity to show how results derived in finite noncommutative geometry~ fit into our categorical setting.}

\subsection{The direct sum of quantum functions}
\label{sec:dirsumqfcts}

\begin{definition} The \textit{direct sum} of two quantum functions $(H,Q)$ and $(H',P)$ is defined as $(H\oplus H', Q\oplus P)$, where $Q\oplus P$ is the following direct sum of the underlying linear maps:\looseness=-2
\begin{calign}\label{eq:directsum}
\begin{tz}[zx,every to/.style={out=up, in=down},xscale=-1]
\draw  (0,1) to (2,3);
\draw[string,arrow data={0.9}{>},arrow data={0.2}{>}](2,1) to node[zxnode=\zxwhite, pos=0.5]{$Q\oplus P$}  (0,3);
\node[dimension, right] at (0,1) {$A$};
\node[dimension, right] at (0,3) {$H\oplus H'$};
\node[dimension, left] at (2,3) {$B$};
\node[dimension, left] at (2,1) {$H\oplus H'$};
\end{tz}
=
\begin{tz}[zx,every to/.style={out=up, in=down},xscale=-1]
\draw  (0,1) to (2,3);
\draw[string,arrow data={0.9}{>},arrow data={0.26}{>}](2,1) to node[zxnode=\zxwhite, pos=0.5]{$Q$}  (0,3);
\node[dimension, right] at (0,1) {$A$};
\node[dimension, right] at (0,3) {$H$};
\node[dimension, left] at (2,3) {$B$};
\node[dimension, left] at (2,1) {$H$};
\end{tz}
\oplus 
\begin{tz}[zx,every to/.style={out=up, in=down},xscale=-1]
\draw  (0,1) to (2,3);
\draw[string,arrow data={0.9}{>},arrow data={0.26}{>}](2,1) to node[zxnode=\zxwhite, pos=0.5]{$P$}  (0,3);
\node[dimension, right] at (0,1) {$A$};
\node[dimension,right] at (0,3) {$H'$};
\node[dimension,left] at (2,3) {$B$};
\node[dimension,left] at (2,1) {$H'$};
\end{tz}
\end{calign}
\end{definition}

\noindent
In particular, if $A$ and $B$ are sets, and $(H,P)$ and $(H',Q)$ are quantum functions from $A$ to $B$, then the direct sum of the corresponding matrices of projectors has underlying Hilbert space $H\oplus H'$ and projectors:
\begin{equation}\left(P \oplus Q\right)_{a,b} = P_{a,b} \oplus Q_{a,b}\end{equation}
We now remark on the operational interpretation of the direct sum of quantum \mbox{functions.}\looseness=-2
\begin{remark}\label{rem:opintofdirsum}
In Remark \ref{rem:opintofquantumfunctions}, we observed that a quantum function between classical sets is a projective measurement controlled by an input element, whose result determines an output element. Operationally, a quantum function is a direct sum precisely when the following procedures coincide:
\begin{itemize}
\item Perform the quantum measurement $(H \oplus H'\!, [P \oplus Q]_{a,-})$ depending on a received input $a$.\looseness=-2
\item Before receiving any input, perform a projective measurement onto $H$ and $H'$, and then, depending on the outcome, perform the quantum measurement $(H,P_{a,-})$ or $(H',Q_{a,-})$  upon receiving the input $a$.
\end{itemize}
In the setting of nonlocal games~\cite{Mancinska2016,Atserias2016}, this can be understood as corresponding to a probabilistic mixture of quantum strategies.
\end{remark}
\noindent On the other hand, a quantum function $P$ is simple if it cannot be decomposed in this way; or equivalently, if there are no non-trivial intertwiners $P \to P$. 
\begin{definition}\label{def:simpleobject} A quantum function $P$ is \textit{simple} if $\QSet(P,P) \cong \mathbb{C}$.
\end{definition}
\noindent
 For a quantum function $P$ between classical sets, simplicity translates into the following condition on the corresponding matrix of projectors: 
 \begin{calign}\label{eq:simplePPM}\forall X \in \End(H): \hspace{1cm}P_{a,b}X = XP_{a,b} ~~\forall a\in A, b\in B \hspace{0.2cm}\Rightarrow\hspace{0.2cm} X\propto \mathbbm{1}_{H}\end{calign}
\noindent
The category of quantum functions $\QSet(A,B)$ is completely determined by the simple quantum functions, as the following theorem shows.
\begin{theorem}\label{thm:semisimple} For quantum sets $A$ and $B$, $\QSet(A,B)$ is a semisimple dagger category.
\end{theorem}
\begin{proof} We defer the proof of this theorem to Section~\ref{app:semisimplicity}. It is completely independent of the following Sections~\ref{sec:dirsumqfcts} and~\ref{sec:fusionstructure}.
\end{proof}
\noindent
The definition of semisimplicity is given in Section~\ref{app:semisimplicity}; it implies that every quantum function is a finite direct sum of simple quantum functions, which cannot be decomposed any further. \ignore{split in the manner of Remark~\ref{rem:opintofdirsum}.} In particular, every classical function is a simple quantum function (cf. Remark~\ref{rem:setsubset}).

\begin{remark} In contrast to many semisimple categories considered in the \mbox{literature}, most of the semisimple categories we consider have an infinite number of \mbox{simple objects.}\looseness=-2
\end{remark}
\ignore{Since quantum elements are quantum functions from the trivial algebra, the notion of simplicity applies to them also. We now consider simple quantum elements. }
\noindent
The only exceptions are the categories of quantum elements (and of course the categories of quantum functions to a one-element set which have a unique simple object).
\begin{proposition} Every quantum set has a finite number of simple quantum elements. 
\end{proposition}
\begin{proof} The category $\QEl(A) = \QSet(*,A)$ is the category of comodules of $A$ and as such the opposite of the category of modules of $A$. Since every symmetric special Frobenius algebra is semisimple (in the algebraic sense) there is only a finite number of inequivalent irreducible representations --- these are the simple quantum elements.
\end{proof}
\begin{proposition}\label{prop:simpleelsofclasssetareordinaryels} \!\!Simple quantum elements of a classical set are ordinary elements.
\end{proposition}
\begin{proof}A classical set is a commutative $C^*$-algebra. In particular, all its irreducible representations are one-dimensional, i.e. classical elements.
\end{proof}
\noindent
There is a close connection between semisimplicity and classicality of quantum functions, which we capture by the following definition. 

\begin{definition}\label{def:classical} A quantum function $Q:X\to Y$ between classical sets $X$ and $Y$ is \textit{essentially classical} if it is a direct sum of classical functions, i.e. of one-dimensional quantum functions. A quantum element of a classical set is \emph{essentially classical} if it is essentially classical as a quantum function from the one-element set.
\end{definition}
\noindent
In other words, an essentially classical quantum function is of the following form, where $\{\ket{i}\}$ is an orthonormal basis corresponding to the decomposition of the Hilbert space $H$ into one-dimensional subspaces $H\cong \bigoplus_{i} \mathbb{C} \ket{i}$ and $f_i:X\to Y$ are classical functions:%
\begin{equation}\label{eq:classicalquantumfunction}\begin{tz}[zx,every to/.style={out=up, in=down},xscale=-1]
\draw (0,0) to (2,3);
\draw[arrow data ={0.2}{>}, arrow data={0.8}{>}] (2,0) to (0,3);
\node[zxnode=\zxwhite] at (1,1.5) {$P$};
\node[dimension,left] at (0,0) {$X$};
\node[dimension,right] at (2,3) {$Y$};
\node[dimension,right] at (2,0) {$H$};
\node[dimension,left] at (0,3){$H$};
\end{tz}
~=~\sum_i~\begin{tz}[zx,every to/.style={out=up, in=down},xscale=-1]
\draw (0,0) to (2,3);
\draw[arrow data ={0.8}{>}] (0,2.) to (0,3);
\draw[arrow data ={0.4}{>}] (2,0) to (2,1);
\node[zxnode=\zxwhite] at (1,1.5) {$f_i$};
\node[zxnode=\zxwhite] at (0,2.) {$i$};
\node[zxnode=\zxwhite] at (2,1) {$i^\dagger$};
\end{tz}
\end{equation}
\noindent
The following proposition justifies Definition~\ref{def:classical}.
\begin{proposition} \label{prop:classicalcommute}A quantum function $P:X\to Y$ between classical sets $X$ and $Y$ is essentially classical if and only if all corresponding projectors $P_{xy}$ commute with each other.
\end{proposition}
\begin{proof} Commuting projectors are simultaneously diagonalizable; all the projectors therefore split as a direct sum over the shared eigenspaces. Conversely, every essentially classical quantum function decomposes into a direct sum of classical functions; all projectors $P_{xy}$ are therefore sums of projectors onto the orthonormal basis vectors corresponding to this decomposition.
\end{proof}

\begin{remark} In the terminology of Definition~\ref{def:classical}, every projective measurement --- thought of as a quantum function from the one-element set to its outcome set --- is essentially classical. 
Indeed, non-classical behaviour only becomes apparent if there are at least two projective measurements; that is, when the quantum function goes from a set with two or more elements to an outcome set. Only in these cases does it make sense to distinguish between essentially classical quantum functions, which correspond to measurements of commuting observables, and non-classical quantum functions.
From the nonlocal game perspective~\cite{Mancinska2016}, essentially classical quantum graph homomorphisms are mixtures of classical homomorphisms that only use the available quantum ressource as a generator of shared classical randomness.
\end{remark}

\ignore{These considerations provide an operational meaning to previous definitions of classical and quantum Latin squares. }
\begingroup
\renewcommand*{\arraystretch}{1.5}
\begin{example} \label{ex:latinsquaresexample}
A Latin square is a square $n{\times}n$-grid of numbers drawn from the set $\{0,\ldots,{n-1}\}$ such that each number appears exactly once in each row and each column. Every Latin square gives  a projective permutation matrix, such as the following:

\begin{equation}\label{eq:LS}\begin{pmatrix}\ketbra{0} & \ketbra{1} & \ketbra{2}\\ \ketbra{1}&\ketbra{2} & \ketbra{0}\\ \ketbra{2}&\ketbra{0}&\ketbra{1}\end{pmatrix}
\end{equation}
Under the decomposition $\mathbb{C}^3 = \mathbb{C} \ket{0} \oplus \mathbb{C} \ket{1} \oplus \mathbb{C} \ket{2}$, this is a sum of permutations: 
\begin{equation}\begin{pmatrix}\ketbra{0} & \ketbra{1} & \ketbra{2}\\ \ketbra{1}&\ketbra{2} & \ketbra{0}\\ \ketbra{2}&\ketbra{0}&\ketbra{1}\end{pmatrix}
 = \begin{pmatrix} 1& 0 & 0 \\0 &0 & 1\\ 0 & 1 & 0\end{pmatrix} \oplus \begin{pmatrix} 0& 1 & 0 \\1 &0 & 0\\ 0 & 0 & 1\end{pmatrix} \oplus \begin{pmatrix} 0& 0 & 1 \\0 &1 & 0\\ 1 & 0 & 0\end{pmatrix}  \end{equation}
\noindent
Indeed, every Latin square can always be decomposed into a sum of permutations. On the other hand, consider the following \textit{quantum Latin square}~\cite{Musto2015}:
\begin{equation}\label{eq:QLS}\hspace{-0.4cm}\begin{pmatrix}\ketbra{0} & \ketbra{1} & \ketbra{2} & \ketbra{3}\\ \ketbra{\psi_-}&\ketbra{\phi_+} & \ketbra{\phi_-} & \ketbra{\psi_+} \\ \ketbra{\psi_+} & \ketbra{\phi_-} & \ketbra{\phi_+} & \ketbra{\psi_-} \\ \ketbra{3} & \ketbra{2} & \ketbra{1} & \ketbra{0}
\end{pmatrix}\hspace{0.75cm}\text{ where }\hspace{0.75cm}\begin{array}{l}\ket{\psi_+} = \frac{1}{\sqrt{2}}\left(\ket{1} + \ket{2}\right) \\ \ket{\psi_-} =\frac{1}{\sqrt{2}}\left(\ket{1} - \ket{2}\right)\\ \ket{\phi_+} =\frac{1}{\sqrt{5}}\left( i \ket{0} + 2 \ket{3}\right) \\\ket{\phi_-} =\frac{1}{\sqrt{5}}\left( 2 \ket{0} + i \ket{3}\right)\end{array}\end{equation}
It can be verified using~\eqref{eq:simplePPM} that this is simple and cannot be written as a direct sum of other projective permutation matrices.

In other words, the concept of simplicity distinguishes truly quantum structures, like the quantum Latin square~\eqref{eq:QLS}, from `classical' structures modelled in Hilbert spaces, like the Latin square~\eqref{eq:LS}.  
\end{example}
\endgroup

\noindent
In this terminology, every Latin square and all quantum elements of a classical set are essentially classical, while the quantum Latin square in~\eqref{eq:QLS} and the quantum elements of certain quantum sets are not. In fact, we can prove the following simple lemma.
\begin{proposition} Quantum sets whose quantum elements are all essentially classical are classical sets.
\end{proposition}
\begin{proof} All simple quantum elements of such a quantum set are one-dimensional; it follows that the corresponding algebra is commutative, i.e. an ordinary set.
\end{proof}
\noindent
The categories of quantum graph homomorphisms $\QGraph((A,G),(A',G'))$ are also semisimple (see Corollary~\ref{cor:semisimplegraph}). The following proposition shows that a quantum graph homomorphism is simple precisely when its underlying quantum function is.

\begin{proposition}\label{prop:graphhomsimpleifffunctionsimple} A quantum graph homomorphism $Q:(A,G)\to (A',G')$ is simple in $\QGraph((A,G),(A',G'))$ if and only if its underlying quantum function $Q:A\to A'$ is simple in $\QSet(A,A')$. 
Moreover, suppose that $Q:(A,G)\to (A',G')$ is a quantum homomorphism and that $Q$ has a decomposition $Q\cong \bigoplus_i f_i$ into simple quantum functions $f_i$ in $\QSet(A,A')$. Then, each $f_i$ is a quantum graph homomorphism $(A,G)\to (A,G')$ and the decomposition is a decomposition in $\QGraph((A,G),(A',G'))$.
\end{proposition}
\begin{proof}
The proposition follows from the fact that $\QGraph((A,G),(A',G')) \subseteq \QSet(A,A')$ is a full inclusion of semisimple categories (see Remark~\ref{rem:last}).
\end{proof}
\ignore{
\begin{corollary}\label{prop:graphhomsimpleifffunctionsimple}
Let  $Q:G\to G'$ be a quantum graph homomorphism in $\QGraph(G,G')$ such that $Q$ has a decomposition $Q\cong \bigoplus_i f_i$ into simple quantum functions $f_i$ in $\QSet(V_G,V_{G'})$. Then, each $f_i$ is a quantum graph homomorphism $G\to G'$ and the decomposition is also a decomposition in $\QGraph(G,G')$.
\end{corollary}}
\ignore{\begin{proof}
It follows straightforwardly from Definition \ref{def:simpleobject} that the simple objects of the subcategory will be the same as the simple objects of the full category.
\end{proof}}

\noindent
Therefore, the question of semisimplicity of a quantum graph homomorphism is simply a question of semisimplicity of the underlying quantum function and does not depend on the graphs. Despite this result, finding the simple objects in the category of quantum homomorphisms between two quantum graphs is quite involved.

\subsection{The fusion structure of quantum bijections}
\label{sec:fusionstructure}
The category $\QSet(A,A)$ of quantum functions from a quantum set $A$ to itself possesses various other structures which interact with the semisimple structure. In particular, under composition of quantum functions, it becomes a \emph{monoidal} semisimple dagger category. If we restrict further to quantum bijections $\QBij(A,A)$, then by the results of Section~\ref{sec:quantumbijection} we obtain a monoidal semisimple dagger category with dualisable objects\footnote{It is clear that the tensor unit of $\QBij(A,A)$, i.e. the quantum bijection $\mathrm{id}_A:A\otimes \mathbb{C} \to \mathbb{C} \otimes A$, is simple. Therefore, $\QBij(A,A)$ is a tensor $*$-category as defined in~\cite[Definition 1.33]{Mueger2007}.} (see Theorem~\ref{thm:dual2}). For a finite number of simple objects such a structure is known as a \textit{unitary fusion category}; such structures generalise finite groups, and have been the subject of considerable research in recent years~\cite{Etingof2015}.

As noted in Section~\ref{sec:Wang}, these quantum permutation categories have been considered in finite noncommutative topology. We restate some of the results of this work now.\looseness=-2 
\begin{proposition}[{\cite[Theorem 6.2.]{Banica2007}}]All quantum bijections on an $n$-element classical sets with $n\leq 3$ are essentially classical. 
\end{proposition}
\begin{proof} By Proposition~\ref{prop:classicalcommute}, we need to show that all projectors in a projective permutation matrix of size $1,2$ and $3$ commute. For $n=1$, this is trivial. For $n=2$, this follows from the fact that every projective permutation matrix is of the form \[\left(\begin{array}{cc} p& 1-p\\ 1-p & p \end{array} \right)\] for some projector $p$. The proof for $n=3$ is rather involved and can be found in~\cite{Banica2007}.
\end{proof}
\noindent
Note that there are quantum bijections between 4-element classical sets that are not essentially classical; these include the quantum Latin square~\eqref{eq:QLS} of dimension 4 in Example \ref{ex:latinsquaresexample}, and the following projective permutation matrix where $p$ and $q$ are non-commuting projectors on a 2-dimensional Hilbert space:
\begin{equation}\begin{pmatrix}p&1-p&0&0\\1-p&p&0&0\\0&0&q&1-q\\0&0&1-q&q\end{pmatrix}\end{equation}
In fact, the general structure of the category of quantum bijections between 4-element sets has been investigated by Banica and Bichon~\cite{Banica2009}, leading in particular to the following surprising result regarding the dimension of quantum bijections (recall Definition~\ref{def:qfctdimension}).
\begin{proposition}[{\cite[Proposition 3.1.]{Banica2009}}]
Every simple quantum bijection between 4-element sets must have dimension 1,2 or 4.
\end{proposition}
\noindent
Similarly, the monoidal semisimple dagger categories of quantum graph automorphisms $\QIso(G,G)$ have been studied~\cite{Banica2005,Banica2007_2,Banica2007_3,Bichon2003}, and several graphs with only essentially classical quantum automorphisms have been identified.

\begin{remark} Essentially classical quantum bijections (in the sense of Definition~\ref{def:classical}) on the $n$-element set $[n]$ form a full fusion subcategory of $\QBij([n], [n])$ equivalent to the category of $S_n$-graded vector spaces~\cite{Etingof2015}, where $S_n$ is the symmetric group (also cf. Remark~\ref{rem:setsubset}). Similarly, essentially classical quantum graph automorphisms on a classical graph $G$ form a full fusion subcategory of $\QIso(G,G)$ equivalent to the category of $\Aut(G)$-graded vector spaces, where $\Aut(G)$ is the automorphism group of $G$.
\end{remark}
\subsection{Proof of semisimplicity}
\label{app:semisimplicity}
In the following section, we show that the hom-categories $\QSet(A,B)$ for quantum sets $A$ and $B$ are semisimple. Recall that a linear category (that is, a category enriched in vector spaces over some field) is semisimple~\cite{Mueger2003}, if it has direct sums, if all idempotents split and if there are objects $X_i$ (called \emph{simple objects}) labeled by an indexing set $I$ such that $\Hom(X_i,X_j) \cong \delta_{i,j} \mathbb{C}$ and such that every object $V$ is isomorphic to a finite direct sum of simple objects. For dagger categories such as those in this paper, this notion of semisimplicity is subsumed by more fundamental properties.

\begin{definition}A \emph{linear dagger category} is a dagger category enriched in finite-dimensional complex vector spaces such that for objects $X$ and $Y$ the induced function $\dagger: \Hom(X,Y) \to \Hom(Y,X)$ is antilinear. A linear dagger category is \emph{positive}, if for every morphism $A\to[r]B$:\begin{equation}r^\dagger r= 0 \hspace{1cm}\Rightarrow \hspace{1cm}r=0\end{equation}
\end{definition}

\begin{proposition} \label{thm:daggerpositive}For two quantum sets $A$ and $B$, the category $\QSet(A,B)$ is a positive linear dagger category.
\end{proposition}
\begin{proof} By Theorem~\ref{thm:QSetdagger2category}, $\QSet(A,B)$ is a dagger category.\! Linearity and positivity follow from the existence of the forgetful dagger functor \mbox{$F:\QSet(A,B) \to \Hilb$.}
\end{proof}
\noindent
Following M\"uger, Roberts and Tuset, we define semisimplicity of such a dagger category as follows~\cite[Definition 2.2]{Mueger2004}.
\begin{definition} A \emph{semisimple dagger category} is a linear dagger category for which the following hold:
\begin{enumerate}
\item It is positive.
\item It has a zero object; that is an object such that every object has a unique morphism into and out of it.
\item It has binary direct sums (cf.~\cite{Vicary2011}, there called dagger biproducts): For objects $X_1$ and $X_2$, there is an object $X_1\oplus X_2$ and morphisms $s_i:X_i \to X_1\oplus X_2$ ($i=1,2$) such that the following hold:
\begin{calign} s_i^\dagger s_i = \mathrm{id}_{X_i}~\text{ for }i=1,2&
s_1s_1^\dagger+s_2s_2^\dagger=\mathrm{id}_{X_1\oplus X_2}\end{calign}
\item All dagger idempotents split (cf.~\cite{Selinger2008}): A dagger idempotent is an endomorphism $p\in \End(X)$ on some object $X$ such that $p^2=p =p^\dagger$. It splits if there is a morphism $i: A\to X$ out of some object $A$ such that the following hold:
\begin{calign} i^\dagger i = \mathrm{id}_A & i i^\dagger= p
\end{calign}
\end{enumerate}
\end{definition}
\noindent
An object $X$ in a semisimple dagger category is \emph{simple}, if $\Hom(X,X)\cong \mathbb{C}$. 

It is well known that semisimple dagger categories are semisimple in the usual sense (for a proof for monoidal semisimple dagger categories, see e.g.~\cite[Lemma 1.35]{Mueger2007}). For completeness, we sketch a version of this argument.

\begin{prop} In a semisimple dagger category, non-isomorphic simple objects $X$ and $Y$ are disjoint ($\Hom(X,Y) = \{0\}$) and every object is isomorphic to a finite direct sum of simple objects.
\end{prop}
\begin{proof}
It is a straightforward consequence of positivity that non-isomorphic simple objects are disjoint. Moreover, given an object $X$, then the dagger makes $\Hom(X,X)$ into a positive $*$-algebra. Since every positive $*$-algebra is semisimple, it follows that if $X$ is not simple, there is a non-trivial dagger idempotent in $\Hom(X,X)$ which we can split to obtain an isometry $i:A\to X$ and use this to decompose $X\cong A \oplus A^\bot$. Continuing inductively, every non-simple object can be decomposed into a finite sum of simple objects.
\end{proof}
\noindent We are now ready to restate and prove Theorem~\ref{thm:semisimple}.\vspace{10pt}

\noindent \textbf{Theorem~\ref{thm:semisimple}.}
\textit{For quantum sets $A$ and $B$, $\QSet(A,B)$ is a semisimple dagger category.}
\begin{proof}
1. Positivity of $\QSet(A,B)$ is proven in Proposition~\ref{thm:daggerpositive}.

\noindent
2. The zero object is the quantum function $(\mathbf{0}, 0)$, where $\mathbf{0}$ is the zero-dimensional Hilbert space and $0$ is the unique morphism $0:  \mathbf{0} \otimes A \to B \otimes \mathbf{0}$.

\noindent
3. The direct sum of two quantum functions is defined in \eqref{eq:directsum} and inherits all its structural properties from $\Hilb$.

\noindent
4. A dagger idempotent in $\QSet(A,B)$ is a self-adjoint idempotent endomorphism $r: H \to H$ on a quantum function $Q:H\otimes A \to B \otimes H$:%
\begin{calign}\label{eq:idempotentsplit}
\begin{tz}[zx, scale=1.7,every to/.style={out=up, in=down},xscale=-1]
\draw (0,0) to (2,2);
\draw[arrow data={0.95}{>},arrow data={0.07}{>}] (2,0) to node[zxnode=\zxwhite,pos=0.5] {$Q$} node[zxnode=\zxwhite, pos=0.2] {$r$} (0,2);
\end{tz}
\quad = \quad
\begin{tz}[zx, scale=1.7,every to/.style={out=up, in=down},xscale=-1]
\draw (0,0) to (2,2);
\draw[arrow data={0.07}{>},arrow data ={0.95}{>}] (2,0) to node[zxnode=\zxwhite,pos=0.5] {$Q$} node[zxnode=\zxwhite, pos=0.8] {$r$} (0,2);
\end{tz}
\end{calign}
In particular, we can split this idempotent in \Hilb, obtaining an isometry $i:V\to H$ such that $ii^\dagger = r$.
It then follows from \eqref{eq:idempotentsplit} that %
\begin{calign}\nonumber
\begin{tz}[zx, scale=1.7,every to/.style={out=up, in=down},xscale=-1]
\draw (0,0) to (2,2);
\draw[arrow data={0.95}{>},arrow data={0.085}{>}] (2,0) to node[zxnode=\zxwhite,pos=0.5] {$Q'$}  (0,2);
\node[dimension,right] at (2,0) {$V$};
\node[dimension,left] at (0,2) {$V$};
\end{tz}
\quad := \quad
\begin{tz}[zx, scale=1.7,every to/.style={out=up, in=down},xscale=-1]
\draw (0,0) to (2,2);
\draw[arrow data={0.95}{>},arrow data={0.085}{>}] (2,0) to node[zxnode=\zxwhite,pos=0.5] {$Q$} node[zxnode=\zxwhite, pos=0.2] {$i$}node[zxnode=\zxwhite, pos=0.8] {$i^\dagger$} (0,2);
\node[dimension,right] at (2,0) {$V$};
\node[dimension,left] at (0,2) {$V$};
\end{tz}
\end{calign}
is itself a quantum function and the isometry $i:V\to H$ is an intertwiner of quantum functions $ Q' \to Q$. Therefore, $r$ splits as an idempotent in $\QSet(A,B)$.
\end{proof}

\noindent
Essentially the same proof works for the 2-category $\QGraph$.
\begin{corollary} \label{cor:semisimplegraph}For quantum graphs $(A,G)$ and $(B,H)$, the category $\QGraph((A,G),(B,H))$ is semisimple.
\end{corollary}

\section{Quantum graphs and quantum relations}
\label{app:quantumrelation}
In this final section, we discuss the relational approach to quantum graphs taken by Kuperberg, Weaver, Duan, Severini, Winter, and others~\cite{Kuperberg2012, Weaver2010,Weaver2015,Duan2013, Ortiz2016} and compare it to our Definition~\ref{def:quantumgraphsbyadjmats}. We show that our quantum graphs can indeed be understood as symmetric and reflexive quantum relations as defined by Kuperberg and Weaver~\cite{Kuperberg2012,Weaver2015} and therefore generalise the noncommutative graphs of Duan, Severini and Winter~\cite{Duan2013}. 

In this section, to fit with the work of other authors, all graphs and quantum graphs will be reflexive (see Definition~\ref{def:quantumgraphsbyadjmats}); wherever `graph' is written it should be taken to mean `reflexive graph'.

The relational approach starts from the observation that a classical graph $G$ may be described as a reflexive and symmetric relation on its vertex set $V_{G}$; that is, a subset $S\subseteq V_G \times V_{G}$ containing both the diagonal $\Delta=\{(x,x)~|~x\in V_G\}$ and such that $(x,y)\in S\Leftrightarrow (y,x) \in S$ for all $x,y\in V_{G}$.

Quantising this description leads to a notion of quantum graph equivalent to Definition~\ref{def:quantumgraphsbyadjmats} which is closer in spirit to previous definitions~\cite{Duan2013, Weaver2010, Weaver2015}. We first quantise relations, recovering the finite-dimensional quantum relations of Kuperberg and Weaver~\cite{Weaver2010,Kuperberg2012}. 

\begin{proposition} Under Gelfand duality, a binary relation between two finite sets $X$ and $Y$ can be expressed as a projector $\!P\!:X\! \otimes Y \!\!\to \!X \otimes Y$ satisfying the following equation:\looseness=-2
\begin{equation}\label{eq:bimodulerelation}
\def\scl{1.25}
\begin{tz}[zx,every to/.style={out=up, in=down},scale=\scl]
\draw (0,0) to (0,2);
\draw (1,0) to (1,2);
\draw (-0.5,0) to[in=-135] (0,0.6);
\draw (1.5,0) to [in=-45] (1,0.6);
\boxwidth{1.25}
\node[widebox] at (0.5,1.25) {$P$};
\node[zxvertex=\zxwhite] at (0,0.6){};
\node[zxvertex=\zxwhite] at (1,0.6){};
\node[dimension,left] at (0,-0.1) {$X$};
\node[dimension, left] at (-0.5,-0.1) {$X$};
\node[dimension,right] at (1.5,-0.1) {$Y$};
\node[dimension, right] at (1,-0.1) {$Y$};
\node[dimension, left] at (0., 2) {$X$};
\node[dimension, right] at (1, 2) {$Y$};
\end{tz} 
~~=~~
\begin{tz}[zx,every to/.style={out=up, in=down},scale=\scl]
\draw (0,0) to (0,2);
\draw (1,0) to (1,2);
\draw (-0.5,0) to[in=-150] (0,1.4);
\draw (1.5,0) to [in=-30] (1,1.4);
\boxwidth{1.25}
\node[widebox] at (0.5,0.75) {$P$};
\node[zxvertex=\zxwhite] at (0,1.4){};
\node[zxvertex=\zxwhite] at (1,1.4){};
\node[dimension,left] at (0,-0.1) {$X$};
\node[dimension, left] at (-0.5,-0.1) {$X$};
\node[dimension,right] at (1.5,-0.1) {$Y$};
\node[dimension, right] at (1,-0.1) {$Y$};
\node[dimension, left] at (0., 2) {$X$};
\node[dimension, right] at (1, 2) {$Y$};
\end{tz} 
\end{equation}
\end{proposition}
\begin{proof}
Projectors on $X\otimes Y$ encode subspaces of $X\otimes Y$; the bimodule condition~\eqref{eq:bimodulerelation} guarantees that the corresponding subspace is the linear span of a \emph{subset} of $X\times Y$.
\ignore{
Under Gelfand duality, a subset $S \subseteq X \times Y$ is a subspace of $X \otimes Y$ spanned by the product basis states $\ket{x} \otimes \ket{y}$, where $x \sim_S y$. The equation \eqref{eq:bimodulemap} precisely enforces that the subspace in the image of the projector is spanned by product basis states.}
\end{proof}
\noindent We define quantum relations analogously.
\begin{definition}\label{def:quantumrelations}
A \emph{quantum relation} between quantum sets $A$ and $B$ is a projector\\ ${P: A \otimes B \to A \otimes B}$ satisfying equation \eqref{eq:bimodulerelation}. A \emph{quantum relation on a quantum set $A$} is  a quantum relation between $A$ and itself.
\end{definition}

\begin{remark} \label{rem:quantumrelationprojection}Let $A$ and $B$ be finite-dimensional $C^*$-algebras. Given a projector\\ $P\in \End(A\otimes B)$ fulfilling equation~\eqref{eq:bimodulerelation}, we obtain a projection\footnote{A projection $p$ in a $C^*$-algebra $A$ is an element $p\in A$ such that $p^*=p=p^2$.} $p:=P(e_A\otimes e_B)\in A^{op}\otimes B$; conversely, any projection $p\in A^{op}\otimes B$ gives rise to a projector $P:=L_p \in \End(A\otimes B)$ fulfilling~\eqref{eq:bimodulerelation}. (Here, $A^{op}$ denotes the opposite algebra of $A$, $e_A\in A,~e_B\in B$ denote the units of $A$ and $B$, and $L_p$ denotes left multiplication with $p$ in the algebra $A^{op}\otimes B$.) A quantum relation between $A$ and $B$ may therefore be equivalently defined as a projection in the algebra $A^{op} \otimes B$. Plugging units into the second and third slot of equation~\eqref{eq:bimodulerelation} gives a simple diagrammatic proof of this fact.

\ignore{A similar simple definition of quantum relation was briefly considered by Weaver in the introduction of~\cite{Weaver2010} (although without using the opposite algebra) but dismissed because of difficulties expressing properties of such relations, such as symmetry or reflexivity. In the following, we will use the Frobenius algebra structure on a finite-dimensional $C^*$-algebra to express these properties and show that our definition of quantum relation coincides with the finite-dimensional quantum relations of Kuperberg and Weaver~\cite{Weaver2010,Kuperberg2012}.}
\end{remark}
Kuperberg and Weaver~\cite{Weaver2010,Kuperberg2012} define a quantum relation on a von Neumann algebra $\mathcal{M}\subseteq \mathcal{B}(H)$ to be a weak$^*$-closed operator bimodule over the commutant $\mathcal{M}'$ (see Definition 2.1 in~\cite{Weaver2010}). In the finite-dimensional case (see Definition 5.1 in~\cite{Weaver2015}), this definition reduces to the following:

\begin{definition}\label{KuperbergWeaverQreldef}Let $A$ be a finite-dimensional $C^*$-algebra, let $H$ be some finite-dimensional Hilbert space such that $A\subseteq \End(H)$ and let $A'=\{b \in \End(H)~|~ ba=ab~\forall a\in A\}$ be the commutant with respect to this embedding. A \emph{quantum relation in the sense of Kuperberg and Weaver} is a subspace $V\subseteq \End(H)$ which fulfills $A' V A'\subseteq V$.\end{definition}
\noindent
It is shown in Theorem 2.7 of~\cite{Weaver2010} that this definition is independent\footnote{More precisely, it can be shown that if $H$ and $H'$ are finite-dimensional Hilbert spaces such that there are embeddings $A\subseteq \End(H)$ and $A\subseteq \End(H')$, then there is a canonical correspondence between the correspondingly defined quantum relations.} of the embedding $A\subseteq \End(H)$. We now show that, in the finite-dimensional case, this definition is equivalent to our own.
\begin{proposition}\label{prop:WeaverKuperberg}
Given a finite-dimensional $C^*$-algebra $A$, there is a Hilbert space $H$ and an embedding $A\subseteq \End(H)$ such that our notion of quantum relation on $A$ (Definition~\ref{def:quantumrelations}) coincides with that of Kuperberg and Weaver (Definition \ref{KuperbergWeaverQreldef}).
\end{proposition}
\begin{proof} 
As noted in Remark~\ref{rem:quantumrelationprojection}, our quantum relations on a finite-dimensional $C^*$-algebra $A$ can be expressed as projections in $A^{op}\otimes A$; Weaver proves in~\cite[Proposition~2.23]{Weaver2010} that quantum relations in the sense of Kuperberg-Weaver on finite-dimensional $C^*$-algebras correspond precisely to such projections.

For concreteness, we will now provide an alternative proof using the Frobenius structure on the $C^*$-algebra.
We have seen in Section~\ref{sec:diagramsforC*} that a finite-dimensional $C^*$-algebra $A$ admits the structure of a symmetric special dagger Frobenius algebra. In particular, $A$ admits the structure of a Hilbert space, and there is an embedding of $C^*$-algebras $A\subseteq \End(A)$ given as follows: \begin{calign}\nonumber
A\to \End(A) & a\mapsto L_a (:b\in A \mapsto ab)\end{calign}
We now show that the corresponding commutant $A'$ is $*$-isomorphic to the opposite algebra $A^{op}$, defined on the vector space $A$ with multiplication $a\star b:= ba$. The following linear map is clearly a faithful $*$-homomorphism:
\begin{calign}\nonumber 
A^{op} \to A' & a\mapsto R_a(:b\in A \mapsto ba)
\end{calign}
We now show that it is also surjective and therefore a $*$-isomorphism. Let $X:A\to A$ be an element of $A'$, and denote the unit of $A$ by $e\in A$. Then $X(e)\in A$ and we will show that $R_{X(e)}=X$. For all $b\in A$ the following holds:
\[R_{X(e)} (b) = b X(e) = L_b X(e) ~~\superequals{$X\in A'$}~~ X L_b(e) = X(b) \]
A quantum relation in the sense of Kuperberg and Weaver is therefore a subspace $\mathcal{V} \subseteq \End(A)$ such that $R_a\mathcal{V} R_b \subseteq \mathcal{V}$ for all $a,b\in A$.
Since $A$ is a symmetric Frobenius algebra, there is a canonical isomorphism $A\cong A^*$ which induces an isomorphism $\End(A) \cong A \otimes A$.  
Under this isomorphism, left and right composition with $R_a$ in $\End(A)$ can be identified with left and right multiplication with $(a\otimes e) $ and $(e\otimes a)$ in $A\otimes A$, respectively:
\begin{calign}\nonumber R_a \circ - : \End(A) \to \End(A) & -\circ R_b: \End(A) \to \End(A)\\\nonumber
(a\otimes e) \cdot - : A\otimes A\to A\otimes A & -\cdot (e\otimes b) : A\otimes A \to A\otimes A
\end{calign}
Therefore, a quantum relation in the sense of Kuperberg and Weaver can be understood as a subspace $\mathcal{V} \subseteq A\otimes A $ such that $(a\otimes e)\mathcal{V}(e\otimes b)\subseteq \mathcal{V}$ for all $a,b\in A$.

This coincides with our Definition~\ref{def:quantumrelations}.
\end{proof}
\noindent
Following Weaver~\cite{Weaver2015}, we now consider additional properties of quantum relations, such as symmetry and reflexivity.

\begin{definition}\label{def:quantumgraphsbyrelations} A quantum relation $P:A\otimes A\to A \otimes A$ is \emph{symmetric} or \emph{reflexive} if one of the following additional equations holds:\
\def\scl{1.25}
\def\w{0.9375}
\begin{calign}\label{eq:symrefrelation}
\begin{tz}[zx,xscale=0.75,yscale=3/3.5,scale=\scl]
\draw (0,0.5) to (0,1.) to [out=up, in=up,looseness=1.75] node[zxvertex=\zxwhite, pos=0.5] {} (2.5,1.) to (2.5,-1);
\draw (1,0.5) to (1,1.) to [out=up, in=up,looseness=3] node[zxvertex=\zxwhite, pos=0.5]{}(1.5,1.) to (1.5,-1);
\draw (-0.5,2.5) to (-0.5,0.5) to [out=down, in=down, looseness=3] node[zxvertex=\zxwhite, pos=0.5]{} (0,0.5);
\draw (-1.5,2.5) to (-1.5,0.5) to [out=down, in=down, looseness=1.75]node[zxvertex=\zxwhite, pos=0.5]{} (1,0.5);
\boxwidth{\w}
\node[widebox] at (0.5,0.75) {$P$};
\end{tz}
~~~=~~~
\begin{tz}[zx,xscale=0.75,scale=\scl]
\draw (0,0) to (0, 3);
\draw (1,0) to (1,3);
\boxwidth{\w}
\node[widebox] at (0.5, 1.5){$P$};
\end{tz}
&
\begin{tz}[zx,xscale=0.75,scale=\scl]
\clip (-0.3,0) rectangle (1.3,3);
\draw (0,3) to (0,1) to [out=down, in=down] node[zxvertex=\zxwhite, pos=0.5]{} (1,1) to (1,3);
\boxwidth{\w}
\node[widebox] at (0.5,2){$P$};
\end{tz}
~~~=~~~
\begin{tz}[zx,xscale=0.75,scale=\scl]
\clip (-0.2,0) rectangle (1.2,3);
\draw (0,3) to (0,2) to [out=down, in=down]node[zxvertex=\zxwhite, pos=0.5]{} (1,2) to (1,3);
\end{tz}\\ \nonumber
\text{a) symmetric} 
& 
\text{b) reflexive}
\end{calign}
\end{definition}
\noindent
We now show that quantum graphs as in Definition~\ref{def:quantumgraphsbyadjmats} are indeed symmetric and reflexive quantum relations. For a quantum graph $(A,G)$ as in Definition~\ref{def:quantumgraphsbyadjmats}, we introduce the following linear map $P_G:A\otimes A \to A \otimes A$:
\begin{equation}\label{eq:graphprojector}
\begin{tz}[zx]
\draw[string] (0, 0) to (0,3);
\draw (1.5,0) to (1.5, 3);
\boxwidth{1.5}
\node[widebox] at (0.75, 1.5 ) {$P_G$};
\end{tz}
~~~=~~~
\begin{tz}[zx]
\draw[string] (0, 0) to (0,3);
\draw (2.5,0) to (2.5, 3);
\draw (0, 1.5) to (2.5, 1.5) ;
\node[zxvertex=\zxwhite] at (0,1.5){};
\node[zxvertex=\zxwhite] at (2.5,1.5){};
\node[zxnode=\zxwhite] at (1.25, 1.5 ) {$G$};
\end{tz}
~~~:=~~~
\begin{tz}[zx,xscale=2.5/3]
\draw (-1.5,0) to (-1.5,1.5) to [out=up, in=-135] (-0.75, 2.5) to (-0.75,3);
\draw (-0.75, 2.5) to [out=-45, in=up] (0, 1.5) to [out=down, in=135] (0.75,0.5) to (0.75,0);
\draw (0.75,0.5) to [out=45, in=down] (1.5,1.5) to (1.5,3);
\node[zxnode=\zxwhite] at (0, 1.5) {$G$};
\node[zxvertex=\zxwhite, zxdown] at (-0.75, 2.5) {};
\node[zxvertex=\zxwhite, zxup] at (0.75, 0.5) {};
\end{tz}
~~~\superequalseq{eq:propadjacency}~~~
\begin{tz}[zx,xscale=-2.5/3]
\draw (-1.5,0) to (-1.5,1.5) to [out=up, in=-135] (-0.75, 2.5) to (-0.75,3);
\draw (-0.75, 2.5) to [out=-45, in=up] (0, 1.5) to [out=down, in=135] (0.75,0.5) to (0.75,0);
\draw (0.75,0.5) to [out=45, in=down] (1.5,1.5) to (1.5,3);
\node[zxnode=\zxwhite] at (0, 1.5) {$G$};
\node[zxvertex=\zxwhite, zxdown] at (-0.75, 2.5) {};
\node[zxvertex=\zxwhite, zxup] at (0.75, 0.5) {};
\end{tz}
\end{equation}
For a classical graph $(V_G,G)$, it is easily verified that this map~\eqref{eq:graphprojector} is the projector onto the subspace of the symmetric and reflexive relation defining the graph:
\begin{equation}\label{eq:projectorsubset}\left\{ (v,w) ~|~ v \sim_G w \right\} \subseteq V_G \times V_G\end{equation}
Conversely, the adjacency matrix $G$ can be recovered from this projector as follows:
\def\scl{1.33}
\begin{equation}\label{eq:recoveradjacency}
\begin{tz}[zx,scale=\scl]
\draw (0,0) to (0,2);
\node[zxnode=\zxwhite] at (0,1) {$G$};
\end{tz}
~~~:=~~~
\begin{tz}[zx,scale=\scl]
\draw (0,0) to (0,1) to [out=up, in=up,looseness=2.5] node[zxvertex=\zxwhite, pos=0.5] {} (-0.5,1) to (-0.5,-0.25);
\draw (0.75,0) to (0.75,1.75);
\boxwidth{1}
\node[widebox] at (0.375,0.65){$P_G$};
\node[zxvertex=\zxwhite] at (0,0){};
\node[zxvertex=\zxwhite] at (0.75,0){};
\end{tz}
~~~\superequalseq{eq:bimodulerelation}~~~
\begin{tz}[zx,scale=\scl]
\draw (0,-0.25) to (0,1.3);
\draw (0.75,0) to (0.75,1.75);
\boxwidth{1}
\node[widebox] at (0.375,0.65){$P_G$};
\node[zxvertex=\zxwhite] at (0,1.3){};
\node[zxvertex=\zxwhite] at (0.75,0){};
\end{tz}
~~~\superequals{$(-)^\dagger$}~~~
\begin{tz}[zx,xscale=-1,scale=\scl]
\draw (0,-0.25) to (0,1.3);
\draw (0.75,0) to (0.75,1.75);
\boxwidth{1}
\node[widebox] at (0.375,0.65){$P_G$};
\node[zxvertex=\zxwhite] at (0,1.3){};
\node[zxvertex=\zxwhite] at (0.75,0){};
\end{tz}
~~~\superequalseq{eq:bimodulerelation}~~~
\begin{tz}[zx,xscale=-1,scale=\scl]
\draw (0,0) to (0,1) to [out=up, in=up,looseness=2.5] node[zxvertex=\zxwhite, pos=0.5] {} (-0.5,1) to (-0.5,-0.25);
\draw (0.75,0) to (0.75,1.75);
\boxwidth{1}
\node[widebox] at (0.375,0.65){$P_G$};
\node[zxvertex=\zxwhite] at (0,0){};
\node[zxvertex=\zxwhite] at (0.75,0){};
\end{tz}
\end{equation}
The same correspondence holds for general quantum graphs.

\begin{theorem} \label{thm:relationadjacency}Given a quantum graph $(A,G)$ as in Definition~\ref{def:quantumgraphsbyadjmats}, the projector \\{$P_G:A\otimes A\to A \otimes A$} defined in equation~\eqref{eq:graphprojector} is a symmetric and reflexive quantum relation as in Definition~\ref{def:quantumrelations}. Conversely, given a symmetric and reflexive quantum relation $P$ on $A$, the map~\eqref{eq:recoveradjacency} defines a quantum adjacency matrix. These two constructions are mutually inverse.
\end{theorem}
\begin{proof} 
Given an arbitrary linear map $G:A\to A$, we define the following linear map $P_G:A\otimes A\to A \otimes A$ fulfilling~\eqref{eq:bimodulerelation}:%
\[\begin{tz}[zx]
\draw[string] (0, 0) to (0,3);
\draw (1.5,0) to (1.5, 3);
\boxwidth{1.5}
\node[widebox] at (0.75, 1.5 ) {$P_G$};
\end{tz}
~~~:=~~~
\begin{tz}[zx,xscale=2.5/3]
\draw (-1.5,0) to (-1.5,1.5) to [out=up, in=-135] (-0.75, 2.5) to (-0.75,3);
\draw (-0.75, 2.5) to [out=-45, in=up] (0, 1.5) to [out=down, in=135] (0.75,0.5) to (0.75,0);
\draw (0.75,0.5) to [out=45, in=down] (1.5,1.5) to (1.5,3);
\node[zxnode=\zxwhite] at (0, 1.5) {$G$};
\node[zxvertex=\zxwhite, zxdown] at (-0.75, 2.5) {};
\node[zxvertex=\zxwhite, zxup] at (0.75, 0.5) {};
\end{tz}\]
Conversely, given a linear map $P_G:A\otimes A \to A\otimes A$ fulfilling~\eqref{eq:bimodulerelation}, we define 
\def\scl{1.33}
\[\begin{tz}[zx,scale=\scl]
\draw (0,0) to (0,2);
\node[zxnode=\zxwhite] at (0,1) {$G$};
\end{tz}
~~~:=~~~
\begin{tz}[zx,xscale=-1,scale=\scl]
\draw (0,-0.25) to (0,1.3);
\draw (0.75,0) to (0.75,1.75);
\boxwidth{1}
\node[widebox] at (0.375,0.65){$P_G$};
\node[zxvertex=\zxwhite] at (0,1.3){};
\node[zxvertex=\zxwhite] at (0.75,0){};
\end{tz}
\]
It follows easily from equation~\eqref{eq:bimodulerelation} that these two constructions are mutually inverse. In the following, we say that $G$ is \emph{real} if the following holds:
\[
\begin{tz}[zx]
\draw (0,0) to (0,3);
\node[zxnode=\zxwhite] at (0,1.5) {$G^\dagger$};
\end{tz}
~~~=~~~
\begin{tz}[zx,xscale=-1]
\draw (1,0) to (1,1.5) to [out=up, in=up, looseness=2.5]node[zxvertex=\zxwhite, pos=0.5] {} (0,1.5) to [out=down, in=down, looseness=2.5]node[zxvertex=\zxwhite, pos=0.5] {} (-1,1.5) to (-1,3);
\node[zxnode=\zxwhite] at (0,1.5) {$G$};
\end{tz}
\]
Simple graphical arguments then establish the following:
\begin{itemize}
\item $G$ is real if and only if $P_G$ is self-adjoint.
\item $G$ fulfills the first equation of~\eqref{eq:propadjacency} if and only if $P_G^2=P_G$.
\item $G$ fulfills the second equation of~\eqref{eq:propadjacency} if and only if $P_G$ is symmetric.
\item $G$ fulfills the last equation of~\eqref{eq:propadjacency} if and only if $P_G$ is reflexive.\qedhere
\end{itemize}
\end{proof}
\begin{remark} We observed in Remark~\ref{rem:quantumrelationprojection} that a quantum relation on $A$ can be understood as a projection $p$ in $A^{op}\otimes A $. In terms of the adjacency matrix $G$, the projection corresponding to a quantum graph $(A,G)$ is the following element $p$ of $A\otimes A$:
\begin{equation}
\begin{tz}[zx]
\clip (-0.2, -0.8) rectangle (1.2,1.5);
\draw (0,0) to (0,1.5);
\draw (1,0) to (1,1.5);
\boxwidth{1}
\node[widebox] at (0.5,0){$p$};
\end{tz}
~~~=~~~
\begin{tz}[zx]
\clip (-0.2, -0.8) rectangle (1.2,1.5);
\draw (1,1.5) to (1,0) to [out=down, in=down, looseness=2] node[zxvertex=\zxwhite, pos=0.5]{} (0,0) to (0,1.5);
\boxwidth{1}
\node[zxnode=\zxwhite] at (1,0.5){$G$};
\end{tz}
\end{equation}
In particular, a quantum relation on $A$ given by a projection $p\in A^{op}\otimes A$ is symmetric if $\sigma(p) = p$, where $\sigma:A\otimes A \to A\otimes A$ is the swap map, and reflexive if $m(p)=e_A$, where $m:A\otimes A \to A$ is the multiplication of $A$ and $e_A$ is the unit of $A$.
\end{remark}

\begin{remark}A related notion of quantum graph was introduced and studied by Duan, Severini and Winter~\cite{Duan2013} as a noncommutative quantum version of the \emph{confusability graphs} in classical zero-error communication. These noncommutative graphs are defined as operator systems on finite-dimensional Hilbert spaces. We now show that such quantum graphs correspond to our quantum graphs on matrix algebras.
 \end{remark}
 \begin{definition}\label{def:operatorsystem} Let $H$ be a finite-dimensional Hilbert space. An \emph{operator system} on $H$ is a subspace $\mathcal{V} \subseteq \End(H)$ such that the following hold:
 \begin{calign}
 X\in \mathcal{V} ~~\Rightarrow ~~X^\dagger \in \mathcal{V} 
 &
 \mathbbm{1}_H \in \mathcal{V}
 \end{calign}
 \end{definition}
 \noindent
We prove that every operator system on $H$ gives rise to a quantum graph $(\End(H), G)$ on the endomorphism algebra $\End(H)$ and vice versa.
 
 \begin{proposition}\label{prop:quantumgraphDuanthesame}Let $H$ be a finite-dimensional Hilbert space. There is a canonical correspondence between quantum graphs $(\End(H),G)$ on the endomorphism algebra of $H$ and operator systems on $H$.
 \end{proposition}
\begin{proof}
It is shown in~\cite{Weaver2015} that symmetric and reflexive quantum relations (in the sense of Kuperberg and Weaver) on $\End(H)$ coincide with operator systems on $H$. The proposition therefore follows from Proposition~\ref{prop:WeaverKuperberg} and Theorem~\ref{thm:relationadjacency}, in which we prove that our notion of quantum graph coincides with Kuperberg and Weaver's symmetric and reflexive quantum relations.
For concreteness, we explicitly construct the correspondence between quantum graphs on endomorphism algebras and operator systems.

An operator system $\mathcal{V}\subseteq \End(H)$ can be described in terms of the projector $P_{\mathcal{V}}:H^*\otimes H \to H^*\otimes H$ onto the subspace $\mathcal{V}\subseteq \End(H)\cong H^* \otimes H$. This projector fulfills the following equations:%
\def\scl{1.25}%
\def\w{0.9375}%
\begin{calign}\label{eq:operatorsystemgraphics}
\begin{tz}[zx,xscale=0.75,yscale=3/3.5,scale=\scl]
\draw[arrow data={0.165}{<},arrow data={0.91}{<}] (0,0.25) to (0,1.25) to [out=up, in=up,looseness=1.75]  (2.5,1.25) to (2.5,-1);
\draw[arrow data={0.26}{>}, arrow data={0.86}{>}] (1,0.25) to (1,1.25) to [out=up, in=up,looseness=3] (1.5,1.25) to (1.5,-1);
\draw[arrow data={0.25}{<}, arrow data={0.99}{<}] (-0.5,2.5) to (-0.5,0.25) to [out=down, in=down, looseness=3] (0,0.25);
\draw[arrow data={0.14}{>}, arrow data={0.99}{>}] (-1.5,2.5) to (-1.5,0.25) to [out=down, in=down, looseness=1.75] (1,0.25);
\boxwidth{\w}
\node[widebox] at (0.5,0.75) {$P_{\mathcal{V}}$};
\end{tz}
~~~=~~~
\begin{tz}[zx,xscale=0.75,scale=\scl]
\draw[arrow data={0.2}{<}, arrow data={0.8}{<}] (0,0) to (0, 3);
\draw[arrow data={0.2}{>}, arrow data={0.8}{>}] (1,0) to (1,3);
\boxwidth{\w}
\node[widebox] at (0.5, 1.5){$P_{\mathcal{V}}$};
\end{tz}
&
\begin{tz}[zx,xscale=0.75,scale=\scl]
\clip (-0.3,0) rectangle (1.3,3);
\draw[arrow data={0.1}{>}, arrow data={0.325}{>}, arrow data={0.695}{>}, arrow data={0.92}{>}] (0,3) to (0,1) to [out=down, in=down]  (1,1) to (1,3);
\boxwidth{\w}
\node[widebox] at (0.5,2){$P_{\mathcal{V}}$};
\end{tz}
~~~=~~~
\begin{tz}[zx,xscale=0.75,scale=\scl]
\clip (-0.2,0) rectangle (1.2,3);
\draw[arrow data={0.17}{>}, arrow data={0.86}{>}] (0,3) to (0,2) to [out=down, in=down] (1,2) to (1,3);
\end{tz}
\end{calign}
The first equation corresponds to the property that $X\in \mathcal{V}\Rightarrow X^\dagger \in \mathcal{V}$, the second equation encodes that $\mathbbm{1}_H \in \mathcal{V}$.

Let $(\End(H), G)$ be a quantum graph, expressed as a symmetric, reflexive quantum relation $P_G: \End(H) \otimes \End(H) \to \End(H) \otimes \End(H)$ (see Theorem~\ref{thm:relationadjacency} and Definition~\ref{def:quantumgraphsbyrelations}).
The bimodule equation~\eqref{eq:bimodulerelation} can now be expressed as the following first equation; the second equation is obtained from contracting the second wire at the bottom with the third and the sixth with the seventh (here we identify $\End(H) \cong H \otimes H^*$):%
\begin{equation}\label{eq:bimodulematrix}
\begin{tz}[zx]
\draw[arrow data={0.15}{<}, arrow data={0.9}{<}] (0,-1.5) to +(0,3);
\draw[arrow data={0.16}{>}, arrow data={0.91}{>}] (1,-1.5) to +(0,3);
\draw[arrow data={0.15}{>}, arrow data={0.91}{>}] (-1.2,-1.5) to [out=up, in=down] (-0.4,0) to (-0.4,1.5);
\draw[arrow data={0.8}{<}] (-0.8, -1.5) to [out=up, in=up, looseness=3] (-0.4,-1.5);
\draw [arrow data={0.15}{<}, arrow data={0.91}{<}](2.2, -1.5) to[out=up, in=down] (1.4,0) to (1.4,1.5);
\draw[arrow data={0.8}{<}] (1.4, -1.5) to [out=up, in=up, looseness=3]  (1.8,-1.5);
\boxwidth{1.8}
\node[widebox] at (0.5, 0.5) {$P_G$};
\end{tz}
~~~=~~~
\begin{tz}[zx]
\draw[arrow data={0.15}{<}, arrow data={0.9}{<}] (0,-1.5) to +(0,3);
\draw[arrow data={0.16}{>}, arrow data={0.91}{>}]  (1,-1.5) to +(0,3);
\draw[arrow data={0.15}{>}, arrow data={0.91}{>}] (-1.2,-1.5) to (-1.2,0) to [out=up, in=down] (-0.4,1.5);
\draw[arrow data={0.12}{<}, arrow data={0.62}{<}, arrow data={0.91}{<}] (-0.8, -1.5) to (-0.8, 0) to  [out=up, in=up, looseness=3] (-0.4,0) to (-0.4, -1.5);
\draw[arrow data={0.15}{<}, arrow data={0.91}{<}] (2.2, -1.5) to (2.2,0) to [out=up, in=down] (1.4,1.5);
\draw[arrow data={0.12}{<}, arrow data={0.41}{<}, arrow data={0.91}{<}]  (1.4, -1.5) to (1.4, 0) to  [out=up, in=up, looseness=3] (1.8, 0) to (1.8, -1.5);
\boxwidth{1.8}
\node[widebox] at (0.5, -0.5) {$P_G$};
\end{tz}
\hspace{1.25cm} \Rightarrow \hspace{1.25cm}
\dim(H)^2~~
\begin{tz}[zx]
\draw[arrow data={0.15}{<}, arrow data={0.85}{<}] (0,-1.5) to (0,1.5);
\draw[arrow data={0.15}{>}, arrow data={0.85}{>}] (1,-1.5) to +(0,3);
\draw[arrow data={0.15}{>}, arrow data={0.85}{>}] (-0.4,-1.5) to + (0,3);
\draw[arrow data={0.15}{<}, arrow data={0.85}{<}] (1.4,-1.5) to + (0,3);
\boxwidth{1.8}
\node[widebox] at (0.5, 0){$P_G$};
\end{tz}
~~~=~~~
\begin{tz}[zx]
\draw[arrow data={0.15}{>}, arrow data={0.85}{>}] (-1.4, -1.5) to + (0,3);
\draw[arrow data={0.15}{<}, arrow data={0.85}{<}] (2.4,-1.5) to + (0,3);
\draw[arrow data={0.15}{<}, arrow data={0.85}{<}] (0,-1.5) to (0,1.5);
\draw[arrow data={0.15}{>}, arrow data={0.85}{>}] (1,-1.5) to +(0,3);
\draw[arrow data={0.18}{>}, arrow data={0.88}{>}] (-0.4,0) to [out=up, in=up, looseness=5] (-1,0) to [out=down, in=down, looseness=5]  (-0.4,0);
\draw[arrow data={0.18}{<}, arrow data={0.88}{<}] (1.4,0) to [out=up, in=up, looseness=5] (2,0) to [out=down, in=down, looseness=5]  (1.4,0);
\boxwidth{1.8}
\node[widebox] at (0.5, 0){$P_G$};
\end{tz}
\end{equation}
Therefore, quantum relations $P_G: \End(H)\otimes \End(H) \to \End(H) \otimes \End(H)$ can be interconverted into projectors $P_\mathcal{V}:H^*\otimes H \to H^*\otimes H$ as follows:%
\begin{calign}\nonumber
\begin{tz}[zx]
\draw[arrow data={0.15}{<}, arrow data={0.85}{<}] (0,-1.5) to (0,1.5);
\draw[arrow data={0.15}{>}, arrow data={0.85}{>}] (1,-1.5) to +(0,3);
\boxwidth{1}
\node[widebox] at (0.5, 0){$P_{\mathcal{V}}$};
\end{tz}
~~~=
~\frac{1}{\dim(H)^2}~~
\begin{tz}[zx]
\draw[arrow data={0.15}{<}, arrow data={0.85}{<}] (0,-1.5) to (0,1.5);
\draw[arrow data={0.15}{>}, arrow data={0.85}{>}] (1,-1.5) to +(0,3);
\draw[arrow data={0.18}{>}, arrow data={0.88}{>}] (-0.4,0) to [out=up, in=up, looseness=5] (-1,0) to [out=down, in=down, looseness=5]  (-0.4,0);
\draw[arrow data={0.18}{<}, arrow data={0.88}{<}] (1.4,0) to [out=up, in=up, looseness=5] (2,0) to [out=down, in=down, looseness=5]  (1.4,0);
\boxwidth{1.8}
\node[widebox] at (0.5, 0){$P_G$};
\end{tz}
&
\begin{tz}[zx]
\draw[arrow data={0.15}{<}, arrow data={0.85}{<}] (0,-1.5) to (0,1.5);
\draw[arrow data={0.15}{>}, arrow data={0.85}{>}] (1,-1.5) to +(0,3);
\draw[arrow data={0.15}{>}, arrow data={0.85}{>}] (-0.4,-1.5) to + (0,3);
\draw[arrow data={0.15}{<}, arrow data={0.85}{<}] (1.4,-1.5) to + (0,3);
\boxwidth{1.8}
\node[widebox] at (0.5, 0){$P_G$};
\end{tz}
~~~=~~~
\begin{tz}[zx]
\draw[arrow data={0.15}{<}, arrow data={0.85}{<}] (0,-1.5) to (0,1.5);
\draw[arrow data={0.15}{>}, arrow data={0.85}{>}] (1,-1.5) to +(0,3);
\draw[arrow data={0.15}{>}, arrow data={0.85}{>}] (-0.4,-1.5) to + (0,3);
\draw[arrow data={0.15}{<}, arrow data={0.85}{<}] (1.4,-1.5) to + (0,3);
\boxwidth{1}
\node[widebox] at (0.5, 0){$P_{\mathcal{V}}$};
\end{tz}
\end{calign}
It follows from equation~\eqref{eq:bimodulematrix} that these constructions are mutually inverse. Under this correspondence, the map $P_G$ is a projector if and only if $P_{\mathcal{V}}$ is a projector, $P_G$ is a symmetric quantum relation if and only if $P_\mathcal{V}$ fulfills the first equation of~\eqref{eq:operatorsystemgraphics}, and $P_G$ is reflexive if and only if $P_{\mathcal{V}}$ fulfills the second equation of~\eqref{eq:operatorsystemgraphics}.
\end{proof}

\bibliographystyle{plainurl}
\bibliography{QSET}

\end{document}